%% file: main.tex
\documentclass[11pt]{article}

\def\focs{0}

\usepackage{fullpage,graphicx}
\usepackage{amsmath,amsfonts,amsthm,amssymb,multirow,mathtools}
\usepackage{algorithm}
\usepackage{algpseudocode}
\usepackage{comment,url,enumerate}
\usepackage{tikz}
\usepackage{framed}
\usepackage{amsthm}
\usepackage{tcolorbox}
\usepackage{enumitem}
\usepackage{caption}
\usepackage{microtype}
\usetikzlibrary{decorations.pathreplacing, shapes}
\usepackage[backref, colorlinks,citecolor=blue,linkcolor=magenta,bookmarks=true]{hyperref}  
\usepackage[nameinlink]{cleveref}
\definecolor{shadecolor}{gray}{0.95}

\usepackage{dialogue}

\usepackage{color-edits}
\addauthor{ah}{red}

\usepackage{hyperref}
\usepackage{liyang}

\crefname{figure}{algorithm}{algorithms}

\renewcommand{\phi}{\varphi}

\newcommand{\cD}{{\mathcal{D}}}
\newcommand{\cU}{{\mathcal{U}}}

\newcommand{\Ber}{\mathrm{Ber}}

\newcommand{\error}{\mathrm{error}}

\newcommand{\balpha}{\boldsymbol{\alpha}}

\newcommand{\sm}{\mathrm{small}}
\newcommand{\lar}{\mathrm{large}}

\newcommand{\MAJ}{{\textsc{Maj}}}
\newcommand{\Maj}{\MAJ}
\newcommand{\lift}{\mathrm{Lift}}

\DeclarePairedDelimiter\abs{|}{|}
\DeclarePairedDelimiter\ltwo{\|}{\|_2}
\DeclarePairedDelimiter\lone{\|}{\|_1}

\DeclarePairedDelimiter\paren{(}{)}
\DeclarePairedDelimiter\bracket{[}{]}
\DeclarePairedDelimiter\set{\{}{\}}

\def\colorful{0}

\ifnum\colorful=1

\fi
\ifnum\colorful=0

\fi

\newcommand{\pparagraph}[1]{\bigskip \noindent {\bf {#1}}}

\begin{document}

\input{intro}
\input{ProofOverview}
\input{Preliminaries}
\input{Juntas}

\input{Lifting}
\input{BoostingProofs}

\section*{Acknowledgments}

% 

% \acks{

We thank the FOCS reviewers for their helpful feedback and suggestions. The authors are supported by NSF awards 1942123, 2211237, 2224246, a Sloan Research Fellowship, and a Google Research Scholar Award. Caleb is also supported by an NDSEG Fellowship. Guy is also supported by a Jane Street Graduate Research Fellowship.

\bibliographystyle{alpha}
\bibliography{references}

\newpage

\appendix

\input{Boosting}

\end{document}

%% file: intro.tex
\ifnum\focs=0

\title{The sample complexity of smooth boosting and \\  the tightness of the hardcore theorem \vspace{10pt}}

\author{ 
Guy Blanc \vspace{6pt} \\ 
\hspace{-10pt} {{\sl Stanford}} \and 
Alexandre Hayderi \vspace{6pt} \\
 {{\sl Stanford}} \and 
Caleb Koch \vspace{6pt} \\ 
\hspace{-8pt} {{\sl Stanford}} \vspace{10pt} \and 
Li-Yang Tan \vspace{6pt}  \\
\hspace{-15pt} {{\sl Stanford}}
}

\date{\small{\today}}

\maketitle

\begin{abstract}
{\sl Smooth} boosters generate distributions that do not place too much weight on any given example. Originally introduced for their noise-tolerant properties, such boosters have also found applications in differential privacy, reproducibility, and quantum learning theory. We study and settle the sample complexity of smooth boosting: we exhibit a class that can be weak learned to $\gamma$-advantage over smooth distributions with $m$ samples, for which strong learning over the uniform distribution requires $\tilde{\Omega}(1/\gamma^2)\cdot m$ samples. This matches the overhead of existing smooth boosters and provides the first separation from the setting of distribution-independent boosting, for which the corresponding overhead is $O(1/\gamma)$. 

Our work also sheds new light on Impagliazzo’s hardcore theorem from complexity theory, all known proofs of which can be cast in the framework of smooth boosting. For a function~$f$ that is mildly hard against size-$s$ circuits, the hardcore theorem provides a set of inputs on which $f$ is extremely hard against size-$s'$ circuits. A downside of this important
result is the loss in circuit size, i.e.~that $s' \ll s$. Answering a question of Trevisan, we show that this size loss is necessary and in fact, the parameters achieved by known proofs are the best possible. 

%Our technical contributions include a new lifting theorem from junta complexity to average-case circuit complexity as well as a new notion of {\sl soft} junta complexity.  
\end{abstract}

\else
\fi

\ifnum\focs=1
{}
\else
 \thispagestyle{empty}
 \newpage 

 \thispagestyle{empty}
 \setcounter{tocdepth}{2}
 \tableofcontents
 \thispagestyle{empty}
 \newpage 
 \setcounter{page}{1}

\fi

\section{Introduction}

Boosting is a technique for generically improving the accuracy of learning algorithms. A boosting algorithm makes multiple calls to a {\sl weak} learner---one with accuracy that is slightly better than trivial---and aggregates the predictions of the weak hypotheses into a single high-accuracy prediction. Since its conception in the 1990s~\cite{KV89,Sha90,Fre92,Fre95}, boosting has become a central topic of study within learning theory, with entire textbooks devoted to it~\cite{FS12}, and it has also had a substantial impact on practice. 

While the story of boosting is one of success, the framework comes with two important downsides. The first is the need for a {\sl distribution-independent} weak learner. Even if the goal is to learn with respect to a fixed and known distribution $\mathcal{D}$, the weak learner has to succeed with respect to all distributions. This is because boosting works by calling the weak learner on a sequence of distributions $\mathcal{D}_1,\ldots,\mathcal{D}_T$, and there are a priori no guarantees as to how similar these $\mathcal{D}_i$’s are to~$\mathcal{D}$. The second issue is that of {\sl noise tolerance}, a challenge already highlighted in Shapire’s original paper~\cite{Sha90}. One would naturally like to convert weak learners into strong ones even in the presence of noise, but popular boosting algorithms such as AdaBoost~\cite{FS97} have long been known to perform poorly in this regard~\cite{Sch99,Die00}. 

\paragraph{Smooth boosting.} Both issues are addressed by {\sl smooth boosting}. First explored in~\cite{Fre95,Jac97,DW00,KS03} and then formalized by Servedio~\cite{Ser03}, a smooth booster is one that only generates {\sl smooth distributions}, distributions that do not place too much weight on any example. {Smooth boosters therefore only require weak learners for smooth distributions rather than fully distribution-independent ones. Additionally,} smoothness is a natural desideratum from the perspective of noise tolerance. Indeed, the poor noise tolerance of AdaBoost has been attributed to its {\sl non}-smoothness~\cite{Die00}: AdaBoost can generate skewed distributions that place a lot of weight on a few examples, which intuitively, would hurt its performance if these examples were noisy. Following~\cite{Ser03}, smooth boosters have been designed in a variety of noise models~\cite{Gav03,KK09,Fel10,BCS20,DIKLST21}.

Beyond noise tolerance, smoothness is also the key property enabling the design of boosters that are differentially private~\cite{DGV10,BCS20}, reproducible~\cite{ILPS22}, and amenable to quantum speedups~\cite{IdW23}. Furthermore, there are classes for which weak learners are only known for smooth distributions but not all distributions. A notable example is the class of DNF formulas~\cite{BFJKMR94}, and indeed smooth boosting was crucially leveraged in Jackson’s celebrated polynomial-time algorithm for strong learning DNFs~\cite{Jac97}.% and in subsequent optimizations of the algorithm~\cite{KS03,BHK09,Fel10}.

\section{This work}

Given the importance of smooth boosting, it is of interest to understand fundamental properties of the framework. Prior work has focused on two such properties, round complexity and the tradeoff between smoothness and error~\cite{Ser03, KS03, Hol05, BHK09}. In this work we consider yet another basic property, sample complexity.

\subsection{First result: The sample complexity of smooth boosting}

Our goal is to understand the sample complexity {\sl overhead} incurred by smooth boosting. Existing smooth boosters convert an $m$-sample $\gamma$-advantage weak learner into a strong learner with sample complexity $O(1/\gamma^2)\cdot m$.
%\footnote{This follows implicitly from the round complexity of~\cite{Ser03}'s smooth booster, as observed by~\cite{ILPS22}.} 
Can this overhead of $O(1/\gamma^2)$ be improved? What if we allow for less time-efficient or completely time-inefficient algorithms? A simple argument, which we give in~\Cref{sec:lower-bound-boosting-overhead} (\Cref{claim:simple-gap-smooth}), shows a lower bound of $\Omega(1/\gamma)$, leaving a quadratic gap.

Our first result closes this gap up to logarithmic factors:\medskip

\begin{tcolorbox}[colback = white,arc=1mm, boxrule=0.25mm]
\begin{theorem}
\label{thm:smooth intro} For any sample size $m$ and parameter $\gamma$, there exists a concept class $\mcC$ such that:
\begin{enumerate}
    \item \textbf{Weak learning $\mcC$ requires few samples}: There exists a weak learner that, given random examples generated by any smooth distribution $\mcD$, uses $m$ samples and w.h.p.~outputs a hypothesis with accuracy $\frac1{2} + \gamma$.
    \item \textbf{Strong learning $\mcC$ requires many samples}: Any algorithm that, given random examples generated according to uniform distribution and w.h.p.~outputs a hypothesis with accuracy at least $0.99$, requires at least $\Tilde{\Omega}(m/\gamma^2)$ samples.
\end{enumerate}
\end{theorem}
\end{tcolorbox}\medskip 

We remark that the upper bound is realized by a time-efficient algorithm whereas the lower bound applies to all learners, even time-inefficient ones. 

\paragraph{Separating the sample complexities of smooth and distribution-independent boosting.} \Cref{thm:smooth intro} highlights a fundamental difference between smooth and distribution-independent boosting. For distribution-independent boosting, one also has an $\Omega(1/\gamma)$ lower bound on the sample complexity overhead (also by~\Cref{claim:simple-gap-smooth}), but this is matched by a $O(1/\gamma)$ upper bound:

\begin{fact}
    \label{fact:dist-free upper bound}
    Let $\mcC$ be a concept class and let $\gamma>0$. 
    If the sample complexity of weak learning $\mathcal{C}$ to accuracy $\frac1{2} + \gamma$ in the distribution-independent setting is $m$, then the sample complexity of learning $\mathcal{C}$ to accuracy $0.99$ in the distribution-independent setting is $O(m/\gamma)$.
\end{fact}

The proof of~\Cref{fact:dist-free upper bound} is simple and follows from basic VC theory; see \Cref{claim:ub-sample-cxty-strong}. 

\begin{remark}[A computational-statistical gap for distribution-independent boosting?] 
The upper bound of~\Cref{fact:dist-free upper bound} is realized by a time-inefficient algorithm.  Existing time-efficient distribution-independent boosters do {\sl not}  match it---they incur an overhead of $O(1/\gamma^2)$. This raises the question of whether there exist time-efficient distribution-independent boosters achieving the optimal sample complexity overhead of $O(1/\gamma)$. A negative answer would be especially interesting as it would show that distribution-independent boosting exhibits a {\sl computational-statistical gap}. %(Since this amounts to proving a super-polynomial lower bound on the time complexity of algorithms, any proof will likely have to rely on complexity-theoretic or cryptographic assumptions.) 

\Cref{thm:smooth intro} together with existing time-efficient smooth boosters, on the other hand, shows that there is {\sl no} computational-statistical gap in the smooth setting. 
\end{remark}

\subsection{Second result: Tightness of Impagliazzo's hardcore theorem}

Our second result concerns Impagliazzo’s hardcore theorem~\cite{Imp95} from complexity theory. Suppose $f : \{0,1\}^n \to \{0,1\}$ is mildly hard for size-$s$ circuits in the sense that every such circuit disagrees with $f$ on at least $1\%$ of inputs. Of course, different size-$s$ circuits may err on different sets of density $1\%$. The hardcore theorem shows that there is nevertheless a {\sl fixed} set of inputs on which $f$’s hardness is concentrated: there is a set $H \subseteq \{0,1\}^n$ of constant density such that $f$ is extremely hard against size-$s'$ circuits on inputs drawn from $H$. In more detail, for every~$\gamma$, there is a set $H$ of constant density such that every circuit of size $s' \le O(\gamma^2)\cdot s$ agrees with $f$ on at most a $\frac1{2} + \gamma$ fraction of inputs within~$H$.%\gnote{Do we get the same hardcore set for all choices of $\gamma$? It may work using techniques from [BB20 ``A New Minimax Theorem for Randomized Algorithms"], but I don't think the original statement of the hardcore theorem suffices.} 

The hardcore theorem was originally introduced to give a new proof of Yao’s XOR lemma~\cite{Yao82,GNW11} and has since found applications in cryptography~\cite{Hol05} and pseudorandomness~\cite{VZ12}. It has also been shown to be closely related to the dense model theorem in arithmetic combinatorics~\cite{RTTV08,TTV09} and the notion of multicalibration in algorithmic fairness~\cite{HJKRR18,CDV24}. Trevisan calls the hardcore theorem “one of the bits of magic of complexity theory”~\cite{Tre07}.   

\paragraph{Size loss and smooth boosting.} A downside of this result is the loss in circuit size, i.e.~the fact that $f$'s hardness on $H$ only holds against circuits of size $s'$ where $s' \ll s$. To see why this size loss occurs in all existing proofs of the hardcore theorem~\cite{Imp95,KS03,Hol05,BHK09}, we note that they all proceed via the contrapositive. One assumes that for every $H$ of constant density, there is a circuit of size $s'$ that agrees with $f$ on at least a $\frac1{2}+\gamma$ fraction of the inputs in $H$, and one constructs a circuit of size $s$ that agrees with $f$ on $99\%$ of all inputs. This size-$s$ circuit is obtained by combining several size-$s'$ circuits that one gets by instantiating the assumption with different~$H$’s. 

Klivans and Servedio~\cite{KS03} observed this formulation of the hardcore theorem in its contrapositive syncs up perfectly with the setup of smooth boosting: the uniform distribution over sets $H$ of constant density correspond to smooth distributions; the size-$s'$ circuits that achieve $\gamma$-advantage on the $H$'s can be viewed as weak hypotheses; the final size-$s$ circuit combines several size-$s'$ weak hypotheses into a strong hypothesis that achieves accuracy $99\%$, exactly like in boosting.  %This connection between smooth boosting and the hardcore theorem has proved fruitful, with subsequent works using techniques for the former to improve the latter and vice versa~\cite{Hol05,BHK09}.

It is clear that such a proof strategy inevitably results in a statement where $s \gg s'$, i.e.~inevitably incurs a size loss. Indeed, Lu, Tsai, and Wu~\cite{LTW11} formalized the notion of a “strongly black box proof” and showed that such proofs must incur a size loss of $s' \le O(\gamma^2) \cdot s$, matching the parameters achieved by known proofs. We refer the reader to their paper for the precise definition of a strongly black box proof, mentioning here that it is a special case of proofs that “proceed via the contrapositive”. 

\cite{LTW11}’s result still leaves open the question, first raised by Trevisan~\cite{Tre10}, of whether such a size loss is {\sl inherent} to the statement of the hardcore theorem, regardless of proof strategy. Our second result shows that this is indeed the case:\medskip

\begin{tcolorbox}[colback = white,arc=1mm, boxrule=0.25mm]
\begin{theorem}
    \label{thm:impagliazzo intro}
    For any $\gamma>0$ and sufficiently large $s$, there is an $f:\bits^n\to\bits$ such that 
    \begin{enumerate}
        \item \textbf{$f$ is mildly hard for size-$s$ circuits}: Every circuit of size $s$ agrees with $f$ on at most $99\%$ of inputs in $\bits^n$.
        \item \textbf{For every hardcore set, $f$ is mildly correlated with a small circuit}: For all constant density sets $H\sse \bits^n$, there is a circuit of size $O(\gamma^2 s)$ which computes $f$ on $\tfrac{1}{2}+\gamma$ fraction of inputs from $H$.
    \end{enumerate}
\end{theorem}
\end{tcolorbox}\medskip

\cite{LTW11} remarked in their paper that proving an unconditional result such as~\Cref{thm:impagliazzo intro}, one with no restriction on proof strategy, “appears to require proving circuit lower bounds, which seems to be far beyond our reach.” This is only a barrier if one requires $f$ to be explicit---our circuit lower bound in~\Cref{thm:impagliazzo intro} is proved using a (fairly involved) counting argument. 

\paragraph{Relationship between~\Cref{thm:smooth intro,thm:impagliazzo intro}.} They are incomparable, but our proof of~\Cref{thm:impagliazzo intro} is simpler and it will be more natural for us to present it first. By known connections between the hardcore theorem and smooth boosting~\cite{KS03},~\Cref{thm:smooth intro} implies an ${\Omega}(1/\gamma^2)$ lower bound on the {\sl round complexity} of smooth boosting.  (This is not a new result as an $\Omega(1/\gamma^2)$ lower bound on the round complexity even of distribution-independent boosting has long been known~\cite{Fre95}.) Proving an $\tilde{\Omega}(1/\gamma^2)$ lower bound on the sample complexity overhead of smooth boosting is significantly more difficult.

\subsection{Other related work}

Larsen and Ritzert~\cite{LR22} studied the sample complexity of distribution-independent boosting in terms of the VC dimension $d$ of the weak learner's hypothesis class, giving matching upper and lower bounds of $\Theta(d/\gamma^2)$. Our focus is on understanding the sample complexity {\sl overhead} of boosting, which is why our bounds are instead parameterized in terms of the sample complexity of weak learning the concept class (which can be different from $d$). Indeed, as already discussed, our work shows a separation between the sample complexity overheads of smooth and distribution-independent boosting, whereas~\cite{LR22}'s lower bound applies equally to both. (Our techniques are entirely different from~\cite{LR22}'s.) 

The size loss in the hardcore theorem translates into a corresponding size loss in Yao's XOR lemma. While our results do not have any direct implications for Yao's XOR lemma, we mention that there is also a line of work devoted to understanding the limitations of ``black box" (and other restricted types of) proofs of it~\cite{Sha04,AS14,AASY16,GR08,SV10,GSV19,Sha23}. Obtaining an analogue of~\Cref{thm:impagliazzo intro} for Yao's XOR lemma is a natural avenue for future work and is already a well-known challenge within complexity theory. Quoting~\cite{Imp95},  ``Why in all Yao-style arguments is there a trade-off between resources and probability, rather than a real increase in the hardness in the problem? If $f$ is hard
for resources $R$, the parity of many copies of $f$ should still be hard for resources $R$, not just some slightly smaller bound."

%% file: ProofOverview.tex
\section{Proof overview for~\Cref{thm:impagliazzo intro}: Tightness of the hardcore theorem}

\subsection{Tightness of the hardcore theorem for junta complexity}

Rather than directly prove that the size loss in the hardcore theorem is necessary for the circuit model of computation, we first prove it necessary for a substantially simpler model of computation, {\sl juntas}. A function $f: \bits^n \to \bits$ is a {\sl $k$-junta} if there is an $h:\bits^k \to \bits$ and subset of $k$ coordinates $S \subseteq [n]$ such that     
        $f(x) = h(x_S)$ for all $x \in \bn$.

\begin{definition}[Junta complexity]
    \label{def:junta-proof-overview}
    For any $g: \bits^n \to \bits$, the \emph{$\delta$-approximate junta complexity of $g$}, denoted $J(g, \delta)$, is the smallest $k$ for which there is a $k$-junta that agrees with $g$ on $(1 - \delta)$-fraction of all inputs. For a set $H\sse \bn$, we write $J_H(g,\delta)$ to denote the analogous quantity where agreement is measured with respect to the fraction of inputs from $H$.
\end{definition}
The hardcore theorem also applies to junta complexity: For any $g$ that is mildly hard for size-$k$ juntas and parameter $\gamma$, there is a hardcore set $H$ of constant density such that all juntas that achieve accuracy $\frac1{2} + \gamma$ on $H$ must have size $\Omega(\gamma^2 k)$. We show that this size loss is necessary and tight for juntas.
\begin{claim}[Tightness of the hardcore theorem for juntas]
    \label{claim:junta-size-loss-proof-overview}
    Fix any constant $c > 0$ and any sufficiently large and even $k$, the majority function on $k$ bits satisfies,
    \begin{enumerate}
        \item Every $k/2$ junta agrees with $\Maj_k$ on less than $0.8$ fraction of inputs. That is, $J(\Maj_k,0.2)>k/2$.
        \item For every set $H$ of density $c$, there is a $1$-junta that agrees with $\Maj_k$ on $\tfrac{1}{2}+\Omega_c(1/\sqrt{k})$ fraction of the points in $H$. That is, $J_H(\Maj_k, \frac1{2} - \Omega_c(1/\sqrt{k})) \le 1$.
    \end{enumerate}
\end{claim}
Both parts of \Cref{claim:junta-size-loss-proof-overview} follow from straightforward calculations. Taking $\gamma \coloneqq 1/\sqrt{k}$, it implies that $\Maj_k$ is mildly hard for $\Omega(k)$-juntas, and yet, for every hardcore set of constant density, it is possible to achieve advantage $\Omega(\gamma)$ using only an $O(\gamma^2  k)$ junta. %This shows that size-loss is necessary in the hardcore theorem, at least for junta complexity.

\subsection{Lifting junta complexity to circuit complexity}
The brunt of the work in proving \Cref{thm:impagliazzo intro} is a {\sl lifting theorem}: If there is a function $g$ showing that size loss is necessary in the hardcore theorem for juntas, there is a corresponding function $F$ showing size loss is necessary for circuits. Proving such a lifting theorem requires a circuit lower bound; therefore, the choice of $F$ will need to be non-explicit. In particular, we will show that there is at least one such $F$ within the {\sl lifted class} of $g$. In the below definition, a ``balanced" function refers to one that outputs $0$ and $1$ on an equal number of inputs.\footnote{We restrict the definition to balanced functions for technical reasons that are not crucial for this high-level discussion.}

\ifnum\focs=1
\begin{definition}[Lifted class]
    \label{def:lifted}
    For any function $g:\bits^k \to \bits$ and $n \in \N$, we use $\lift_n(g)$ to denote the $n$-bit \emph{lifted class} of $g$ defined as
    \begin{align*}
        \lift_n(g) \coloneqq \{ g( &f_1, \ldots, f_k)\mid f_i: \bits^n \to \bits \\
        & \text{ is balanced for each }i = 1,\ldots, k\}.
    \end{align*}
\end{definition}
\else
\begin{definition}[Lifted class]
    \label{def:lifted}
    For any function $g:\bits^k \to \bits$ and $n \in \N$, we use $\lift_n(g)$ to denote the $n$-bit \emph{lifted class} of $g$ defined as
    \begin{equation*}
        \lift_n(g) \coloneqq \set*{ g(f_1, \ldots, f_k)\mid f_i: \bits^n \to \bits\text{ is balanced for each }i = 1,\ldots, k}.
    \end{equation*}
\end{definition}
\fi

We show that the circuit complexity of approximating the worst-case function in $\lift_n(g)$ is {\sl characterized} by the junta complexity of $g$:\medskip

\begin{tcolorbox}[colback = white,arc=1mm, boxrule=0.25mm]
\begin{theorem}[Lifting junta complexity to circuit complexity]
    \label{thm:lift-proof-overview}
    For any $g:\bits^k \to \bits$ and $n \geq k$,
    \begin{itemize}
        \item \textbf{Upper bounds lift:} Fix any constant $c > 0$. If for all sets $H_g \subseteq \bits^k$ of density $c$, we have $J_{H_g}(g,\lfrac1{2}-\gamma) \le r_{\sm}$, then for all $F \in \lift_n(g)$ and sets $H_F \subseteq \bits^{nk}$ of density $c$, there is a circuit of size
        \begin{equation*}
            O(r_{\sm} \cdot \tfrac{2^n }{n})
        \end{equation*}
        that agrees with $F$ on $\tfrac{1}{2} + \gamma$ fraction of inputs in $H_F$.
        \item \textbf{Lower bounds lift:} If $J(g,\delta) \ge r_{\lar}$
 %all $r_{\lar}$ juntas agree with $g$ on at most $1 - \delta$ fraction of inputs in $\bits^k$, 
 then there is an $F \in \lift_n(g)$ for which all circuits of size
        \begin{equation*}
            \Omega(r_{\mathrm{large}} \cdot \tfrac{2^n }{n})
        \end{equation*}
        agree with $F$ on at most $1 - \Omega(\delta)$ fraction of inputs in $\bits^{nk}$.
    \end{itemize}
\end{theorem}
\end{tcolorbox}
\medskip 

\Cref{thm:impagliazzo intro} (tightness of the hardcore theorem for circuits) follows by combining \Cref{thm:lift-proof-overview} with \Cref{claim:junta-size-loss-proof-overview} (tightness of the hardcore theorem for juntas).

The upper bound of \Cref{thm:lift-proof-overview} is straightforward. A basic fact of circuit complexity shows that every $n$-bit function $f$ can be computed {\sl exactly} by a circuit of size $O(2^n/n)$~\cite{Lup58}. By the assumption on the junta complexity of $g$, there is some set of $r$ many $f_i$'s that are sufficient to approximate $F$ to the desired accuracy. For the upper bound, we compute these $f_i$'s exactly using $r$ many circuits of size $O(2^n/n)$ and then combine the responses. The lower bound shows that this naive strategy is optimal.

\begin{remark}[Contrast with Uhlig's mass production theorem]
    It is interesting to contrast our lower bound with {\sl Uhlig's mass production theorem} \cite{Uhl74,Uhl92}. This surprising theorem states that for any $g:\bits^k \to \bits$ and any {\sl single} $f:\bits^n \to \bits$, the composed function
\begin{equation*}
    F(X^{(1)},\ldots, X^{(k)}) = g(f(X^{(1)}), \ldots, f(X^{(k)}))
\end{equation*}
can be computed by a circuit of size $O(2^n/n)$, with {\sl no} overhead in terms of $g$'s complexity. This implies that many copies of a function $f$ can be computed ``for free," since a single copy of a worst-case $f$ requires circuits of size $\Omega(2^n/n)$~\cite{Sha49}. 

In contrast, the lower bound of \Cref{thm:lift-proof-overview} shows that if we wish to compute $g$ applied to $k$ {\sl different} functions $f_1,\ldots,f_k$, then an overhead equaling $g$'s junta complexity is necessary.
\end{remark}

\subsection{Soft junta complexity}
One reason that the lower bound in \Cref{thm:lift-proof-overview} is challenging to prove is that there is a broader class of ``softer" strategies to consider: Rather than {\sl exactly}  computing $r$ many $f_i$'s, one could choose to {\sl approximate} much more than $r$ many $f_i$'s. This could result in a smaller overall circuit since an arbitrary $f_i:\bits^n \to \bits$ can be approximated to non-trivial accuracy by a circuit of size $\ll 2^n/n$. To capture and reason about such strategies, we will introduce {\sl soft} junta complexity.
\begin{definition}[$\alpha$-correlated-distance and error]
    \label{def:alpha-corr-error}
    For any $\alpha \in [-1,1]^k$, let $\bx,\by$ be random variables on $\bits^k$ with the following joint distribution:
    \begin{enumerate}
        \item $\bx$ is drawn uniformly on $\bits^k$.
        \item Each bit of $\by_i$ is independently set to $\bx_i$ with probability $(1 + \alpha_i)/2$ and otherwise set to $-\bx_i$. Note that this guarantees the correlation is $\Ex[\bx_i \by_i] = \alpha_i$.
    \end{enumerate}
    For any $g,h: \bits^k \to \bits$ and $\alpha \in [-1,1]^k$, the \emph{$\alpha$-correlated-distance of $g$ and $h$} is defined as
    \begin{equation*}
        \dist_{\alpha}(g,h) \coloneqq \Pr[g(\by) \neq h(\bx)].
    \end{equation*}
    When $\alpha = \vec{1}$, we drop the subscript and refer to this quantity as simply the distance between $g$ and $h$,
    \begin{equation*}
        \dist(g,h) \coloneqq \Prx_{\bx \sim \bits^k}[g(\bx) \neq h(\bx)].
    \end{equation*}
    Finally, the \emph{$\alpha$-correlated-error of $g$} is the quantity
    \[ \error_{\alpha}(g) \coloneqq \mathop{\min}_{h : \bits^k\to\bits} \Big\{ \dist_{\alpha}(g,h) \Big\}. \]
\end{definition}
\begin{definition}[Soft junta complexity]
    \label{def:soft-junta-complexity}
    For any $g:\bits^k \to \bits$ and $\delta > 0$, the \emph{$\delta$-approximate soft junta complexity} of $g$, denoted $\tilde{J}(g, \delta)$, is defined as
    \begin{equation*}
        \tilde{J}(g, \delta) \coloneqq \inf_{\substack{\alpha \in [-1,1]^k, \\ \error_{\alpha}(g) \leq \delta}} \Bigg\{\sum_{i \in [k]}\alpha_i^2\Bigg\}.
    \end{equation*}
\end{definition}
Note that standard (non-soft) junta complexity can similarly be defined in terms of $\alpha$-correlated error, but where $\alpha$ is only allowed to be chosen from the set $\zo^k$. Soft junta complexity can therefore be thought of as a continuous relaxation of standard junta complexity.
%\begin{equation*}
%    J(g, \delta) = \min_{\substack{\alpha \in \zo^k, \\ \error_{\alpha}(g) \leq \delta}} \Bigg\{\sum_{i \in [k]}\alpha_i^2\Bigg\}.
%\end{equation*}

The proof of the lower bound in~\Cref{thm:lift-proof-overview} has two main steps: First, we show that if $g$ has high soft junta complexity, then there is a function in $\lift_{n}(g)$ that requires a large circuit to approximate. 

\begin{lemma}[Step 1: Lower bound in terms of soft junta complexity]
\label{lem:lower-bound-from-soft}
     For any $k \leq 2^{n-1}$ and $g:\bits^k \to \bits$, there is some $F \in \lift_n(g)$ for which any circuit that agrees with $F$ on $1-\delta$ fraction of inputs has size at least $\Omega(\tilde{J}(g,2\delta) \cdot 2^n /n )$.
\end{lemma}

Second, we show that although soft juntas are a broader class than (standard) juntas, their expressive powers are equivalent up to constant factors.

\begin{lemma}[Step 2: Relating soft junta complexity and standard junta complexity]
    \label{lem:connect-soft-and-hard-juntas}
    For any $g:\bits^n \to \bits$ and $\delta \geq 0$,
    \begin{equation*}
        \lfrac{1}{2} \cdot J(g, 4\delta) \leq \tilde{J}(g,\delta) \leq J(g,\delta).
    \end{equation*}
\end{lemma}

We overview our proofs of~\Cref{lem:lower-bound-from-soft,lem:connect-soft-and-hard-juntas} in turn in~\Cref{sec:lower-bound-from-soft,sec:connect-soft-and-hard-juntas} respectively. 

\begin{remark}[Soft {\sl query} complexity] 
In~\cite{BB20} Ben-David and Blais introduced a soft notion of query complexity (which they term {\sl noisy} query complexity) that generalizes standard query complexity the same way our definition of soft junta complexity generalizes standard junta complexity. \cite{BB20} show that relating soft and standard query complexity in the same way as we relate soft and standard junta complexity in \Cref{lem:connect-soft-and-hard-juntas} would resolve the {\sl randomized composition conjecture}, a major open problem in complexity theory.
\end{remark}

\subsection{\Cref{lem:lower-bound-from-soft}: Lower bound in terms of soft junta complexity}
\label{sec:lower-bound-from-soft}

We prove~\Cref{lem:lower-bound-from-soft} using a net-based argument. 

\begin{lemma}[Many functions are needed to cover $\lift_n(g)$]
    \label{lem:each-circuit-covers-few-F-soft}
    For any $g:\bits^k \to \bits$ and $C:\bits^{nk} \to \bits$,
    \begin{equation*}
        \Prx_{\bF \sim \lift_n(g)}[\dist(C, \bF) \leq \delta] \leq \exp({-\tilde{J}(g,2\delta)\cdot \Omega(2^n-k)})
    \end{equation*}
\end{lemma}
By \Cref{lem:each-circuit-covers-few-F-soft}, if every $F \in \lift_n(g)$ can be approximated to accuracy $1-\delta$ by a circuit of size-$s$, then the number of circuits of size $s$ must be at least $2^{\tilde{J}(g,2\delta)\cdot \Omega(2^n-k)}$. This, combined with the fact that there are only $(n+s)^{O(s)}$ size-$s$ circuits gives \Cref{lem:lower-bound-from-soft}.

The first observation in the proof of \Cref{lem:each-circuit-covers-few-F-soft} is that,
\ifnum\focs=1
\begin{align*}
    &\max_{C:\bits^{nk} \to \bits}\bigg\{\Prx_{\bF \sim \lift_n(g)}[\dist(C,\bF) \leq \delta]\bigg\}\\& \leq\max_{F \in \lift_n(g)}\bigg\{\Prx_{\bF' \sim \lift_n(g)}[\dist(F,\bF') \leq 2\delta]\bigg\}.
\end{align*}
\else
\begin{equation*}
    \max_{C:\bits^{nk} \to \bits}\bigg\{\Prx_{\bF \sim \lift_n(g)}[\dist(C,\bF) \leq \delta]\bigg\} \leq \max_{F \in \lift_n(g)}\bigg\{\Prx_{\bF' \sim \lift_n(g)}[\dist(F,\bF') \leq 2\delta]\bigg\}.
\end{equation*}
\fi
The above follows an easy application of the triangle inequality: If $C$ is $\delta$-close to both $F$ and $F'$, then $\dist(F, F') \leq 2\delta$. As a result, our goal is to analyze $\Prx_{\bF' \sim \lift_n(g)}[\dist(F,\bF') \leq 2\delta]$. This is where soft junta complexity plays a key role. As we show in \Cref{prop:soft-junta-error}, if $F=g(f_1,\ldots,f_k)$, and $F'=g({f_1}',\ldots,{f_k}')$ satisfy $\dist(F, F') \leq 2\delta$, then
\begin{equation*}
    \sum_{i =1}^k \Ex_{\bx \sim \bits^n}[f_i(\bx){f_i}'(\bx)]^2 \geq \tilde{J}(g,2\delta).
\end{equation*}
\Cref{lem:each-circuit-covers-few-F-soft} therefore follows from the below concentration inequality.
\begin{lemma}[Main concentration inequality]
\label{lem:concentration-alphas-proof-overview}
    For each $i \in [n]$, let $f_i:\bits\to\bits$ be an arbitrary balanced function,  ${\boldf_i}':\bits^n\to\bits$ be a uniformly random balanced function (chosen independently for each $i$), and $\balpha_i \coloneqq \Ex_{\bx\sim\bits^n}[f_i(\bx)\boldf{_i}'(\bx)]$. Then, for all $t \ge 0$,
    \begin{equation}
        \label{eq:Bernstein-applied}
        \Prx_{\balpha_1, \dots, \balpha_k}\left[\sum_{i=1}^k \balpha_i^2 \ge t\right]\le \exp(-\Omega(t\cdot 2^n-k)).
    \end{equation}
\end{lemma}
For each $i \in [k]$ and $x \in \bits^n$, define $\bz(i,x) \coloneqq f_i(x)\boldf{_i}'(x)$. Then, \Cref{lem:concentration-alphas-proof-overview} \emph{almost} follows from the following logic using standard properties of sub-Gaussian and sub-exponential random variables.
\begin{enumerate}
    \item Each $\bz(i,x)$ is bounded on $[-1,1]$ and is therefore sub-Gaussian.
    \item If the $\bz(i,x)$'s {\sl were} independent---which unfortunately, they are {\sl not}---then the random variables $\balpha_i \coloneqq \Ex_{\bx \sim \bits^n}[\bz(i,\bx)]$ would also be sub-Gaussian with sub-Gaussian norm $O(1/\sqrt{2^n})$.
    \item Since the square of a sub-Gaussian random variable is sub-exponential, $\balpha_i^2$ is sub-exponential. Then, \Cref{eq:Bernstein-applied} follows from an appropriate form of Bernstein's inequality.
\end{enumerate}

The $\bz(i,x)$'s are not independent because $\boldf_i'$ is chosen uniformly among {\sl balanced} functions, meaning there are correlations between the coordinates of $\boldf_i'$. For example, consider the probability that $\balpha_i = 1$. If the $\bz(i,x)$ were independent, this probability would be $2^{-2^n}$. However, as $f_i$ and $\boldf_i'$ are balanced, this probability is $\binom{2^n}{ 2^{n-1}}^{-1} = \Theta(\sqrt{2^n} \cdot 2^{-2^n})$, which is substantially larger.

To get around this issue of independence, we use a coupling argument. We show that the $\bz(i,x)$'s can be coupled to idealized $\wh{\bz}(i,x)$'s that are independent, such that the number of $x \in \bits^n$ on which $\bz(i,x)$ and $\wh{\bz}(i,x)$ differ is also sub-Gaussian. After this coupling, a similar but carefully modified series of steps to the prior proof strategy gives \Cref{lem:concentration-alphas-proof-overview}.

\subsection{\Cref{lem:connect-soft-and-hard-juntas}: Relating soft and standard junta complexity}
\label{sec:connect-soft-and-hard-juntas}

One side of \Cref{lem:connect-soft-and-hard-juntas} is immediate: Soft juntas are more expressive than standard juntas so $\tilde{J}(g,\delta) \leq J(g,\delta)$. The other direction is more challenging: It says that, if $g$ has a soft junta achieving error $\delta$, there is a standard junta using only twice as many coordinates that achieves $4\delta$ error. To prove this, we will argue that an appropriately chosen random standard junta satisfies this.
\begin{claim}[Error of a random hard junta]
    \label{claim:random-junta-double-error-proof-overview}
    For any $\alpha \in [-1,1]^k$, let $\bz_i$ be drawn independently from $\Ber(\alpha_i^2)$ for each $i \in [k]$. Then, the expected $\bz$-correlated-error of $g$ is at most double the $\alpha$-correlated-error of $g$.
\end{claim}
Given \Cref{claim:random-junta-double-error-proof-overview}, the other direction of \Cref{lem:connect-soft-and-hard-juntas} follows from the probabilistic method. 

The proof of \Cref{claim:random-junta-double-error-proof-overview} recasts $\alpha$-correlated-error in a more convenient form. For $\mcD(\alpha)$ be the distribution on $\bx,\by$ defined in \Cref{def:alpha-corr-error},
\begin{align*}
    \error_{\alpha}(g) \coloneqq& \min_{h:\bits^k \to \bits}\set*{\Prx_{\bx,\by \sim \mcD(\alpha)}[g(\by) \neq h(\bx)]} \\
    =& \Ex_{\bx} \bracket*{\min_{h(\bx) \in \bits} \Prx_{\by \mid \bx}[g(\by) \neq h(\bx)]} \\
    =& \Ex_{\bx} \bracket*{\frac{1 - \abs*{\Ex_{\by \mid \bx}[g(\by)]}}{2}}.
\end{align*}
The absolute value in the above expression is a bit difficult to work with, so we will replace it with a quadratic approximation. In particular, for $\Phi(t) = 1- t^2$, we have $\Phi(t)/4 \leq \tfrac{1 - |t|}{2} \leq \Phi(t)/2$. Therefore,
\begin{equation*}
    \error_{\alpha}(g) = \Theta\paren*{1 - \Ex_{\bx} \bracket*{ \Ex_{\by \mid \bx}[g(\by)]^2}}.
\end{equation*}
The last step is show that $\Ex_{\bx} \bracket*{ \Ex_{\by \mid \bx}[g(\by)]^2}$ is constant regardless of whether $\bx,\by \sim \mcD(\alpha)$ or $\bx,\by \sim \mcD(\bz)$ where $\bz$ is drawn as in \Cref{claim:random-junta-double-error-proof-overview}. To do so, we use Fourier analysis to write both quantities in terms of $g$'s Fourier spectrum and show they are equal.

\section{Proof overview for~\Cref{thm:smooth intro}: Sample complexity of smooth boosting}
% Our result will be stated in the (distribution-dependent) PAC learning setting.
% \begin{definition}[Sample complexity of learning]
%     For any concept class $\mcC$, distribution $\mcD$ on inputs, and $\eps, \delta > 0$, we use $m_{\mcD,a,\delta}(\mcC)$ to denote the sample complexity $\mcC$ to accuracy $a$: 
% \end{definition}
Proving \Cref{thm:smooth intro} requires exhibiting a concept class $\mcC$ with two properties: First, there is a weak learner that uses $m$ samples and achieves accuracy $1/2 + \gamma$ with high probability on any smooth distribution, and second, any algorithm that learns $\mcC$ to accuracy 0.99 must use $\tilde{\Omega}(m/\gamma^2)$ samples. We'll set $\mcC = \lift_n(\Maj_k)$ where $n = \log m$ and $k = \tilde{\Theta}(1/\gamma^2)$. 

The lower bound transfers nicely from our proof that the hardcore theorem is tight.
\begin{lemma}[Strong learning $\lift_n(\Maj_k)$ requires many samples]
    \label{lem:learning-lb}
    For any $n \geq \Omega(\log k)$ and learning algorithm that, on the uniform distribution of inputs, learns $\lift_n(\Maj_k)$ to accuracy 0.99 with high probability must use at least $\Omega(k 2^n)$ samples.
\end{lemma}
The proof of \Cref{lem:learning-lb} utilizes the tools we have developed to prove the tightness of the hardcore theorem. Combining \Cref{lem:each-circuit-covers-few-F-soft} and \Cref{claim:junta-size-loss-proof-overview} gives that any hypothesis can only ``cover" $2^{-\Omega(k 2^n)}$ fraction of the possible $F \in \lift_n(\Maj_k)$. Any algorithm using $m$ samples only receives $m$ bits of information about $F$, and so can only effectively output $2^m$ possible hypothesis. Combining these, we must have that $m \geq \Omega(k 2^n)$.

\subsection{The weak learner}
All that remains is to prove the upper bound:
\begin{lemma}[$\lift_n(\Maj_k)$ can be weak learned with few samples]
    \label{lem:weak-learner}
    For any $n \geq \Omega(\log k)$, there is an algorithm that, for any smooth distribution and $F \in \lift_n(\Maj_k)$, uses $2^n$ samples and, with high probability outputs a hypothesis that has accuracy at least $\tfrac{1}{2} + \tilde{\Omega}(1/\sqrt{k})$.
\end{lemma}
One could hope that \Cref{lem:weak-learner} follows easily from the upper bound in \Cref{thm:lift-proof-overview}. Indeed, one view \Cref{thm:lift-proof-overview} is that learning $2^n$ bits of information about $F$ is sufficient to weak learn. In particular, it says that, for $F = \Maj(f_1, \ldots, f_k)$, fully learning the truth table of one of the $f_i$ would suffice. Unfortunately, while the learning algorithm will receive $2^n$ bits of information about $F$ through the sample, they won't be the right bits to strong learn any $f_i$. This is because the sample is labeled by $F$, not $f_i$. Therefore, there is only a weak correlation between the samples we see and the truth table of each $f_i$.

Instead, our learner will, roughly speaking, simultaneously weak learn all of $f_1, \ldots, f_k$ and combine these weak learners into one hypothesis.

\pparagraph{Learning over the uniform distribution.} For intuition, we first overview how to weak learn $\lift_n(\Maj_k)$ over the uniform distribution. Over an arbitrary smooth distribution, the algorithm will be similar, though the analysis is noticeably more involved.

Our weak learner builds weak learners $g_1, \ldots, g_k$ for $f_1, \ldots, f_k$ as follows. Whenever it receives a sample $(X,y)$, it sets $g_i(X^{(i)}) = y$ for each $i \in [k]$. The intuition is the label $y = \Maj(f_1(X^{(1)}), \ldots, f_k(X^{(k)}))$ is slightly correlated with each $f_i(X^{(i)})$, and so setting $g_i(X^{(i)}) = y$ achieves a positive correlation. As a result,
\ifnum\focs=1
\begin{equation}
    \label{eq:weak-learn-all}
    \Ex_{\bx \sim \bits^n}[f_i(\bx) g_i(\bx)] = \Theta\paren*{\tfrac{1}{\sqrt{k}}}
\end{equation}
for each $i = 1,\ldots, n$ with high probability.
\else
\begin{equation}
    \label{eq:weak-learn-all}
    \Ex_{\bx \sim \bits^n}[f_i(\bx) g_i(\bx)] = \Theta\paren*{\tfrac{1}{\sqrt{k}}} \quad\quad\text{for each $i = 1,\ldots, n$ with high probability.}
\end{equation}
\fi
The final step is to combine these weak learners by outputting the hypothesis $h(X) \coloneqq \Maj(g_1, \ldots, g_k)$. Since the base distribution is uniform, the weak learners, $g_1, \ldots, g_k$, are independent, and so it is fairly straightforward to compute the expected accuracy of $h$. After an appropriate calculation, we see that $h$ will, on average, achieve accuracy $\tfrac{1}{2} + \tfrac{1}{\sqrt{k}}$, as desired.

\subsection{Challenges of learning over non-uniform distributions}

We wish for our learner to succeed over any smooth distribution. The first challenge is that when the base distribution is not guaranteed to be uniform, \Cref{eq:weak-learn-all} may not hold: There are smooth distributions for which our algorithm will fail to weak learn some of the blocks. For example, consider the distribution that is uniform over $X$ satisfying, $F(X) \neq f_1(X^{(1)})$. This condition happens with probability $\tfrac{1}{2} - o(1)$ on the uniform distribution, so the resulting distribution is smooth (with parameter $2 + o(1)$). However, our strategy will give a weak learner that is \emph{anti}-correlated to $f_1$:
\begin{equation*}
     \Ex_{\bx}[f_1(\bx) g_1(\bx)] < 0.
\end{equation*}
The solution is, roughly speaking, to show that we weak learn on average over the blocks. The actual result we need is the following: Defining,
\begin{equation*}
    G(X) \coloneqq \sum_{i \in [k]} g_i(X^{(i)}),
\end{equation*}
we will show that $F$ and $G$ are well-correlated. For intuition, consider the case where the base distribution is truly uniform. Then, each $g_i$ has correlation $\Omega(1/\sqrt{k})$ with $f_i$, and each $f_i$ has correlation $\Omega(1/\sqrt{k})$ with $F$. Combining these gives that the correlation of $g_i$ and $F$ is $\Omega(1/k)$, which, by summing over the blocks, gives that $F$ and $G$ have a constant amount of correlation. We'll show that as long as the base distribution is smooth, the same holds:
\begin{equation}
    \label{eq:sum-correlates}
    \Ex[F(\bX)G(\bX)] \geq \Omega(1).
\end{equation}

\pparagraph{Loss of independence.} The second and more delicate challenge is that our weak learners, $g_{1}, \ldots, g_{k}$ are no longer independent. For example, the base distribution can be constructed in such a way that, for $x_1, x_2 \in \bits^n$, if we have successfully learned $g_1(x_1)$, then we are more likely to have also successfully learned $g_2(x_2)$. This can be accomplished by putting a relatively large weight on the inputs $X$ where $X^{(1)} = x_1, X^{(2)} = x_2,$ and $F(X) = f_1(x_1) = f_2(x_2)$.

To see why this lack of independence can be an issue, suppose our weak learners, $g_1, \ldots, g_k$, satisfied the following:
\begin{enumerate}
    \item On a third of inputs $X$, we get all blocks correct, meaning $f_i(X^{(i)}) = g_i(X^{(i)})$ for all $i \in [k]$.
    \item On the other two-thirds of inputs $X$, we get $\frac{k}{2} - 1$ blocks correct, meaning $f_i(X^{(i)}) = g_i(X^{(i)})$ for $\frac{k}{2} - 1$ choices of $i \in [k]$.
\end{enumerate}
In this setting, we will still have that $F$ and $G$ are well correlated (satisfying \Cref{eq:sum-correlates}), but if we output the hypothesis that is the majority of $g_1, \ldots, g_k$, that hypothesis will only get $1/3$ of inputs correct, worse than a random guess.

Our solution to this issue is to \emph{not} output the majority of the weak learners. Instead, we will show that for a randomly chosen threshold $\btau$, the hypothesis
\begin{equation*}
    h_{\btau}(X) \coloneqq \Ind\bracket{G(X) \geq \btau}
\end{equation*}
successfully weak learns on average over the choice of $\btau$. This random threshold alleviates the issue from earlier where, with large probability, the weak learners get exactly $\frac{k}{2} - 1$ blocks correct. Now, $h_{\btau}$ will successfully classify such inputs with probability close to $1/2$ (with the exact probability depending on the distribution of $\btau$).

We will choose $\btau$ uniformly from $\set{-u, -u+1, \ldots, u-1, u}$ for an appropriately chosen $u$. With a bit of arithmetic, we can lower bound the expected advantage at
\ifnum\focs=1
\begin{align*}
    \Ex_{\bX,\btau}&[f(\bX)h_{\btau}(\bX)] \\ &\geq \frac{\Ex_{\bX}[f(\bX)G(\bX)]}{u} - k \cdot \Prx_{\bX}[|G(\bX)| \geq u] \\
    &\geq \Omega(\tfrac{1}{u}) - k \cdot \Prx_{\bX}[|G(\bX)| \geq u].
\end{align*}
\else
\begin{equation*}
    \Ex_{\bX,\btau}[f(\bX) h_{\btau}(\bX)] \\ \geq \frac{\Ex_{\bX}[f(\bX)G(\bX)]}{u} - k \cdot \Prx_{\bX}[|G(\bX)| \geq u] \geq \Omega(\tfrac{1}{u}) - k \cdot \Prx_{\bX}[|G(\bX)| \geq u].
\end{equation*}
\fi
Here, we see the tension in choosing $u$: If it's too large, we will get little advantage from the first term, but if it's too small, the second term will subtract too much. The last step is to show that, as the base distribution is smooth, a Chernoff-like bound holds: For $u = O(\sqrt{k \log k})$, $\Prx_{\bX}[|G(\bX)| \geq u] \leq 1/k^2$, which makes the second term negligible for our purposes. As a result, our weak learner achieves advantage $\Omega(1/u) = \tilde{\Omega}(1/\sqrt{k})$.

%% file: Preliminaries.tex
\section{Preliminaries}
\paragraph{Notation and naming conventions.}{We write $[n]$ to denote the set $\{1,2,\ldots,n\}$. We use lowercase letters to denote bitstrings e.g. $x,y\in\zo^n$ and subscripts to denote bit indices: $x_i$ for $i\in [n]$ is the $i$th index of $x$. We use \textbf{boldface letters} e.g.~$\bx,\by$ to denote random variables. } For any distribution $\mcD$, we use $\mcD(x)$ as shorthand for $\Prx_{\bx \sim \mcD}[\bx =x]$.

\pparagraph{Standard concentration and anticoncentration inequalities.}

\ifnum\focs=1
\begin{fact}[Hoeffding's inequality \cite{hoeffding63:probability}]
\label{fact:hoeffding}
    Let $\bx_1, \dots, \bx_n$ be independent random variables such that for all $i$, $a_i \leq \bx_i \leq b_i$ with probability $1$. Then, for all $t > 0$,
    \begin{align*}
        \Prx&\bracket*{\abs*{\sum_{i \in [n]} \bx_i - \Ex \bracket*{\sum_{i \in [n]} \bx_i}} \geq t} \\ &\leq 2 \exp \paren*{\frac{-2t^2}{\sum_{i \in [n]} (b_i - a_i)^2}}.
    \end{align*}
\end{fact}
\else
\begin{fact}[Hoeffding's inequality \cite{hoeffding63:probability}]
\label{fact:hoeffding}
    Let $\bx_1, \dots, \bx_n$ be independent random variables such that for all $i$, $a_i \leq \bx_i \leq b_i$ with probability $1$. Then, for all $t > 0$,
    \begin{equation*}
        \Prx\bracket*{\abs*{\sum_{i \in [n]} \bx_i - \Ex \bracket*{\sum_{i \in [n]} \bx_i}} \geq t} \leq 2 \exp \paren*{\frac{-2t^2}{\sum_{i \in [n]} (b_i - a_i)^2}}.
    \end{equation*}
    
\end{fact}
\fi

\begin{fact}[Bounded differences inequality \cite{Mcd89}]
    \label{fact:bounded-diff}
    For any domain $\mcX$, product distribution $\mcD$ over $\mcX^m$, and function $\Psi:\mcX^m \to \R$ that satisfies the $c$-bounded differences inequality, meaning for any $X,X' \in \mcX^m$ that differ in one coordinate, $\Psi(X) - \Psi(X') \leq c$,
    \begin{equation*}
        \Pr_{\bX \sim \mcD}\bracket*{\Psi(\bX) \leq \Ex[\Psi(\bX)] -  \eps} \leq \exp\paren*{-\frac{2\eps^2}{mc^2}}.
    \end{equation*}
\end{fact}

\begin{fact}[Max probability the binomial puts on any outcome]
    \label{fact:bin-max-p}
    For any $k \in \N$, let $\bx_1, \ldots, \bx_k$ be independent and each uniform on $\bits$. Then, for any possible outcome $v$,
    \begin{equation*}
        \Pr\bracket[\Bigg]{\,\sum_{i \in [k]} \bx_i = v } \leq O\paren*{\tfrac{1}{\sqrt{k}}}.
    \end{equation*}
\end{fact}

\pparagraph{Smooth distributions and density of a distribution.}
\begin{definition}[$\kappa$-smooth distribution]
    For any $\kappa \ge 1$ a probability distribution $\mcD$ over a domain $\mcX$ is \emph{$\kappa$-smooth} if for all $x \in \mcX$,
    $$\mcD(x) \leq \frac{\kappa}{|\mcX|}.$$
\end{definition}

\begin{definition}[Density of a distribution]
    \label{def:distribution-density}
    For any $c \in (0,1]$, a probability distribution $H$ over $\mcX$ has \emph{density $c$} if for all $x\in\mcX$, we have $$H(x)\le \frac{1}{c |\mcX|}.$$
\end{definition}

We remark that a distribution has density $c$ if and only if it is $(\kappa \coloneqq 1/c)$-smooth.

\pparagraph{A helpful function.}
For any $t \in \R$, we use the $\sign$ function to denote
\begin{equation*}
    \sign(t) = \begin{cases}
        1&\text{if }t \geq 0 \\
        -1&\text{otherwise.}
    \end{cases}
\end{equation*}

\pparagraph{Standard learning definitions.}
\begin{definition}[Distribution specific PAC learning]
    \label{def:PAC-learning}
    For any concept class $\mcC$, we say an algorithm, $\mcA$, \emph{learns} $\mcC$ to accuracy $1 - \eps$ with success probability $1 - \delta$ over distribution $\mcD$ using $m$ samples if the following holds: For any $f \in \mcC$, given $m$ independent samples of the form $(\bx, f(\bx))$ where $\bx \sim \mcD$, $\mcA$ returns a hypothesis $h$, that with probability at least $1- \delta$, satisfies
    \begin{equation*}
        \Prx_{\bx \sim \mcD}[f(\bx) = h(\bx)] \geq 1 - \eps.
    \end{equation*}
    Furthermore, if the hypothesis $h$ satisfies, with probability at least $1 - \delta$,
    \begin{equation*}
    \Prx_{\bx \sim \mcD}[f(\bx) = h(\bx)] \geq \frac{1}{2} + \gamma,
    \end{equation*}
    then we say that $\mcA$ \emph{$\gamma$-weak learns} $\mcC$ with success probability $1-\delta$ over distribution $\mcD$ using $m$ samples.
\end{definition}
Since weak learning is concerned with hypotheses that have accuracy close to $\frac{1}{2}$, it will often be more convenient to work with \emph{advantage}.
\begin{definition}[Advantage]
    For any function $f$, hypothesis $h$, and distribution $\mcD$, we define the \emph{advantage} of $h$ w.r.t.~$f$ on distribution $\mcD$ as
    \begin{equation*}
        \Ex_{\bx \sim \mcD}\bracket*{h(\bx)f(\bx)}
    \end{equation*}
    When the distribution $\mcD$ and function $f$ are clear from context, we simply call the above quantity the advantage of $h$.
\end{definition}

\begin{definition}[Boosting algorithm] An algorithm $\mcB$ is a \emph{distribution-independent boosting algorithm} if for any function $f$ and any distribution $\mcD$, if $\mcB$ is given parameters $\eps > 0$, $\delta > 0$ and has access to a $\gamma$-weak learner $\mcA$ for any distribution and an example oracle $\textup{EX}(f,\mcD)$ then $\mcB$ returns a hypothesis $h$, that with probability at least $1 - \delta$, satisfies  
    \begin{equation*}
        \Prx_{\bx \sim \mcD}[f(\bx) = h(\bx)] \geq 1 - \eps.
    \end{equation*}

A \emph{smooth boosting algorithm} is a boosting algorithm that only has access to weak learners for \emph{smooth} distributions.
\end{definition}

%% file: Juntas.tex
\section{Tightness of the hardcore theorem for juntas: Proof of \Cref{claim:junta-size-loss-proof-overview}}
\label{sec:maj-properties}

We prove the two points in \Cref{claim:junta-size-loss-proof-overview} separately. In this section, we switch to considering hardcore \textit{distributions} rather than hardcore \textit{sets} since this will simplify some of the proofs.

% \begin{definition}[Density of a distribution]
%     \label{def:distribution-density}
%     A distribution $H$ over $\bits^n$ has density $c>0$ if for all $x\in\bits^n$, we have $\Pr_{\bx\sim H}[\bx=x]\le 1/(c 2^{n})$.
% \end{definition}
% \red{
% \begin{remark}[Relationship between smoothness and density of a distribution]
%     A distribution has density $c$ if and only if it is $1/c$-smooth.
% \end{remark}
% }
% The set of distributions with density $c$ is the convex hull of distributions which are uniform over a subset $H$ of size $|H|\ge c 2^n$.

\subsection{First part of \Cref{claim:junta-size-loss-proof-overview}}

\begin{claim}[Constant hardness of $\MAJ_k$ for $\frac{k}{2}$-juntas: first part of \Cref{claim:junta-size-loss-proof-overview}]
    \label{claim:maj-mildy-hard}
    Let $h:\bits^k\to\bits$ be any $\frac{k}{2}$-junta. Then,
    $$
    \Prx_{\bx\sim\bits^k}[\MAJ_k(\bx)=h(\bx)]\le \frac{3}{4}+O(k^{-1/2}).
    $$
\end{claim}

\begin{proof}
We prove the claim by showing that for any $\frac{k}{2}$-junta $h$, we have
$$\underset{x \sim \bits^k}{\Pr}[h(x) \neq \MAJ_k(x)] \geq \frac{1}{4} - O ( k^{-1/2})$$
{By definition, there is some $|S| = \frac{k}{2}$ and $f:\bits^{k/2} \to \bits$ for which, for all $x \in \bits^n$
\begin{equation*}
    h(x) = f(x_{S}).
\end{equation*}
Our first observation is that $\Pr[h(\bx) \neq \Maj_k(\bx)]$ is minimized when the above $f$ is the $\Maj_{k/2}$ function. This is because, to choose the $f$ with minimum error, we should set
\begin{equation*}
    f(y) = \sign\paren*{\Ex_{\bx \sim \bits^n}[\Maj_k(\bx) \mid \bx_S = y]}.
\end{equation*}
Furthermore, conditioned on $\bx_S = y$, the sum $\bX = \sum_{i \in [k]}\bx_i$ is a random variable that is symmetric about the sum $Y = \sum_{i \in [k]}y_i$. Therefore, the sign of the above expectation is exactly the majority of $y$. We are left with analyzing,
\begin{equation*}
    \Prx_{\bx \sim \bits^k}[\Maj_{k/2}(\bx_S) \neq \Maj_{k}(\bx)].
\end{equation*}
For notational convenience, we introduce two random variables,
\begin{equation}
    \label{eq:compare-majs}
    \bY \coloneqq \sum_{i \in S}\bx_i \quad\quad\text{and}\quad\quad\bZ \coloneqq \sum_{i \notin S}\bx_i.
\end{equation}
These two random variables are independent, identically distributed, and each symmetric about $0$. Furthermore, \Cref{eq:compare-majs} can be rewritten as
\ifnum\focs=1
\begin{align*}
    \Pr&[\sign(\bY) \neq \sign(\bY + \bZ)] \\ &\geq \Pr[|\bZ| > |\bY| \text{ and }\sign(\bZ) \neq \sign(\bY)].
\end{align*}
\else
\begin{equation*}
    \Pr[\sign(\bY) \neq \sign(\bY + \bZ)] \geq \Pr[|\bZ| > |\bY| \text{ and }\sign(\bZ) \neq \sign(\bY)].
\end{equation*}
\fi
Since $\bZ$ and $\bY$ are independent and identically distributed,
\begin{equation*}
    \Pr[|\bZ| > |\bY|] = \frac{1}{2} \cdot \Pr[|\bZ| \neq |\bY|].
\end{equation*}
Furthermore, conditioned on $|\bZ| > |\bY|$, we know that $|\bZ| \geq 1$ and since $\bZ$ is equally likely to take on the values $+v$ and $-v$ for each $v \geq 1$, we have that
\begin{equation*}
    \Pr[\sign(\bZ) \neq \sign(\bY) \mid |\bZ| > |\bY|] = \frac{1}{2}.
\end{equation*}
Combining the above, we conclude,
\ifnum\focs=1
\begin{align*}
    \Pr&[\sign(\bY) \neq \sign(\bY + \bZ)] \\&\geq \frac{1}{4} \cdot \Pr[|\bZ| \neq |\bY|] \\
    &\geq \frac{1}{4} \cdot \paren*{1 - \max_{v} \Pr[|\bZ| = v]} \tag{Independence of $\bZ$ and $\bY$}\\
    &\geq \frac{1}{4} \cdot \paren*{1 - O(k^{-1/2})}. \tag{\Cref{fact:bin-max-p}}
\end{align*}
\else
\begin{align*}
    \Pr[\sign(\bY) \neq \sign(\bY + \bZ)] &\geq \frac{1}{4} \cdot \Pr[|\bZ| \neq |\bY|] \\
    &\geq \frac{1}{4} \cdot \paren*{1 - \max_{v} \Pr[|\bZ| = v]} \tag{Independence of $\bZ$ and $\bY$}\\
    &\geq \frac{1}{4} \cdot \paren*{1 - O(k^{-1/2})}. \tag{\Cref{fact:bin-max-p}}
\end{align*}
\fi
Therefore, the distance of every $\tfrac{k}{2}$-junta to $\Maj_k$ is at least $1/4 - O(k^{-1/2})$, which is exactly what we wished to show.

}
\end{proof}

\subsection{Second part of \Cref{claim:junta-size-loss-proof-overview}}
{
\begin{claim}[$\Maj_k$ is well-correlated with a random dictator: second part of \Cref{claim:junta-size-loss-proof-overview}]
\label{claim:maj-easy}
Let $H$ be a distribution of density $c$ over $\bits^n$. Then, 
$$
\Ex_{\bi \sim [k]}\bracket*{\Ex_{\bx \sim H}[\Maj_k(\bx) \bx_{\bi}]} \geq \Omega\paren*{\tfrac{c}{\sqrt{k}}}
$$
% \anote{It feels like this should be $\Omega(c/\sqrt{k})$? This is what we had in the boosting section (using the fact that $c$-density is equivalent to $1/c$-smoothness).}\cnote{Yes, I agree. I flipped the definition in the proof and will update accordingly.}
\end{claim}

\begin{proof}
    We start by observing that 
    \ifnum\focs=1
    \begin{align*}
        \Ex_{\bi \sim [k]}&\bracket*{\Ex_{\bx \sim H}[\Maj_k(\bx) \bx_{\bi}]} \\ &= \frac{1}{k}\sum_{i=1}^k\Ex_{\bx \sim H}[\Maj_k(\bx) \bx_{i}]\\
        &=\frac{1}{k}\Ex_{\bx\sim H}\bigg[\bigg(\sum_{i\in [k]}\bx_i\bigg)\MAJ_k(\bx)\bigg]\tag{Linearity of expectation}\\
        &=\frac{1}{k}\Ex_{\bx\sim H}\bigg[\bigg|\sum_{i\in [k]}\bx_i\bigg|\bigg].\tag{Definition of $\MAJ_k$}
    \end{align*} 
    \else
    \begin{align*}
        \Ex_{\bi \sim [k]}\bracket*{\Ex_{\bx \sim H}[\Maj_k(\bx) \bx_{\bi}]} &= \frac{1}{k}\sum_{i=1}^k\Ex_{\bx \sim H}[\Maj_k(\bx) \bx_{i}]\\
        &=\frac{1}{k}\Ex_{\bx\sim H}\bigg[\bigg(\sum_{i\in [k]}\bx_i\bigg)\MAJ_k(\bx)\bigg]\tag{Linearity of expectation}\\
        &=\frac{1}{k}\Ex_{\bx\sim H}\bigg[\bigg|\sum_{i\in [k]}\bx_i\bigg|\bigg].\tag{Definition of $\MAJ_k$}
    \end{align*} 
    \fi
    Thus, it is sufficient to prove
    \begin{equation}
        \label{eq:maj-corr}
        \Ex_{\bx\sim H}\bigg[\bigg|\sum_{i\in [k]}\bx_i\bigg|\bigg]\ge \Omega({c\sqrt{k}}).
    \end{equation}
    For all $L\le n$, we have
    % \begin{align*}
    %     \Prx_{\bx\sim\bits^k}\bigg[\bigg|\sum_{i=1}^k\bx_i\bigg|\le L\bigg]&=\sum_{\ell=0}^{L}\Prx_{\bx\sim\bits^k}\bigg[\bigg|\sum_{i=1}^k\bx_i\bigg|=\ell\bigg]\\
    %     &=\sum_{\ell=0}^{L}2\Prx_{\bx\sim\bits^k}\bigg[\sum_{i=1}^k\bx_i=\ell\bigg]\\
    %     &\le O(Lk^{-1/2}).\tag{\Cref{fact:bin-max-p}}
    % \end{align*}
    \begin{align*}
        \Prx_{\bx\sim H}\bigg[\bigg|\sum_{i=1}^k\bx_i\bigg|\le L\bigg]&=\sum_{\ell=0}^{L}\Prx_{\bx\sim H}\bigg[\bigg|\sum_{i=1}^k\bx_i\bigg|=\ell\bigg]\\
        &\le \sum_{\ell=0}^{L}\frac{1}{c2^k}\cdot 2\binom{k}{\ell}\tag{$H$ is a $c$-density distribution}\\
        &\le O\left(\tfrac{L}{c\sqrt{k}}\right).\tag{$\binom{k}{\ell}\le O(2^k/\sqrt{k})$ for all $\ell$}
    \end{align*}
    Therefore, by choosing $L=\Theta(c\sqrt{k})$, we get
    \begin{align*}
        \Ex_{\bx\sim H}\bigg[\bigg|\sum_{i\in [k]}\bx_i\bigg|\bigg]&\ge L\Prx_{\bx\sim H}\bigg[\bigg|\sum_{i=1}^k\bx_i\bigg|>L\bigg]\\
        &\ge \Omega(c\sqrt{k})\tag{Choice of $L=\Theta(c\sqrt{k})$}
    \end{align*}
    which establishes \Cref{eq:maj-corr} as desired. 
    % Using this:
    % \begin{align*}
    %     \Ex_{\bx\sim H}\bigg[\bigg|\sum_{i\in [k]}\bx_i\bigg|\bigg]&\ge \frac{1}{c}\Ex_{\bx\sim \bits^n}\bigg[\bigg|\sum_{i\in [k]}\bx_i\bigg|\bigg]\tag{$H$ is a $c$-density distribution}\\
    %     &\ge \Omega\paren*{\tfrac{\sqrt{k}}{c}}
    % \end{align*}
    % which completes the proof.
\end{proof}

\subsection{Proof of \Cref{claim:junta-size-loss-proof-overview}}
We prove the two points separately:
\begin{enumerate}
    \item This follows immediately from \Cref{claim:maj-mildy-hard};
    \item \Cref{claim:maj-easy} shows that the correlation of $\MAJ_k$ with a random dictator is $\Omega_c\paren*{k^{-1/2}}$. Therefore, there is a fixed $i\in [n]$ achieving the desired accuracy: $\Ex_{\bx \sim H}[\Maj_k(\bx) \bx_{i}]\ge \Omega_c(k^{-1/2})$.
\end{enumerate}
This completes the proof.\hfill\qed
}

%% file: Lifting.tex
\section{Lifting junta complexity to circuit covering number: Proof of \Cref{thm:lift-proof-overview}}
% \gnote{This section is all about proving \Cref{thm:lift-proof-overview} from the proof overview. Maybe we should restate (should be able to use restatable environment) it to start the section?}
\label{sec:lift}

In this section, we prove \Cref{thm:lift-proof-overview}. We prove the two parts of the theorem separately. 

\subsection{Proof of the first part of \Cref{thm:lift-proof-overview}}
\begin{claim}[Formal version of the first part of \Cref{thm:lift-proof-overview}]
    \label{claim:circuit-upper-bound}
    For any $g:\bits^k\to\bits$ and constant $c>0$, the following holds. If for all distributions $H_g$ over $\bits^k$ of density $c$ we have $J_{H_g}(g,\lfrac1{2}-\gamma) \le r$, then for all $F \in \lift_n(g)$ and distributions $H_F$ over $\bits^{nk}$ of density $c$, there is a circuit of size
        \begin{equation*}
            O(r\cdot \tfrac{2^n }{n}+\tfrac{2^r}{r})
        \end{equation*}
        that agrees with $F$ on $\tfrac{1}{2} + \gamma$ of inputs in $H_F$. 
\end{claim}
The first part of \Cref{thm:lift-proof-overview} follows immediately from \Cref{claim:circuit-upper-bound} by observing that the circuit in the claim has size at most $O(r\cdot 2^n/n+2^{k}/k)$ which is $O(r\cdot 2^n/n)$  when $k\le n$.

To prove the claim, we consider distributions induced by applying $n$-bit functions $f_1,\ldots,f_k$ to the blocks of a string $\bX$ sampled from a distribution $H$ over $\bits^{nk}$:

\begin{definition}[Induced distributions]
\label{def:induce}
    For $f_1,\ldots,f_k:\bn\to\bits$ and distribution $H$ over $\bits^{nk}$, the \emph{induced distribution} of $H$ with respect to $f_1,\ldots,f_k$, $\mathrm{Ind}(H)$, is a distribution over $\bits^k$ defined by the following experiment:
    \begin{enumerate}
        \item sample $\bX\sim H$;
        \item output $(f_1(\bX^{(1)}),\ldots,f_k(\bX^{(k)}))$.
    \end{enumerate}
\end{definition}

We prove the following simple proposition about induced distributions.

\begin{proposition}[Density of distributions induced by balanced functions]
\label{prop:induce}
    For balanced $f_1,\ldots,f_k:\bn\to\bits$ and density $c$ distribution $H$ over $\bits^{nk}$, the induced distribution $\mathrm{Ind}(H)$ with respect to $f_1,\ldots,f_k$ has density $c$.
\end{proposition}

\begin{proof}
    Since $f_1,\ldots,f_k$ are balanced, any $y\in\bits^k$ can be obtained as $y=(f_1(X^{(1)}),\ldots,f_k(X^{(k)}))$ by at most $(2^{n-1})^k$ distinct strings $X\in\bits^{nk}$. Therefore, for all $y\in \bits^k$, we have
    \ifnum\focs=1
    \begin{align*}
        &\Prx_{\by\sim\mathrm{Ind}(H)}[\by =y]\\&=\Prx_{\bX\sim H}\big[y=(f_1(\bX^{(1)}),\ldots,f_k(\bX^{(k)}))\big]\\
        &\le \frac{|\{X\in \bits^{nk}\mid y=(f_1(X^{(1)}),\ldots,f_k(X^{(k)}))\}|}{c2^{nk}}\tag{$H$ has density $c$}\\
        &\le \frac{(2^{n-1})^k}{c2^{nk}}=\frac{1}{c2^{k}}\tag{$f_1,\ldots,f_k$ are balanced}
    \end{align*}
    \else
    \begin{align*}
        \Prx_{\by\sim\mathrm{Ind}(H)}[\by =y]&=\Prx_{\bX\sim H}\big[y=(f_1(\bX^{(1)}),\ldots,f_k(\bX^{(k)}))\big]\\
        &\le \frac{1}{c2^{nk}}\cdot |\{X\in \bits^{nk}\mid y=(f_1(X^{(1)}),\ldots,f_k(X^{(k)}))\}|\tag{$H$ has density $c$}\\
        &\le \frac{(2^{n-1})^k}{c2^{nk}}=\frac{1}{c2^{k}}\tag{$f_1,\ldots,f_k$ are balanced}
    \end{align*}
    \fi
    which completes the proof.
\end{proof}

We also use the following standard fact about the circuit complexity of Boolean functions.
\begin{fact}[Upper bound on the circuit size of Boolean functions \cite{Lup58}]
    \label{fact:lupanov}
    Every Boolean function on $n$ variables is computed by a circuit of size $O(2^n/n)$.
\end{fact}

Now we prove the main claim of this section.
\begin{proof}[Proof of \Cref{claim:circuit-upper-bound}]
    Let $F=g(f_1,\ldots,f_k)$ for balanced $f_1,\ldots,f_k$. Let $H_g=\mathrm{Ind}(H_F)$ be the induced distribution of $H_F$ with respect to $f_1,\ldots,f_k$. By \Cref{prop:induce}, $H_g$ has density $c$. Therefore, there is a junta over $r$ many variables $\{x_{i_1},\ldots,x_{i_r}\}$ which computes $g$ to accuracy $\frac{1}{2}+\gamma$. \Cref{fact:lupanov} implies that $g$ is computed to accuracy $\frac{1}{2}+\gamma$ by a circuit $C_g$ of size $O(2^{r}/r)$. Let $C_i$ be a circuit of size $O(2^n/n)$ which computes $f_i$ exactly. Then, we construct a circuit $C_F$ for $F$ defined by:
    $$
    C_F(X)\coloneqq C_g(C_{i_1}(X^{(i_1)}),\ldots,C_{i_{r}}(X^{(i_r)})).
    $$
    This circuit has size $r\cdot O(2^n/n)+O(2^r/r)$ and computes $F$ to accuracy $\frac{1}{2}+\gamma$ over $H_F$ as desired.
\end{proof}

\subsection{Proof of the second part of \Cref{thm:lift-proof-overview}}
We state the formal version of the second part of \Cref{thm:lift-proof-overview}.
\begin{theorem}[Formal version of the second part of \Cref{thm:lift-proof-overview}]
\label{thm:lift-proof-formal}
    For any $g:\bits^k\to \bits$ and $k\le 2^{n-1}$, if $J(g,\delta) \ge r_{\lar}$ then there is an $F \in \lift_n(g)$ for which all circuits of size
        \begin{equation*}
            \Omega(r_{\mathrm{large}} \cdot \tfrac{2^n }{n})
        \end{equation*}
        agree with $F$ on at most $1 - \delta/8$ fraction of inputs in $\bits^{nk}$.
\end{theorem}

\begin{proof}
    By \Cref{lem:lower-bound-from-soft}, there is an $F\in\lift_n(g)$ such that all circuits of size $O(\Tilde{J}(g,\delta/4)\cdot 2^n/n)$ agree with $F$ on at most $1-\delta/8$ fraction of inputs in $\bits^{nk}$. The proof is completed by observing that $J(g,\delta)\le O(\Tilde{J}(g,\delta/4))$ by \Cref{lem:connect-soft-and-hard-juntas}.
\end{proof}

\subsection{Proof of \Cref{lem:lower-bound-from-soft}}
\Cref{lem:each-circuit-covers-few-F-soft} implies that the number of functions needed to approximate each $F\in \text{Lift}_n(g)$ to accuracy $1-\delta$ is at least $2^{\Tilde{J}(g,2\delta)\cdot \Omega(2^n)}$ when $k\le 2^{n-1}$. The number of circuits of size $\Tilde{J}(g,2\delta)\cdot O(2^n/n)$ is at most $\left(\Tilde{J}(g,2\delta)\cdot 2^n/n\right)^{\Tilde{J}(g,2\delta)\cdot O(2^n/n)}\le 2^{\Tilde{J}(g,2\delta)\cdot O(2^n)}$ since $\Tilde{J}(g,2\delta)\le 2^{n-1}$. Therefore, there must exist some $F\in \text{Lift}_n(g)$ which cannot be approximated to accuracy $1-\delta$ by any circuit of size $O(\Tilde{J}(g,2\delta)\cdot 2^n/n)$.\hfill\qed

\subsection{Proof of \Cref{lem:each-circuit-covers-few-F-soft}}
As discussed in \Cref{sec:lower-bound-from-soft}, by the triangle inequality, it is sufficient to show that $\Pr[\dist(F,\bF') \leq 2\delta]\le 2^{-\tilde{J}(g,2\delta)\cdot \Omega(2^n-k)}$ for every $F=g(f_1,\ldots,f_k)$. First, we show that if $F$ and $F'=g({f_1}',\ldots,{f_k}')$ satisfy $\dist(F, F') \leq 2\delta$, then the sum of the correlations squared of $f_i$ and ${f_i}'$ is lower bounded by the $\delta$-error soft junta complexity of $g$:

\begin{proposition}[Soft junta complexity lower bounds the correlation of inner balanced functions]
\label{prop:soft-junta-error}
    For all $g:\bits^k\to\bits$, $F=g(f_1,\ldots,f_k)$ and $F'=g(f_1',\ldots,f_k')$, if $\dist(F,F')\le \delta$, then
    $$
    \sum_{i=1}^k\Ex_{\bx\sim\bits^n}[f_i(\bx){f_i}'(\bx)]^2\ge \Tilde{J}(g,\delta).
    $$
\end{proposition}
\begin{proof}
    Let $i\in [k]$ and consider the random variable $(\by_i,\by{_i}')=(f_i(\bx),f{_i}'(\bx))$ where $\bx\sim\bits^n$ is drawn uniformly at random and the random variable $(\bz_i,\bz{_i}')$ where $\bz_i$ sampled uniformly at random from $\bits$ and $\bz{_i}'$ is sampled independently from the distribution where $\bz{_i}'$ is set to $\bz_i$ with probability $(1+\alpha_i)/2$ for $\alpha_i\coloneqq \Ex_{\bx\sim\bn}[f_i(\bx){f_i}'(\bx)]$ and is otherwise set to $-\bz_i$. This corresponds to the joint distribution from \Cref{def:alpha-corr-error} where the correlation vector $\alpha\in [-1,1]^n$ is determined by $\Ex_{\bx\sim\bn}[f_i(\bx){f_i}'(\bx)]$ for $i=1,\ldots,n$.

    We claim that $(\by{_i},\by{_i}')$ and $(\bz{_i},\bz{_i}')$ are distributed the same: for all $(y,y')\in \bits\times\bits$, we have $\Pr[(\by{_i},\by{_i}')=(y,y')]=\Pr[(\bz{_i},\bz{_i}')=(y,y')]$. To see this, note that the pdfs of the random variables have four possible values and each bit is marginally uniform. Therefore, it is sufficient to show that the correlation of the two random variables is the same. And indeed by definition, we have 
    \ifnum\focs=1
    \begin{align*}
        \E[\bz{_i}\bz{_i}']&=\lfrac{1}{2}\E[\bz{_i}'\mid\bz{_i}=1]-\lfrac{1}{2}\E[\bz{_i}'\mid\bz{_i}=-1]\tag{$\bz{_i}$ is uniform random}\\
        &=\lfrac{1}{2}\Ex_{\bx\sim\bn}[f_i(\bx){f_i}'(\bx)]\\
        &\qquad-\lfrac{1}{2}\left(-\Ex_{\bx\sim\bn}[f_i(\bx){f_i}'(\bx)]\right)\tag{Definition of $(\bz_i,{\bz_i}')$}\\
        &=\Ex_{\bx\sim\bn}[f_i(\bx){f_i}'(\bx)]=\E[\by{_i}\by{_i}']\tag{Definition of $(\by{_i},\by{_i}')$}.
    \end{align*}
    \else
    \begin{align*}
        \E[\bz{_i}\bz{_i}']&=\lfrac{1}{2}\E[\bz{_i}'\mid\bz{_i}=1]-\lfrac{1}{2}\E[\bz{_i}'\mid\bz{_i}=-1]\tag{$\bz{_i}$ is uniform random}\\
        &=\lfrac{1}{2}\Ex_{\bx\sim\bn}[f_i(\bx){f_i}'(\bx)]-\lfrac{1}{2}\left(-\Ex_{\bx\sim\bn}[f_i(\bx){f_i}'(\bx)]\right)\tag{Definition of $(\bz_i,{\bz_i}')$}\\
        &=\Ex_{\bx\sim\bn}[f_i(\bx){f_i}'(\bx)]=\E[\by{_i}\by{_i}']\tag{Definition of $(\by{_i},\by{_i}')$}.
    \end{align*}
    \fi
    It follows that
    \begin{align*}
        \delta&\ge\Prx_{\bX\sim\bits^{nk}}[F(\bX)\neq F'(\bX)]\tag{Assumption}\\
        % &=\Prx_{\bx\sim\bits^{nk}}\left[g(f_1(\bx^{(1)}),\ldots, f_k(\bx^{(k)}))\neq g({f_1}'(\bx^{(1)}),\ldots, {f_k}'(\bx^{(k)}))\right]\\
        &=\Pr\left[g(\by_1,\ldots, \by_k)\neq g(\by{_1}',\ldots, \by{_k}')\right]\tag{Definition of $\by_i,\by{_i}'$}\\
        &=\Pr\left[g(\bz_1,\ldots, \bz_k)\neq g(\bz{_1}',\ldots, \bz{_k}')\right]\tag{$(\by{_i},\by{_i}')$ and $(\bz{_i},\bz{_i}')$ are distributed the same}\\
        &=\Pr[g(\bz)\neq g(\bz')].
    \end{align*}
    Therefore, the $\alpha$-correlated error of $g$ is at most $\delta$. By the definition of $\Tilde{J}(g,\delta)$, we get
    $$
    \Tilde{J}(g,\delta)\le \sum_{i=1}^k\alpha_i^2=\sum_{i=1}^k\Ex_{\bx\sim\bits^n}[f_i(\bx){f_i}'(\bx)]^2
    $$
    which completes the proof.
\end{proof}
Now, letting ${\boldf_i}':\bits^n\to\bits$ be a uniformly random balanced function (chosen independently for each $i$), and $\balpha_i \coloneqq \Ex_{\bx\sim\bits^n}[f_i(\bx)\boldf{_i}'(\bx)]$, we have:
\begin{align*}
    \Pr[\dist(F,\bF') \leq 2\delta] &\le \Prx_{\balpha_1, \dots, \balpha_k}\left[\sum_{i=1}^k \balpha_i^2 \ge \Tilde{J}(g,2\delta)\right]\tag{\Cref{prop:soft-junta-error}}\\
    &\le \exp(-\Omega(\Tilde{J}(g,2\delta)\cdot 2^n-k))\tag{\Cref{lem:concentration-alphas-proof-overview}}
\end{align*}
which completes the proof.\hfill\qed

\subsection{Proof of \Cref{lem:concentration-alphas-proof-overview}}
% Now we move on to proving the main concentration inequality (\Cref{lem:concentration-alphas-proof-overview}). 
The proof of this lemma uses basic facts about the sums of sub-exponential random variables. Before proving the lemma, we state the requisite definitions and facts that we use.

\begin{definition}[Sub-Gaussian random variable]
    A random variable $\bX$ is \emph{sub-Gaussian} if there is some $t>0$ for which 
    $$
    \E[\exp(\bX^2/t^2)]\le 2
    $$
    The sub-Gaussian norm of $\bX$ is defined to be $\inf\{t>0:\E[\exp(\bX^2/t^2)]\le 2\}$.
\end{definition}
As suggested by the name, the sub-Gaussian norm is a norm of the vector space over $\R$ of sub-Gaussian random variables. In particular, the sum of two sub-Gaussian random variables $\bX,\bY$ (not necessarily independent) is itself a sub-Gaussian random variable and the sub-Gaussian norm of $\bX+\bY$ is bounded by the sum of the sub-Gaussian norms of $\bX$ and $\bY$.

A symmetric Bernoulli random variable is one that is uniform on $\pm 1$. The sum of independent symmetric Bernoulli random variables is sub-Gaussian:
\begin{fact}[Sum of symmetric Bernoullis is sub-Gaussian]
\label{fact:sum-symmetric-bernoullis}
    The sum of $N$ independent symmetric Bernoulli random variables is sub-Gaussian with sub-Gaussian norm $O(\sqrt{N})$ and variance $N$. 
\end{fact}

To see that the variance is $N$, note that, by independence, the variance is the sum of the variances of each independent symmetric Bernoulli random variable, and since a symmetric Bernoulli random variable is uniform on $\pm 1$, its variance is $1$.

% \begin{fact}[Sum of sub-Gaussian random variables]
% \label{fact:sum-subgaussians}
%     Let $\bX$ and $\bY$ be sub-Gaussian random variables with sub-Gaussian norms $M_{X}$ and $M_{Y}$, respectively. Then, $\bX+\bY$ is sub-Gaussian with sub-Gaussian norm $O(M_X+M_Y)$. 
% \end{fact}

\begin{fact}[Dominating sub-Gaussian random variables]
\label{fact:dominating}
    If $\bY$ is a sub-Gaussian random variable with sub-Gaussian norm $M_Y$ and $\bX$ is a random variable such that $|\bX|\le |\bY|$ with probability $1$, then $\bX$ is sub-Gaussian with sub-Gaussian norm $\le M_Y$.  
\end{fact}

\begin{definition}[Sub-exponential random variable]
    A random variable $\bX$ is \emph{sub-exponential} if there is a $t>0$ such that 
    $$
    \E[\exp(|\bX|/t)]\le 2.
    $$
    The sub-exponential norm of $\bX$ is defined to be $\inf\{t>0:\E[\exp(|\bX|/t)]\le 2\}$. Alternatively, a random variable $\bX$ is sub-exponential if and only if $\sqrt{|\bX|}$ is sub-Gaussian. If $\sqrt{|\bX|}$ has sub-Gaussian norm $M$, then $\bX$ has sub-exponential norm $M^2$. 
\end{definition}

One of the basic facts about sub-exponential random variables is Bernstein's inequality which bounds the tails of sums of independent, sub-exponential random variables. See the textbook by Vershynin \cite[Theorem 2.8.1]{Ver18}) for an overview of this inequality along with its proof.

\begin{fact}[Bernstein's inequality \cite{Ber46}]
\label{fact:bernstein}
    Let $\bZ_1,\ldots,\bZ_k$ be independent, sub-exponential random variables with mean $0$ and sub-exponential norm $M$. Then for every $t\ge 0$, we have
    $$
    \Pr\left[\sum_{i=1}^k \bZ_i\ge t\right]\le \exp\left(-\Omega\left(t/M-k\right)\right).
    $$
\end{fact}

% Indeed, the proof of Bernstein's inequality (see for example \cite[Theorem 2.8.1]{Ver18}) shows that for all $\lambda$ such that $|\lambda|\le O(1/M)$ and $t\ge 0$,\gnote{Is this description of the proof just for the reader's enjoyment or something we need? I found it a tad distracting}
% $$
% \Pr\left[\sum_{i=1}^k \bZ_i\ge t\right]\le \exp(-\lambda t+\Theta(\lambda^2 kM^2)).
% $$
% By choosing $\lambda=\Theta(1/M)$, we recover \Cref{fact:bernstein}.

In order to apply Bernstein's inequality, we need to center a sub-exponential random variable so it has mean $0$. It's straightforward to show that a centered sub-exponential random variable is still sub-exponential and has the same sub-exponential norm (up to constants):

\begin{fact}[Centering a sub-exponential random variable]
    \label{fact:centering}
    If $\bX$ is a sub-exponential random variable with mean $\mu$ and sub-exponential norm $M$, then $\bX-\mu$ is a sub-exponential random variable with sub-exponential norm $O(M)$. 
\end{fact}
% \begin{definition}
%     We define the random variable $\bW$ as $\sum_{i=1}^N \bW_i$ where the $\bW_i$ are independent, symmetric Bernoulli random variables supported on $\bits$. 
% \end{definition}

\paragraph{An equivalent way of sampling a random balanced function.}
We consider a random function $f:\bits^n\to\bits$ as a uniform random string from $\bits^N$ corresponding to $f$'s truth table in lexicographic order where $N=2^n$. Therefore, a uniform random \textit{balanced} function will correspond to a uniform random string from $\bits^N$ of Hamming weight $N/2$. An equivalent way of sampling a uniform random balanced function is to first sample a uniform random function and then correct it to be balanced. The following proposition formalizes this equivalence.

\begin{proposition}
\label{prop:sampling-from-middle-layer}
    Let $N\in \N$ be even. Let ${\bW}$ be a random string obtained by the following
    \begin{itemize}
        \item sample $\bU$ uniformly at random from $\bits^N$ and for $\bell=\sum_{i=1}^N \bU_i$:
        \begin{itemize}
            \item if $\bell >0$, then {uniformly at random} select $\bell/2$ many $+1$ coordinates in $\bU$ and flip them to $-1$ to form $\bW$;
            \item if $\bell<0$, then {uniformly at random} select $|\bell|/2$ many $-1$ coordinates in $\bU$ and flip them to $+1$ to form $\bW$.
        \end{itemize}
    \end{itemize}
    Then, ${\bW}$ is distributed uniformly at random among strings $\bits^N$ of Hamming weight $N/2$.
\end{proposition}

\begin{proof}
    By construction, $\bW$ is always a string of Hamming weight $N/2$. Moreover, the sampling process is invariant under any permutation of the coordinates. Therefore, all strings of Hamming weight $N/2$ are equally likely.
    % \gnote{Can we just do your first sentence, and then end this proof with something along the lines of: ``Moreover, the sampling process is invariant to any permutation of the coordinates, and so all such strings are equally likely."} Moreover, each string $x\in\bits^N$ of Hamming weight $N/2$ has an equal likelihood of being sampled under this process. Indeed, for each choice of $\bell$, there is exactly one way to flip to coordinates in $\bU$ to form $x$\gnote{I don't understand this. Once we fix $\bU$, not all $x$ are necessarily possible. Before we fix $\bU$, there are multiple ways to form each $x$}, and each choice of $|\bell|/2$ flips is equally likely. 
\end{proof}

% \begin{fact}
%     Let $\bX$ be a sub-Gaussian random variable with sub-Gaussian norm $M_X$. Then for any $r\in\R$, $r\bX$ is sub-Gaussian with sub-Gaussian norm $r M_X$. 
% \end{fact}
% \begin{lemma}
% \label{lem:subexp}
%     The random variables $\balpha_i^2$ from \Cref{lem:concentration-alphas} are sub-exponential with sub-exponential norm $O(1/N)$ and mean $O(1/N)$ where $N=2^n$. 
% \end{lemma}

{
\begin{corollary}[Each $\balpha_i$ is the sum of sub-Gaussians]
\label{cor:sub-gaussian}
    For each $i \in [k]$ let $\balpha_i$ be as defined in \Cref{lem:concentration-alphas-proof-overview}. It can be coupled to $(\by_i, \bz_i)$, each having mean $0$, variance $1/2^n$, and sub-Gaussian norm at most $O(1/\sqrt{2^n})$, so that
    \begin{equation*}
        \abs*{\balpha_i} \leq \abs*{\by_i} + \abs*{\bz_i} \quad\quad\text{with probability $1$.}
    \end{equation*}
\end{corollary}
}

\begin{proof}
    Let $N=2^n$ and $W,\bW'\in \bits^N$ correspond to the truth tables for $f_i,{\boldf_i}'$, respectively. Note that both $W$ and $\bW'$ are strings of Hamming weight $N/2$ and $\bW'$ is uniform random among all such strings. We can rewrite $\balpha_i$ as $\balpha_i=\tfrac{{W}\cdot \bW'}{N}$. Let $\bW'$ be obtained by the process in \Cref{prop:sampling-from-middle-layer} and let $\bU'$ be the intermediate random variable which is independently distributed uniform at random in $\bits^N$. Since $\bW'$ is distributed uniformly at random among strings $\bits^N$ of Hamming weight $N/2$, we can write
    \begin{align*}
        \left|\balpha\right| &=\left|\frac{{W}\cdot \bW'}{N}\right|\tag{\Cref{prop:sampling-from-middle-layer}}\\
        &\le \frac{1}{N}\left(|W\cdot \bU'|+\left|\sum_{i=1}^N \bU{_i}'\right|\right).
    \end{align*}
    The last inequality follows from the fact that $|{W}\cdot \bW'|$ is at most $|W\cdot\bU'|$ plus the total number of coordinate changes made to $\bU'$ to form $\bW'$ which is bounded by $\left|\sum_{i=1}^N \bU{_i}'\right|$. Both $W\cdot \bU'$ and $\sum_{i=1}^N \bU{_i}'$ are distributed as the sum of $N$ independent symmetric Bernoulli random variables. Therefore, $W\cdot \bU'$ and $\sum_{i=1}^N \bU{_i}'$ are random variables with mean $0$, variance $N$ and sub-Gaussian norm $O(\sqrt{N})$. Finally, multiplying each random variable by $1/N$ makes the variance $1/N$ and the sub-Gaussian norm $O(1/\sqrt{N})$ as desired.
\end{proof}

\begin{proof}[Proof of \Cref{lem:concentration-alphas-proof-overview}]
    Since the absolute value of a sub-Gaussian random variable is also sub-Gaussian, \Cref{cor:sub-gaussian} implies that $|\balpha_i|$ dominated by the sum of sub-Gaussian random variables with sub-Gaussian norm $O(1/\sqrt{2^n})$. Therefore, by the property of dominating sub-Gaussian random variables (\Cref{fact:dominating}), $|\balpha_i|$ is sub-Gaussian with sub-Gaussian norm $O(1/\sqrt{2^n})$. In particular, $\balpha_i^2$ is sub-exponential with sub-exponential norm $O(1/2^n)$. To apply Bernstein's inequality (\Cref{fact:bernstein}), we first need to center the $\balpha_i^2$'s so that they have mean $0$. If $\mu$ denotes the mean of $\balpha^2$, then the random variable $\balpha_i^2-\mu$ has mean $0$ and sub-exponential norm $O(1/2^n)$. Moreover, the mean $\mu$ is at most $O(1/2^n)$ since
    \begin{align*}
        \E[\balpha_i^2]&\le \E[(|\by_i|+|\bz_i|)^2]\tag{\Cref{cor:sub-gaussian}}\\
        &\le 2\left(\E[\by_i^2]+\E[\bz_i^2]\right)\tag{Cauchy-Schwarz inequality}\\
        &= 2\left(\Var[\by_i]+\Var[\by_i]\right)\tag{Definition of variance and $\by_i,\bz_i$ have mean $0$}\\
        &\le O(1/2^n).\tag{\Cref{cor:sub-gaussian}}
    \end{align*}
    Therefore, for all $t'>0$, we get
    \begin{align*}
        \exp(-\Omega(t'\cdot 2^n-k))&\ge \Pr\left[\sum_{i=1}^k \left(\balpha_i^2-\mu\right) \ge t'\right]\tag{\Cref{fact:bernstein}}\\
        &=\Pr\left[\sum_{i=1}^k \balpha_i^2 \ge t'+k\mu\right]\\
        &\ge \Pr\left[\sum_{i=1}^k \balpha_i^2 \ge t'+O(k/2^n)\right]\tag{$\mu\le O(1/2^n)$}.
    \end{align*}
    By choosing $t'=t-\Theta(k/2^n)$, we get the desired result. 
\end{proof}

\subsection{Proof of \Cref{lem:connect-soft-and-hard-juntas}}
Let $\alpha \in [-1,1]^k$ with $\ltwo{\alpha}^2 = \tilde{J}(g,\delta)$ be the correlation vector for which the $\alpha$-correlated-error of $g$ is at most $\delta$. Let $\bz_i$ be drawn independently from $\Ber(\alpha_i^2)$ for each $i \in [k]$. In expectation, $\lone{\bz} =\ltwo{\alpha}^2 = \tilde{J}(g,\delta)$, so by Markov's inequality, with probability at least $1/2$, $\lone{\bz} \leq 2\cdot \tilde{J}(g,\delta)$. 
    
Next, we know that the expected $\bz$-correlated-error of $g$ is at most $2\delta$. Conditioning on a probability-$(1/2)$ event can at most double that expectation, so the expected $\bz$-correlated-error of $g$ conditioned on $\lone{\bz} \leq 2\cdot \tilde{J}(g,\delta)$ is at most $4\delta$. In particular, this means there is a single choice of $z$ for which $\lone{z} \leq 2\cdot \tilde{J}(g,\delta)$ and the $z$-correlated-error of $g$ is at most $4\delta$. Therefore,
\begin{equation*}
    J(g,4\delta) \leq 2\cdot \tilde{J}(g,\delta)
\end{equation*}
which completes the proof.\hfill\qed

\subsection{Proof of \Cref{claim:random-junta-double-error-proof-overview}}

We start by defining a useful quantity: correlated-variance.
\begin{definition}[Correlated-variance]
    \label{def:correlated-variance}
    For any $\alpha \in [-1,1]^k$, the \emph{$\alpha$-correlated-variance} of $g$ is defined for $\bx,\by$ distributed according to $\mathcal{D}(\alpha)$ as:
    \begin{equation*}
        \Ex_{\bx}\bracket*{\Varx_{\by \mid \bx}\bracket*{g(\by)}} = \Ex_{\bx}\bracket*{1 - \Ex_{\by \mid \bx}\bracket*{g(\by)}^2}
    \end{equation*}
\end{definition}
The reason $\alpha$-correlated variance is useful is because it has two key properties.

\begin{claim}[Properties of correlated-variance]
    \label{claim:corr-var-properties}
    There is a notion of correlated-variance (defined in \Cref{def:correlated-variance}) satisfying, for any $g:\bits^k \to \bits$,
    \begin{enumerate}
        \item The $\alpha$-correlated-variance of $g$ is between double the $\alpha$-correlated-error of $g$ and quadruple the $\alpha$-correlated-error of $g$.
        \item For any $\bz$ supported on $[-1,1]^k$ drawn from a product distribution with $\Ex[\bz_i^2] = \alpha_i^2$, the expected $\bz$-correlated-variance of $g$ is equal to the $\alpha$-correlated-variance of $g$. 
    \end{enumerate}
\end{claim}
\begin{proof} 
    Let $\bx,\by$ be drawn from the joint distribution $\mathcal{D}(\alpha)$ defined in \Cref{def:alpha-corr-error}. For the first property, we note the $\alpha$-correlated-error of $g$ can be written as
    \begin{equation*}
        \min_{h:\bits^k \to \bits} \set*{\Ex_{\bx}\bracket*{\Prx_{\by \mid \bx}[g(\by) \neq h(\bx)]}}.
    \end{equation*}
    To minimize the above, $h(x)$ should be set to $\sign(\Ex[g(\by) \mid \bx])$, giving that the $\alpha$-correlated error of $g$ is
    \begin{equation*}
        \Ex_{\bx}\bracket*{\frac{1}{2} - \frac{1}{2} \cdot\abs*{\Ex_{\by \mid\bx}[g(\by) ]}}.
    \end{equation*}
    Let $f_1(x) = (1 - |x|)/2$ and $f_2 = 1-x^2$. The first property follows from the sandwiching $2 f_1(x) \leq f_2(x) \leq 4f_1(x)$ which holds for all $x \in [-1,1]$.

    The proof of the second property uses basic Fourier analysis. Recall that every function $g:\bits^n\to\bits$ has a Fourier expansion which can be written as
    $$
    g(x)=\sum_{S\sse [n]}\hat{g}(S)\prod_{i\in S}x_i
    $$
    where $\hat{g}(S)\in \R$. Using this, we first observe that for any $x\in\bits^k$, we have
    \begin{align*}
        \Ex_{\by\mid x}[g(\by)]&=\Ex_{\by\mid x}\bracket*{\sum_{S\sse [n]}\hat{g}(S)\prod_{i\in S}\by_i}\tag{Fourier expansion of $g$}\\
        &=\sum_{S\sse [n]}\hat{g}(S)\prod_{i\in S}\Ex_{\by\mid x}[\by_i]\tag{Linearity of expectation and the independence of $\by_i$}\\
        &=\sum_{S\sse [n]}\hat{g}(S)\prod_{i\in S}x_i\alpha_i\tag{Definition of $\by_i$}.
    \end{align*}
    It follows that
    \ifnum\focs=1
    \begin{align*}
        &\Ex_{\bx}\bracket*{\Ex_{\by \mid\bx }\bracket*{g(\by)}^2}\\
        &=\Ex_{\bx}\bracket*{\paren*{\sum_{S\sse [n]}\hat{g}(S)\prod_{i\in S}\bx_i\alpha_i}^2}\\
        &=\sum_{S_1,S_2\sse [n]}\hat{g}(S_1)\hat{g}(S_2)\Ex_{\bx}\bracket*{\prod_{i\in S_1}\bx_i\alpha_i\prod_{i\in S_2}\bx_i\alpha_i}\\
        &=\sum_{S_1,S_2\sse [n]}\Bigg(\hat{g}(S_1)\hat{g}(S_2)\cdot \\
        &\qquad \Ex_{\bx}\bracket*{\prod_{i\in S_1\cap S_2}(\bx_i\alpha_i)^2\prod_{i\in S_1\Delta S_2}\bx_i\alpha_i}\Bigg)\\
        &=\sum_{S_1,S_2\sse [n]}\hat{g}(S_1)\hat{g}(S_2)\prod_{i\in S_1\cap S_2}\alpha_i^2\prod_{i\in S_1\Delta S_2}\alpha_i\Ex_{\bx}[\bx_i]\tag{Independence of $\bx_i$ and $\bx_i^2=1$}.
    \end{align*}
    \else
    \begin{align*}
        \Ex_{\bx}\bracket*{\Ex_{\by \mid\bx }\bracket*{g(\by)}^2}&=\Ex_{\bx}\bracket*{\paren*{\sum_{S\sse [n]}\hat{g}(S)\prod_{i\in S}\bx_i\alpha_i}^2}\\
        &=\sum_{S_1,S_2\sse [n]}\hat{g}(S_1)\hat{g}(S_2)\Ex_{\bx}\bracket*{\prod_{i\in S_1}\bx_i\alpha_i\prod_{i\in S_2}\bx_i\alpha_i}\\
        &=\sum_{S_1,S_2\sse [n]}\hat{g}(S_1)\hat{g}(S_2)\Ex_{\bx}\bracket*{\prod_{i\in S_1\cap S_2}(\bx_i\alpha_i)^2\prod_{i\in S_1\Delta S_2}\bx_i\alpha_i}\\
        &=\sum_{S_1,S_2\sse [n]}\hat{g}(S_1)\hat{g}(S_2)\prod_{i\in S_1\cap S_2}\alpha_i^2\prod_{i\in S_1\Delta S_2}\alpha_i\Ex_{\bx}[\bx_i]\tag{Independence of $\bx_i$ and $\bx_i^2=1$}.
    \end{align*}
    \fi
    In the above sum, if $S_1\neq S_2$ then $S_1\Delta S_2$ is nonempty and so the entire term evaluates to $0$ because $\Ex_{\bx}[\bx_i]=0$. Therefore, we can rewrite 
    \begin{equation}
    \label{eq:correlated-variance}
    \Ex_{\bx}\bracket*{\Ex_{\by \mid \bx}\bracket*{g(\by)}^2}=\sum_{S\sse [n]} \hat{g}(S)^2\prod_{i\in S}\alpha_i^2.
    \end{equation}
    In particular, if $\E[\bz_i^2]=\alpha_i^2$, then for $\bx,\by$ drawn from the distribution $\mathcal{D}(\bz)$, we have
    \begin{align*}
        \Ex_{\bx}\bracket*{\Ex_{\by \mid \bx}\bracket*{g(\by)}^2}&=\E\bracket*{\sum_{S\sse [n]} \hat{g}(S)^2\prod_{i\in S}\bz_i^2}\tag{\Cref{eq:correlated-variance}}\\
        &=\sum_{S\sse [n]} \hat{g}(S)^2\prod_{i\in S}\E[\bz_i^2]\tag{Linearity of expectation and independence of the $\bz_i$}\\
        &=\sum_{S\sse [n]} \hat{g}(S)^2\prod_{i\in S}\alpha_i^2\tag{Assumption that $\E[\bz_i^2]=\alpha_i^2$}\\
        &=\Ex_{\bx}\bracket*{\Ex_{\by \mid \bx}\bracket*{g(\by)}^2}\tag{\Cref{eq:correlated-variance}}
    \end{align*}
    where in the last equation $\bx,\by$ are distributed according to $\mathcal{D}(\alpha)$. This shows that the expected $\bz$-correlated variance of $g$ is equal to the $\alpha$-correlated-variance of $g$ which completes the proof.
\end{proof}

\begin{proof}[Proof of \Cref{claim:random-junta-double-error-proof-overview}]
    Let $\delta$ be the $\alpha$-correlated-error of $g$. Then, by property 1 of \Cref{claim:corr-var-properties}, the $\alpha$-correlated-variance of $g$ is at most $\delta/2$. By property 2, the expected $\bz$-correlated-variance is therefore also at most $\delta/2$. Using the other side of property 1 gives that the expected $\bz$-correlated-error is at most $2\delta$.
\end{proof}

\section{Proof of \Cref{thm:impagliazzo intro}}

\begin{theorem}[Formal statement of \Cref{thm:impagliazzo intro}]
    \label{thm:impagliazzo formal}
    For any $\gamma>0$ and $s\ge \Omega(1/\gamma^2)$, there is a function $F:\bits^N\to\bits$ such that 
    \begin{enumerate}
        \item {$F$ is mildly hard for size-$s$ circuits}: every circuit of size $s$ agrees with $F$ on at most $99\%$ of inputs in $\bits^N$.
        \item {For every hardcore distribution, $F$ is mildly correlated with a small circuit}: for all constant density distributions $H$ over $\bits^N$, there is a circuit of size $O(s\gamma^2)$ which computes $F$ with probability $\tfrac{1}{2}+\gamma$ over $H$.
    \end{enumerate}
\end{theorem}

\begin{proof}
There is an $F\in \text{Lift}_n(\MAJ_k)$ such that all circuits of size $O(k\cdot\tfrac{2^n}{n})$ agree with $F$ on at most $0.99$ fraction of inputs from $\bits^{nk}$. This is because $J(\MAJ_k,0.2)\ge k/2$ by \Cref{claim:junta-size-loss-proof-overview} and so \Cref{thm:lift-proof-formal} implies there is an $F$ for which all circuits of size $O(k\cdot\tfrac{2^n}{n})$ agree with it on at most $1-0.2/8\le 0.99$ fraction of inputs from $\bits^{nk}$. This $F$ satisfies the first part of the theorem statement. 

For the second part, let $H$ be a distribution of constant density over $\bits^{nk}$. By \Cref{claim:circuit-upper-bound}, there is a circuit of size $O(2^n/n)$ that computes $F$ to accuracy $1/2+\Omega(1/\sqrt{k})$ over $H$. This is because by \Cref{claim:maj-easy}, $\MAJ_k$ can be computed to accuracy $1/2+\Omega(1/\sqrt{k})$ over constant density distributions by a $1$-junta. 

Therefore, given a parameters $\gamma,s$, we choose $n$ and $k\le 2^{n-1}$ so that $s=\Theta(k\cdot\tfrac{2^n}{n})$, and $\gamma=\Theta(1/\sqrt{k})$. Such a choice of $k\le 2^{n-1}$ exists by our assumption that $\gamma \ge \Omega(1/\sqrt{s})\ge \Omega(1/2^n)$. By the above two paragraphs, the theorem holds for this choice of parameters. 
\end{proof}

%% file: BoostingProofs.tex
\newcommand{\uint}{u}
% for lower interval
\newcommand{\lint}{-u}

\newcommand{\uintval}{O(\sqrt{k \log{k}\kappa})}
\newcommand{\lintval}{-\uintval}
\newcommand{\advantage}{\frac{1}{\sqrt{k \log k}\kappa^{7/2}}}
\newcommand{\advantageline}{1/\sqrt{k \log k \kappa^7}}
\newcommand{\pconcentrate}{\frac{1}{k^2}}

% weak correlation
\newcommand{\wcorr}{\Omega(1/\kappa^3)}

\newcommand{\bitsN}{\bits^N}
\newcommand{\bitsnk}{\bits^{k n}}
\newcommand{\bitsmk}{\bitsnk}

\newcommand{\bitsmklog}{\bits^{k \log m}}
\newcommand{\bitsmlog}{\bits^{\log m}}

\newcommand{\bitsm}{\bits^{n}}

\section{Proof of \Cref{thm:smooth intro}}

This section will give the proof (up to a log factor) of \Cref{thm:smooth intro}. We will allow the user to specify the desired weak learner's sample complexity $m$ and weak learning parameter $\gamma$. 

\begin{theorem}[\Cref{thm:smooth intro} formalized]
%\gnote{This should reference the lifted class $\lift_n(\Maj_k)$ defined in \Cref{def:lifted}, so we will need to specific $n$. This might be a bit of work, but I think it would be nice if everything is parameterized by $n$ and $k$ and we define $m$ and $\gamma$ as functions of them. So something like: For any $n \geq \Omega(\log(k))\ldots$. Then, after the statement of \Cref{thm:boosting-formal}, we can say, setting $m = 2^n$ and $\gamma = \ldots$, this implies ...
%This way, we can replace all the $\bits^{\log m}$ with $\bits^n$ which is nicer.
%}
    \label{thm:boosting-formal}
     For any $k$ and $n \geq \Omega_{\kappa}(\log k)$ let $\mcC \coloneqq \lift_n(\Maj_k)$. Then,
%    \begin{equation*}
%        k \coloneqq O \left (\frac{\log 1/\gamma}{\gamma^2} \right ).
%    \end{equation*}
    
    \begin{itemize}
        \item \Cref{lem:boosting-upper}: There is an $O(2^n)$-sample learner which, for any distribution that is $\kappa$-smooth on $\bitsnk$, achieves advantage $\Omega\paren*{\advantageline}$ with high probability for the concept class $\mcC$.
        \item \Cref{lem:boosting-lower}: Learning $\mcC$ to accuracy 0.99 w.r.t.~the uniform distribution requires $\Omega(k2^n)$ samples.
    \end{itemize}
Setting $m = O(2^n)$ and $\gamma = \Omega(\advantageline)$, this implies that for $\mcC = \lift_n(\Maj_k)$, there exists a weak learner that achieves $\gamma$ advantage with high probability using $m$ samples but any algorithm that learns $\mcC$ to accuracy 0.99 must use $\tilde{\Omega}(m/\gamma^2)$ samples.
\end{theorem}

We start by proving the lower bound since it follows directly from our results on the tightness of the hardcore theorem. In particular, combining \Cref{claim:junta-size-loss-proof-overview} and \Cref{lem:connect-soft-and-hard-juntas,lem:each-circuit-covers-few-F-soft}, we immediately obtain the following.
\begin{corollary}
    \label{cor:maj-covering}
    For any $n \in \N$ and $h: \bits^{nk} \to \bits$,
    \begin{equation*}
        \Prx_{\bF \sim \lift_n(\Maj_k)}[\dist(h, \bF) \leq 0.01] \leq 2^{-\Omega(k \cdot (2^n - k))}.
    \end{equation*}
\end{corollary}
We show how the lower bound of \Cref{thm:boosting-formal} follows easily from \Cref{cor:maj-covering}.

\begin{lemma}[Lower bound of \Cref{thm:boosting-formal}, Restatement of \Cref{lem:learning-lb}]
    \label{lem:boosting-lower} For any $n \geq \Omega(\log k)$, any algorithm that learns $\lift_n(\Maj_k)$ to accuracy $0.99$ with success probability $0.01$ over the uniform distribution must use $m \geq \Omega(2^n k)$ samples.
\end{lemma}

\begin{proof}
    By the easy direction of Yao's lemma, it suffices to show that for any \emph{deterministic} learner $\mcA$, there is a distribution of concepts $\bF$ supported on $\lift_n(\Maj_k)$ for which the probability that $\mcA$ successfully learns $\bF$ is less than $0.01$. We'll set this distribution to the uniform distribution on $\lift_n(\Maj_k)$. Therefore, for $\bS$ denoting the sample of $m$ points $\mcA$ receives, it suffices to show that
    \begin{equation*}
        \Ex_{\bF \sim \lift_n(\Maj_k)}\bracket*{\Prx_{\bS}\bracket*{\dist(\mcA(\bS), \bF) \leq 0.01}} < 0.01.
    \end{equation*}
    For the sample $\bS = \bracket*{(\bx_1, f(\bx_1)), \ldots, (\bx_m, f(\bx_m))}$, we denote the unlabeled portion of the sample and labeled portion as
    \begin{equation*}
        \bS_x \coloneqq \bracket*{\bx_1, \ldots, \bx_m} \quad\quad\text{and} \quad\quad \bS_y \coloneqq\bracket*{f(\bx_1), \ldots, f(\bx_m)}.
    \end{equation*}
    The key observation is that the unlabeled portion of the sample is \emph{independent} of $\bF$. Therefore, we can rewrite
    \ifnum\focs=1
    for $\bF\sim \lift_n(\Maj_k)$:
    \begin{align*}
        &\Ex_{\bF}\bracket*{\Prx_{\bS}\bracket*{\dist(\mcA(\bS), \bF) \leq 0.01}} \\&= \Ex_{\bS_x}\bracket*{\Ex_{\bF} \bracket*{\Prx_{\bS_y}\bracket*{\dist(\mcA(\bS_x, \bS_y), \bF) \leq 0.01}}}\\
        &\leq \sup_{S_x}\paren*{\Ex_{\bF} \bracket*{\Prx_{\bS_y}\bracket*{\dist(\mcA(S_x, \bS_y), \bF) \leq 0.01}}}
    \end{align*}
    \else
    \begin{align*}
        \Ex_{\bF \sim \lift_n(\Maj_k)}\bracket*{\Prx_{\bS}\bracket*{\dist(\mcA(\bS), \bF) \leq 0.01}} &= \Ex_{\bS_x}\bracket*{\Ex_{\bF \sim \lift_n(\Maj_k)} \bracket*{\Prx_{\bS_y}\bracket*{\dist(\mcA(\bS_x, \bS_y), \bF) \leq 0.01}}}\\
        &\leq \sup_{S_x}\paren*{\Ex_{\bF \sim \lift_n(\Maj_k)} \bracket*{\Prx_{\bS_y}\bracket*{\dist(\mcA(S_x, \bS_y), \bF) \leq 0.01}}}
    \end{align*}
    \fi
    Since $\mcA$ is deterministic and $S_y$ contains only $m$ bits of information, after fixing $S_x$ there are only $2^m$ possible hypotheses that $\mcA$ can output. Therefore, by union bound, the above is at most
    \ifnum\focs=1
    \begin{align*}
        2^m \cdot \sup_{h} & \paren*{\Prx_{\bF \sim \lift_n(\Maj_k)}[\dist(h, \bF) \leq 0.01]} \\
        &\leq 2^m \cdot 2^{-\Omega(k \cdot (2^n - k))}
    \end{align*}
    where the supremum is taken over all $h:\bits^{nk} \to \bits$ and
    \else
    \begin{equation*}
        2^m \cdot \sup_{h:\bits^{nk} \to \bits} \paren*{\Prx_{\bF \sim \lift_n(\Maj_k)}[\dist(h, \bF) \leq 0.01]} \leq 2^m \cdot 2^{-\Omega(k \cdot (2^n - k))}
    \end{equation*}
    \fi
    where the second inequality uses \Cref{cor:maj-covering}. Therefore, for $\mcA$ to successfully learn, it must be that $m \geq \Omega(k \cdot (2^n - k)) = \Omega(2^n k)$ using the fact that $n \geq \Omega(\log k)$.
\end{proof}

The rest of this section will be devoted to proving the upper bound of \Cref{thm:boosting-formal}.

\subsection{Proof overview of the upper bound of~\Cref{thm:boosting-formal}}
\label{sec:boosting-upper-bound}

\begin{figure}[t] 

  \captionsetup{width=.9\linewidth}

    \begin{tcolorbox}[colback = white,arc=1mm, boxrule=0.25mm]
    \vspace{2pt} 
    
    \textbf{Initialization:} Draw a random sample of $2m$ many points from $\mathcal{D}$ and split it into a size-$m$ training set $\boldsymbol{S_\mathrm{train}}$ and a size-$m$ validation set $\boldsymbol{S_\mathrm{val}}$. Since we will mostly prove properties relating to $\boldsymbol{S_\mathrm{train}}$, we will use the simpler notation $\bS \coloneqq \boldsymbol{S_\mathrm{train}}$ when the sample used is clear from context.
    
    Initialize $g_{1,\bS}, \ldots, g_{k,\bS}:\bitsm \to \{-1,0,1\}$ each as the constant zero function. \vspace{2pt}\\
    \textbf{Learning:} For each point $(X,y) \in \bitsmk \times \bits$ in the training set and coordinate $i \in [k]$, overwrite
    \begin{equation*}
        g_{i,\bS}(X^{(i)}) \leftarrow  y.
    \end{equation*}
    Afterwards, define $G_{\bS}:\bitsmk \to \{-k, \dots, k\}$ as
    \begin{equation*}
        G_{\bS}(X) \coloneqq \sum_{i \in [k]} g_{i,\bS}(X^{(i)}).
    \end{equation*}
    \textbf{Choose threshold:} for a given threshold $\tau$, let
    $$h_{\tau}(X) \coloneqq \sign[G_{\bS}(X) \geq \tau].$$

    Let $u = \uintval$. We define $\mcH$, the set of hypotheses to be $\mcH = \{h_\tau(X) \, | \, \tau \in \set{\lint, \dots, \uint}\} \cup \{\boldsymbol{-1}, \boldsymbol{1}\}$ where $\boldsymbol{-1}$ and $\boldsymbol{1}$ are the constant $-1$ and $1$ functions respectively.
    
    Output the $h \in \mcH$ with maximum advantage on the validation set. 
    \end{tcolorbox}
\caption{Our algorithm for weak learning $\mcC$.}
\label{fig:weak-learner}
\end{figure}

We present a sample-efficient weak learning algorithm that satisfies the upper bound for the problem in \Cref{thm:boosting-formal}. 
%The algorithm works on the following intuition: on average, the $f_i$'s are weakly correlated with $\Maj(f_1, \ldots, f_k)$ so on each input $(X^{(1)}, \ldots, X^{(k)}) \in \bitsmk$ we can learn some information about $f_i(X^{(i)})$. For instance, if the only information we know is that $\Maj(f_1(X^{(1)}), \ldots, f_k(X^{(k)})) = 1$ then it is slightly more likely that $f_1(X^{(1)}) = 1$ than $f_1(X^{(1)}) = 0$. Combining this information across all $m$ samples allows us to build a learner with advantage $\Omega_\kappa(\gamma)$.\lnote{I worry that this intuition makes the proof seem easier than it actually is.}\gnote{Are you still worried about this after what is written in the proof overview?} More formally:
\begin{lemma}[Upper bound of \Cref{thm:boosting-formal}, Formal version of \Cref{lem:weak-learner}]
    \label{lem:boosting-upper}
    In the setting of \Cref{thm:boosting-formal}, let $\mcD$ be any distribution that is $\kappa$-smooth on $\bitsmk$. There is a $2m$-sample weak learning algorithm achieving advantage $\Omega\paren*{\advantage}$ for the concept class $\mcC$ on the input distribution $\mcD$.
\end{lemma}

In the following, we will let $F \in \mcC$ be the target function, $\mcD$ be a $\kappa$-smooth distribution on $\bitsmk$,  and $G_S$ be as defined in \Cref{fig:weak-learner}.

We describe the learning algorithm of \Cref{lem:boosting-upper} in \Cref{fig:weak-learner}. It build estimators $g_{1,S}, \dots, g_{k,S}$ for $f_1, \dots, f_k$ as follows: For each point $X$ in the training set $S$, if $\Maj(f_1(X^{(1)}), \ldots, f_k(X^{(k)})) = y$ then set $g_{i,S}(X^{(i)}) = y$. A natural approach, then, would be to return $G_S \coloneqq \sum_{i \in [k]}g_{i,S}$ as the final weak learner, however, $G_S$ here is not a $\bits$ classifier so we turn it into a classifier by trying different threshold functions and returning the one with best advantage on the validation set.

At a high level, the proof will consist in showing that there always exists a hypothesis $h \in \mcH$ that achieves weak correlation with $F$ where, by weak correlation, we mean correlation $\Tilde{\Omega}(1/\sqrt{k})$. We then show that if such a hypothesis exists, the hypothesis chosen by \Cref{fig:weak-learner} also achieves weak correlation with $F$.

We note $\boldsymbol{-1}$ and $\boldsymbol{-1}$ the constant $-1$ and $1$ functions respectively. We start by noticing that if $F$ is very biased ($\abs*{\Ex_{\bX \sim \cD}[F(\bX)]} > 1 / \sqrt{k}$) then either $\Ex_{\bX \sim \cD}[F(\bX) (\boldsymbol{-1}(\bX))] > 1/\sqrt{k}$ or $\Ex_{\bX \sim \cD}[F(\bX) \boldsymbol{1}(\bX)] > 1/\sqrt{k}$. Note that $\boldsymbol{-1}$ and $\boldsymbol{1}$ are both hypotheses in $\mcH$, the set of possible hypotheses for the weak learner. Hence for the rest of the proof, we'll assume that $F$ has low bias ($\abs*{\Ex_{\bX \sim \cD}[F(\bX)]} \leq 1 / \sqrt{k}$).

The rest of the proof can be broken down into 4 main steps:
\begin{enumerate}
    \item \Cref{lem:weak-correlation}: We show that with high probability over the draw of the random sample $\bS$, $G_{\bS}$ achieves constant correlation with $F$ ($\Ex_{\bX \sim \cD}[G_{\bS}(\bX)F(\bX)] \geq \Omega_\kappa(1)$).
    \item \Cref{lem:g-concentrates}: We show that conditioned on $F$ having low bias, then with high probability over the samples $\bS$, the $G_{\bS}$ that we construct has low values on most inputs, $$\Prx_{\bX \sim \cD} \left [ \abs*{G_{\bS}(\bX)} \leq O(\sqrt{k \log{k} \kappa}) \right ] \geq 1 - \frac{1}{k^2}.$$
    \item \Cref{lem:weakly-correlated-h-exists-if-f-low-bias}: We show that having both items 1 and 2 is sufficient to prove that there exists a good hypothesis $h^* \in \mcH$ that achieves weak correlation $ \Ex_{\bX \sim \cD}[F(\bX) \cdot h^*(\bX)] \geq \Omega \paren*{\advantageline}$. This implies that with high probability over the draw of the random sample $\bS$, there exists a hypothesis that achieves weak correlation with $F$.
    \item Proof of \Cref{lem:boosting-upper}: In the final step, we show that if there exists a hypothesis that is weakly correlated with $F$, then the hypothesis we choose by minimizing validation error is also weakly correlated with $F$ with high probability. This is a simple generalization argument and uses standard learning theory results. The previous steps show that a weakly correlated hypothesis exists with high probability over the draw of the random sample $\bS$, and so applying the generalization result proves \Cref{lem:boosting-upper}.
\end{enumerate}

\subsection{Notation and basic technical tools}
\label{sec:notation-technical-tools}
\paragraph{Notation.} {We start by introducing a few pieces of notation. For each $i \in [k]$ and $x \in \bitsm$, let:
\begin{itemize}
    \item $\cD_i(x) \coloneqq \Prx_{\bX \sim \mcD}[\bX^{(i)} = x]$
    \item $\mu_i(x) \coloneqq \Ex_{\bX \sim \mcD}[F(\bX) \mid \bX^{(i)} = x]$
    \item $q_{i,\bS}(x) \coloneqq \Prx_{\bS}[\exists \bX \in \bS: \bX^{(i)} = x]$
\end{itemize}
Notice that $q_{i,\bS}(x) = 1 - (1 - \cD_i(x))^m$.}

\paragraph{Basic technical tools.}
We show that for any choice of $F \in \mcC$ and $\kappa$-smooth $\mcD$, the individual $f_i$ have a noticeable amount of correlation with $F$:\begin{corollary}[The $f_i$ are correlated with $F$]
    \label{cor:many-f-correlated}
    For any $F = \Maj(f_1, \ldots, f_k) \in \mcC$ and $\kappa$-smooth $\mcD$,
    \begin{equation*}
        \sum_{i \in [k]} \Ex_{\bX \sim \mcD}[F(\bX)f_{i}(\bX^{(i)})] \geq \Omega\paren*{\frac{\sqrt{k}}{\kappa}}.
    \end{equation*}
\end{corollary}
\begin{proof}
This result holds by applying \Cref{claim:maj-easy} using the fact that a $\kappa$-smooth distribution has $1/\kappa$ density.
 %   Let $\mcU$ be the uniform distribution on $\bitsmk$. Then, by \Cref{fact:bin-max-p}, for any choice of $v \in [k/4,3k/4]$,
%    \begin{equation*}
%        \Prx_{\bX \sim \mcU}\bracket*{\sum_{i \in [k]}f_i(\bX^{(i)}) = v} \leq O(1/\sqrt{k}).
%    \end{equation*}
%    Let $\ell \coloneqq \Omega(\sqrt{k}/\kappa)$. Adding up the probability from levels close to the center, we have that
%    \begin{equation*}
%        \Prx_{\bX \sim \mcD}\bracket*{\sum_{i \in [k]}f_i(\bX^{(i)}) \in [-\ell,\ell]} \leq \kappa \cdot \Prx_{\bX \sim \mcU}\bracket*{\sum_{i \in [k]}f_i(\bX^{(i)}) \in [-\ell,\ell]} \leq 1/2.
%    \end{equation*}
%    Therefore,
%    \begin{align*}
%        \Ex_{\bX \sim \mcD}\bracket*{F(\bX)\sum_{i \in [k]}f_{i}(\bX^{(i)})} &= \Ex_{\bX \sim \mcD}\bracket*{\abs*{\sum_{i \in [k]}f_{i}(\bX^{(i)})}} \\
%        &\geq \ell \cdot \Prx_{\bX \sim \mcD}\bracket*{\abs*{\sum_{i \in [k]}f_{i}(\bX^{(i)})} \geq \ell} \geq \ell/2. \qedhere
%    \end{align*}
\end{proof}

We'll also use the following basic probability fact:
\begin{fact}
    \label{fact:bin-large}
    Let $\ba \sim \mathrm{Bin}(n,p)$ with mean $\mu \coloneqq np$. Then,
    \begin{equation*}
        \Pr[\ba \geq 1] \geq \frac{\mu}{\mu+1}.
    \end{equation*}
\end{fact}
\begin{proof}
    \begin{equation*}
        \Pr[\ba = 0] = (1 - p)^n \leq e^{-\mu} \leq \frac{1}{\mu + 1} \tag{Use $1 + x \leq e^x$ twice}
    \end{equation*}
    which implies the desired result.
\end{proof}

\subsection{$G_{S}$ is well correlated with $F$}
We want to show that, with high probability over the random sample $\bS$, $G_{\bS}$ (defined in the weak learning algorithm) has a constant amount of correlation with $F$.
%Toward showing that there exist hypotheses in $\mcH$ that are well correlated with $F$ even when $F$ has low bias, we start by showing that, with high probability over the random sample $\bS$, $G_{\bS}$ (defined in the weak learning algorithm) is well correlated with $F$. To do this, we first prove that the expected value over $\bS$ of the correlation between $G_{\bS}$ and $F$ is high in \Cref{claim:weak-correlation-expectation}. We then use a concentration inequality to show how this implies that $G_{\bS}$ is well correlated with $F$ with high probability over the draw of the random sample $\bS$.\gnote{I wonder if some of this would be better moved after the statement of \Cref{lem:weakly-correlated-h-exists-if-f-low-bias}. Since we already have the 4-step breakdown, I think all we need above the statement of \Cref{lem:weak-correlation} is your sentence ``We start by showing that, with high probability over the random sample $\bS$, $G_{\bS}$ (defined in the weak learning algorithm) is well correlated with $F$." Then, after the statement of \Cref{lem:weak-correlation}, your two sentences ``To prove \Cref{lem:weak-correlation}, we first prove that the expected value over $\bS$ of the correlation between $G_{\bS}$ and $F$ is high in \Cref{claim:weak-correlation-expectation}. We will then use a concentration inequality to show how this implies that $G_{\bS}$ is well correlated with $F$ with high probability over the draw of the random sample $\bS$."}

\begin{lemma}
    \label{lem:weak-correlation}
    With probability at least $1 - \exp \left (-\Omega \left (\lfrac{m}{k^2 \kappa^8} \right ) \right )$ over the draw of the random sample $\bS$,
    \begin{equation*}
        \Ex_{\bX \sim \mcD}\bracket*{F(\bX) \cdot G_{\bS}(\bX)} \geq \Omega(1/\kappa^3).
    \end{equation*}
\end{lemma}

To prove \Cref{lem:weak-correlation}, we first prove that the expected value over $\bS$ of the correlation between $G_{\bS}$ and $F$ is high in \Cref{claim:weak-correlation-expectation}. We will then use a concentration inequality to show how this implies that $G_{\bS}$ is well correlated with $F$ with high probability over the draw of the random sample $\bS$.

%To prove \Cref{lem:weak-correlation}, we start by showing that its statement is true in expectation over the random sample $\bS$.
\begin{claim}
    \label{claim:weak-correlation-expectation}
    \Cref{lem:weak-correlation} is true in expectation over the random sample $\bS$, that is,
    \begin{equation*}
        \Ex_{\bS}\bracket*{\Ex_{\bX \sim \mcD}\bracket*{F(\bX) \cdot G_{\bS}(\bX)}} \geq \Omega(1/\kappa^3).
    \end{equation*}
\end{claim}

\begin{proof}

    We begin by recalling some notation that will aid in the proof. For any training set $S \in (\bitsmk \times \bits)^m$, let $g_{1,S}, \ldots, g_{k,S}$ be the functions that \Cref{fig:weak-learner} would construct given the dataset $S$, and $G_S$ the function summing them up $G_S \coloneqq \sum_{i \in [k]} g_{i,S}$. Our goal is to understand the average correlation,
    \begin{equation*}
        \Ex_{\bX \sim \mcD}\bracket*{F(\bX) \cdot G_{S}(\bX)} = \sum_{i \in [k]}\Ex_{\bX \sim \mcD}[F(\bX^{(i)}) \cdot g_{i,S}(\bX)] \tag{Linearity of expectation}.
    \end{equation*}
    Recalling the notation introduced in \Cref{sec:notation-technical-tools}, the above can be written as
    \ifnum\focs=1
    \begin{equation}
        \label{eq:def-correlation}
        \Ex_{\bX}\bracket*{F(\bX) \cdot G_{\bS}(\bX)} = \sum_{i \in [k], x \in \bitsm} \cD_i(x) \mu_i(x) g_{i, \bS}(x)
    \end{equation}
    where $bX \sim \mcD$.
    \else
    \begin{equation}
        \label{eq:def-correlation}
        \Ex_{\bX \sim \mcD}\bracket*{F(\bX) \cdot G_{\bS}(\bX)} = \sum_{i \in [k], x \in \bitsm} \cD_i(x) \mu_i(x) g_{i, \bS}(x).
    \end{equation}
    \fi
    The goal of this proof is to show that the expected value of $\Ex_{\bX \sim \mcD}\bracket*{F(\bX) \cdot G_{\bS}(\bX)}$ over the randomness of the sample $\bS$ is large. We start by noting that $\Ex_{\bS}[g_{i,\bS}(x)] = q_{i,\bS}(x) \cdot \mu_i(x)$ since $g_{i,\bS}(x)$ will be overwritten by $F(X)$ where $X$ is the last point in the sample whose $i^{th}$ coordinate is $x$. Therefore,
    \ifnum\focs=1
    \begin{align*}
        &\Ex_{\bS}\bracket*{\Ex_{\bX \sim \mcD}\bracket*{F(\bX) \cdot G_{\bS}(\bX)}} \\
        &= \sum_{i \in [k], x \in \bitsm} \cD_i(x) \mu_i(x) \Ex_{\bS}[g_{i,\bS}(x)] \\
        &= \sum_{i \in [k], x \in \bitsm} \cD_i(x) q_{i,\bS}(x) \mu_i(x)^2.
    \end{align*}
    \else
    \begin{equation*}
        \Ex_{\bS}\bracket*{\Ex_{\bX \sim \mcD}\bracket*{F(\bX) \cdot G_{\bS}(\bX)}} = \sum_{i \in [k], x \in \bitsm} \cD_i(x) \mu_i(x) \Ex_{\bS}[g_{i,\bS}(x)] = \sum_{i \in [k], x \in \bitsm} \cD_i(x) q_{i,\bS}(x) \mu_i(x)^2.
    \end{equation*}
    \fi
    \Cref{fact:bin-large} gives that for all $S$, $q_{i,S}(x) \geq \frac{m \cD_i(x)}{1 + m\cD_i(x)}$. Since $\mcD$ is $\kappa$-smooth, we know that $\cD_i(x) \leq \kappa/m$. Therefore, $q_{i,S}(x) \geq \frac{m \cD_i(x)}{1 + \kappa}$, and so
    \ifnum\focs=1
    \begin{align*}
        \Ex_{\bS}&\bracket*{\Ex_{\bX \sim \mcD}\bracket*{F(\bX) \cdot G_{\bS}(\bX)}} \\ & \geq \frac{m}{1+\kappa} \cdot \sum_{i \in [k], x \in \bitsm} \cD_i(x)^2 \mu_i(x)^2.
    \end{align*}
    \else
    \begin{equation*}
        \Ex_{\bS}\bracket*{\Ex_{\bX \sim \mcD}\bracket*{F(\bX) \cdot G_{\bS}(\bX)}} \geq \frac{m}{1+\kappa} \cdot \sum_{i \in [k], x \in \bitsm} \cD_i(x)^2 \mu_i(x)^2.
    \end{equation*}
    \fi
    Since $f_i(x)^2 = 1$, we are free to add it as a term to the above equation, giving
    \ifnum\focs=1
    \begin{align*}
        &\Ex_{\bS}\bracket*{\Ex_{\bX \sim \mcD}\bracket*{F(\bX) \cdot G_{\bS}(\bX)}} \\ &\geq \frac{m}{1+\kappa}\sum_{i \in [k]}\sum_{x \in \bitsm}(f_i(x)\cD_i(x)\mu_i(x))^2 \\
        &\geq \frac{1}{1+\kappa}\sum_{i \in [k]} \paren*{\sum_{x \in \bitsm}f_i(x)\cD_i(x)\mu_i(x)}^2 \tag{Jensen's inequality}\\
        &=\frac{1}{1+\kappa}\sum_{i \in [k]}  \Ex_{\bX \sim \mcD}\bracket*{f_i(\bX^{(i)})F(\bX)}^2 \tag{$\mu_i(x) \coloneqq \Ex_{\bX \sim \mcD}[F(\bX) \mid \bX^{(i)} = x]$}\\
        &\geq \frac{1}{1+\kappa} \cdot \frac{1}{k} \cdot \paren*{\sum_{i \in [k]}  \Ex_{\bX \sim \mcD}\bracket*{f_i(\bX^{(i)})F(\bX)}}^2\tag{Jensen's inequality} \\
        &\geq \frac{1}{1+\kappa} \cdot \frac{1}{k} \cdot \Omega\paren*{\frac{\sqrt{k}}{\kappa}}^2 = \Omega(1/\kappa^3). \tag{\Cref{cor:many-f-correlated}}
    \end{align*}
    \else
    \begin{align*}
        \Ex_{\bS}\bracket*{\Ex_{\bX \sim \mcD}\bracket*{F(\bX) \cdot G_{\bS}(\bX)}} &\geq \frac{m}{1+\kappa}\sum_{i \in [k]}\sum_{x \in \bitsm}(f_i(x)\cD_i(x)\mu_i(x))^2 \\
        &\geq \frac{1}{1+\kappa}\sum_{i \in [k]} \paren*{\sum_{x \in \bitsm}f_i(x)\cD_i(x)\mu_i(x)}^2 \tag{Jensen's inequality}\\
        &=\frac{1}{1+\kappa}\sum_{i \in [k]}  \Ex_{\bX \sim \mcD}\bracket*{f_i(\bX^{(i)})F(\bX)}^2 \tag{$\mu_i(x) \coloneqq \Ex_{\bX \sim \mcD}[F(\bX) \mid \bX^{(i)} = x]$}\\
        &\geq \frac{1}{1+\kappa} \cdot \frac{1}{k} \cdot \paren*{\sum_{i \in [k]}  \Ex_{\bX \sim \mcD}\bracket*{f_i(\bX^{(i)})F(\bX)}}^2\tag{Jensen's inequality} \\
        &\geq \frac{1}{1+\kappa} \cdot \frac{1}{k} \cdot \Omega\paren*{\frac{\sqrt{k}}{\kappa}}^2 = \Omega(1/\kappa^3). \tag{\Cref{cor:many-f-correlated}}
    \end{align*}
    \fi
    This completes the proof.
\end{proof}

    We have thus proved that, in expectation over the random draw of the sample $\bS$, $G_{\bS}$ has constant correlation with $F$. We now want to prove \Cref{lem:weak-correlation} by showing that this happens with high probability over the draw of the sample. To do this, we show that $\Ex_{\bX \sim \mcD}\bracket*{F(\bX) \cdot G_{\bS}(\bX)}$ concentrates around its mean using the bounded differences inequality. 
    
    \begin{proof}
    
    Consider any samples $S, S'$ differing in one data point. Then $g_i(S)$ and $g_i(S')$ can differ on at most $2$ inputs (corresponding to the $X^{(i)}$ and ${X^{(i)}}'$ of the differing point). By \Cref{eq:def-correlation},
    \ifnum\focs=1
    \begin{align*}
        &\Ex_{\bX \sim \mcD}\bracket*{F(\bX) \cdot G_{S}(\bX)} - \Ex_{\bX \sim \mcD}\bracket*{F(\bX) \cdot G_{S'}(\bX)} \\ &= \sum_{i \in [k], x \in \bitsm}\cD_i(x) \mu_i(x)\cdot (g_{i,S}(x) - g_{i,S'}(x))\\
        &\leq   \sum_{i \in [k], x \in \bitsm}\frac{2\kappa}{m}\cdot \Ind[g_{i,S}(x) \neq g_{i,S'}(x)] \tag{$\cD_i(x) \leq \kappa/m$, $\mu_i(x) \leq 1$} \\
        &\leq \frac{4k\kappa}{m}. \tag{At most $2k$ points differ.}
    \end{align*}
    Therefore, by \Cref{fact:bounded-diff},
    \begin{align*}
        \Prx_{\bS} &\left [\Ex_{\bX}\bracket*{F(\bX) \cdot G_{\bS}(\bX)} \leq \Ex_{\bS}\bracket*{\Ex_{\bX}\bracket*{F(\bX) \cdot G_{\bS}(\bX)}} - \eps \right ] \\
        &\leq  \exp\paren*{-\frac{\eps^2m}{8k^2 \kappa^2 }}.
    \end{align*}
    Setting $\eps = O(1/\kappa^3)$ using the earlier bound that $\Ex_{\bS}\bracket*{\Ex_{\bX \sim \mcD}\bracket*{F(\bX) \cdot G_{\bS}(\bX)}} \geq \Omega(1/\kappa^3)$, we have that
    \begin{align*}
        \Prx_{\bS}&\bracket*{\Ex_{\bX \sim \mcD}\bracket*{F(\bX) \cdot G_{\bS}(\bX)} \leq \Omega(1/\kappa^3)} \\ &\leq \exp\paren*{-\Omega\paren*{\frac{m}{k^2 \kappa^8}}}.\qedhere
    \end{align*}
    \else
    \begin{align*}
        \Ex_{\bX \sim \mcD}\bracket*{F(\bX) \cdot G_{S}(\bX)} - \Ex_{\bX \sim \mcD}\bracket*{F(\bX) \cdot G_{S'}(\bX)} &= \sum_{i \in [k], x \in \bitsm}\cD_i(x) \mu_i(x)\cdot (g_{i,S}(x) - g_{i,S'}(x))\\
        &\leq   \sum_{i \in [k], x \in \bitsm}\frac{2\kappa}{m}\cdot \Ind[g_{i,S}(x) \neq g_{i,S'}(x)] \tag{$\cD_i(x) \leq \kappa/m$, $\mu_i(x) \leq 1$} \\
        &\leq \frac{4k\kappa}{m}. \tag{At most $2k$ points differ.}
    \end{align*}
    Therefore, by \Cref{fact:bounded-diff},
    \begin{equation*}
        \Prx_{\bS} \left [\Ex_{\bX \sim \mcD}\bracket*{F(\bX) \cdot G_{\bS}(\bX)} \leq \Ex_{\bS}\bracket*{\Ex_{\bX \sim \mcD}\bracket*{F(\bX) \cdot G_{\bS}(\bX)}} - \eps \right ] \leq  \exp\paren*{-\frac{\eps^2m}{8k^2 \kappa^2 }}.
    \end{equation*}
    Setting $\eps = O(1/\kappa^3)$ using the earlier bound that $\Ex_{\bS}\bracket*{\Ex_{\bX \sim \mcD}\bracket*{F(\bX) \cdot G_{\bS}(\bX)}} \geq \Omega(1/\kappa^3)$, we have that
    \begin{equation*}
         \Prx_{\bS}\bracket*{\Ex_{\bX \sim \mcD}\bracket*{F(\bX) \cdot G_{\bS}(\bX)} \leq \Omega(1/\kappa^3)} \leq \exp\paren*{-\Omega\paren*{\frac{m}{k^2 \kappa^8}}}. \qedhere
    \end{equation*}
    \fi
\end{proof}

\subsection{$G$ concentrates if $F$ has low bias}
Given the result given in \Cref{lem:weak-correlation}, it could seem that we are essentially done since we know that, with high probability over $\bS$, $G_{\bS}$ achieves constant correlation with $F$. However, $G_{\bS}$ is not a $\bits$ classifier (in particular, $G_{\bS}$ takes values in $\{-k, \dots, k\}$.) We could easily turn $G_{\bS}$ into a classifier by returning $\sign(G_{\bS})$ instead of $G_{\bS}$ but this could cause the following subtle issue. Imagine the scenario where $F$ has bias $\E_{\bX \sim \cD}[F(\bX)] = \frac{1}{k}$ and $G_S(X) = k$ with probability $1$. In this case, the property from \Cref{lem:weak-correlation} is respected since $F$ and $G_{S}$ have constant correlation. However, by turning $G_{S}$ into a classifier, the correlation between $F$ and $\sign(G)$ goes down to $1/k$, whereas our goal is to prove a correlation of $\Omega(1/\sqrt{k})$. The issue here is that the correlation between $F$ and $G_{S}$ comes mostly from the magnitude of $G$, not from $G_{S}$ being a good predictor for $F$. Thankfully, it turns out that we can show that if $F$ has low bias then, with high probability over $\bS$, $G_{\bS}$ will only take on values with small magnitude. Hence we can prove that the bad scenario described above is very unlikely.

\begin{lemma}[$G$ concentrates]
\label{lem:g-concentrates}
    Let $\uint = \uintval$ as defined in \Cref{fig:weak-learner}. If $\abs{\Ex_{\bX \sim \mcD}[F(\bX)]} \leq 1/\sqrt{k}$, then, with probability at least $1 - \exp \paren*{-\Omega(\frac{m}{k})}$ over the draw of the random sample $\bS$,
    \begin{equation*}
        \Prx_{\bX \sim \cD}\bracket*{\abs*{G_{\bS}(\bX)} \leq \uint} \geq 1 - \frac{1}{k^2}.
    \end{equation*}
\end{lemma}

In order to prove \Cref{lem:g-concentrates}, we'll use the following:
\begin{claim}
    \label{claim:p-concentrate}
    For any $\kappa$-smooth distribution $\mcD$ on $\bitsmklog$, $i \in [k]$, $x \in \bitsmlog$, and for all $v \geq 0$,
    \begin{equation*} 
        \Prx_{\bi \in [k], \bx \sim \bitsmlog}\bracket*{\cD_{\bi}(\bx) \notin \bracket*{ \lfrac{1-v}{m}, \lfrac{1+v}{m}}} \leq \frac{2\kappa}{v^2k}.
    \end{equation*}
\end{claim}
\begin{proof}
    We begin by showing $\Prx_{\bi \in [k], \bx \sim \bitsmlog}\bracket*{\cD_{\bi}(\bx) \geq \lfrac{1+v}{m}} \leq \frac{\kappa}{v^2k}$. For each $i \in [k]$, let
    \begin{equation*}
        B_i \coloneqq \set{x \in \bitsmlog \mid \cD_i(x) \geq \lfrac{1+v}{m}}.
    \end{equation*}
    The $B_i$'s are the ``bad'' sets of values for each $i$. We want to show that they are relatively small.

    Define $r_i \coloneqq \frac{|B_i|}{m}$ and $\mu \coloneqq r_1 + \cdots + r_k$. Clearly,
    \begin{equation*}
         \Prx_{\bi \in [k], \bx \sim \bitsmlog}\bracket*{\cD_{\bi}(\bx) \geq \lfrac{1+v}{m}} = \frac{\mu}{k}.
    \end{equation*}
    For any $X \in \bitsmklog$, define,
    \begin{equation*}
        \ell(X) \coloneqq \sum_{i \in [k]} \Ind[X^{(i)} \in B_i].
    \end{equation*}
    $\ell(X)$ counts the number of bad sets $x$ is contained in.
    
    Then, by linearity of expectation,
    \begin{equation}
        \label{eq:sum-of-p}
        \Ex_{\bX \sim \mcD}[\ell(\bX)] = \sum_{i \in [k], x \in B_i} \cD_i(x).
    \end{equation}
    We also know that, for $\bX \sim \bitsmklog$, $\ell(\bX)$ is the sum of $k$ independent Bernoullis with means $r_1, \ldots, r_k$ respectively. Therefore, it has mean $\mu$ and variance at most $\mu$. This gives that,
    \begin{align*}
        \Ex_{\bX \sim \mcD}\bracket*{\abs*{\ell(\bX) - \mu}} &\leq \sqrt{\Ex_{\bX \sim \mcD}\bracket*{\paren*{\ell(\bX) - \mu}^2} } \tag{Jensen's inequality} \\
        &\leq \sqrt{\kappa\cdot \Ex_{\bX \sim \bitsmklog}\bracket*{\paren*{\ell(\bX) - \mu}^2} }\tag{$\mcD$ is $\kappa$-smooth} \\
        &\leq \sqrt{\kappa\mu}.
    \end{align*}
    which means that $\Ex_{\bX \sim \mcD}[\ell(\bX)] \leq \mu + \sqrt{\kappa\mu}$. Combining with \Cref{eq:sum-of-p},
    \begin{equation*}
         \sum_{i \in [k], x \in B_i} \cD_i(x)  \leq \mu + \sqrt{\kappa\mu}.
    \end{equation*}
    However, we also know that every $x \in B_i$ satisfies $\cD_i(x) \geq \frac{1+v}{m}$, giving that
    \begin{equation*}
         \sum_{i \in [k], x \in B_i} \cD_i(x) \geq m \mu \cdot \left (\frac{1+v}{m} \right ) = \mu + \mu v.
    \end{equation*}
    Combining the above,
    \begin{equation*}
        \mu + \mu v\leq \mu + \sqrt{\kappa \mu}.
    \end{equation*}
    Solving the above equation, we have
    \begin{equation*}
        \mu \leq \frac{\kappa}{v^2}
    \end{equation*}
    which gives the first desired statement. For the second, we instead define $B_i$ as the set of points for which $\cD_i(x) \leq \frac{1-v}{m}$ and the rest of the proof is identical.
\end{proof}
Toward proving \Cref{lem:g-concentrates}, we start by showing that the expected value of $G_{\bS}$ concentrates with high probability over the draw of the random sample $\bS$.
\begin{claim}[The expected value of $G_{\bS}$ concentrates]
\label{claim:g-concentrates-in-expectation}
    If $\abs{\Ex_{\bX \sim \mcD}[F(\bX)]} \leq 1/\sqrt{k}$, then, with probability at least $1 - \exp \paren*{-\Omega(\frac{m}{k})}$ over the draw of the random sample $\bS$,
    \begin{equation*}
        \abs*{\Ex_{\bX \sim \mcD}[G_{\bS}(\bX)]} \leq O(\sqrt{k}\kappa).
    \end{equation*}
\end{claim}
\begin{proof}
    Similarly to previous proofs, we will start by showing that the statement holds in expectation over the random sample $\bS$ then use the bounded differences inequality to show that it holds with high probability over the draw of $\bS$. Using $\cD_i, \mu_i,$ and $q_{i,S}$ as defined in \Cref{sec:notation-technical-tools},
    \begin{align*}
        \Ex_{\bS}\bracket*{\Ex_{\bX \sim \mcD}[G_{\bS}(\bX)]} = \sum_{i \in [k], x \in \bitsm} \cD_i(x) \mu_i(x) q_{i,\bS}(x). 
    \end{align*}
    \ifnum\focs=1
    Expanding the above notation,
    \begin{align*}
        &\Ex_{\bS}\bracket*{\Ex_{\bX \sim \mcD}[G_{\bS}(\bX)]} \\
        &= \sum_{\substack{i \in [k]\\ x \in \bitsm}} \Prx_{\bX}[\bX^{(i)} = x] \Ex_{\bX} [F(\bX) \mid \bX^{(i)} = x] q_{i,\bS}(x)\\
        &= \Ex_{\bX}\bracket*{F(\bX) \cdot \sum_{i \in [k]} q_{i,\bS}(\bX^{(i)})}
    \end{align*}
    where $\bx\sim \mcD$.
    The key insight is that $q_{i,\bS}(\bx)$ concentrates so the above is roughly proportional to $\Ex_{\bX \sim \mcD}[F(\bX)]$. Let $c \coloneqq 1 - (1 - 1/m)^m$. Then, since $F(\bX) \in \bits$,
    \begin{align*}
        &\abs*{\Ex_{\bX \sim \mcD}\bracket*{F(\bX) \cdot \sum_{i \in [k]} q_{i,\bS}(\bX^{(i)})} - ck \cdot \Ex_{\bX \sim \mcD}[F(\bX)]} \\
        &\leq \Ex_{\bX \sim \mcD}\bracket*{\sum_{i \in [k]}\abs*{q_{i,\bS}(\bX^{(i)}) - c}}.
    \end{align*}
    Recall that $q_{i,\bS}(x) = 1 - (1 - \cD_i(x))^m$. This function is $m$-Lipschitz as long as $t \geq 0$, so $|q_{i,\bS}(x_i) - c| \leq m|\cD_i(x_i) - 1/m| = |m\cD_i(x_i) - 1|$ , and therefore,
    \begin{align*}
        \Ex_{\bX \sim \mcD}&\bracket*{\sum_{i \in [k]}\abs*{q_{i,\bS}(\bX^{(i)}) - c}} \\
        &\leq \sum_{i \in [k]}  \Ex_{\bX \sim \mcD}\bracket*{\abs*{m\cD_i(\bX^{(i)}) - 1}} \\
        &= \sum_{i \in [k]}\int_{0}^\infty \Prx_{\bX \sim \mcD}[\abs{m\cD_i(\bX^{(i)}) - 1} \geq t]dt \\
        &= k \cdot \int_{0}^\infty\Prx_{\bX \sim \cD, \bi \in [k]}[\abs{m\cD_{\bi}(\bX^{(\bi)}) - 1} \geq t] dt \\
        &\leq k \cdot \paren*{1/\sqrt{k}+ \int_{1/\sqrt{k}}^\infty \frac{2\kappa}{t^2k}dt}\tag{\Cref{claim:p-concentrate}} \\
        &= k \cdot \paren*{1/\sqrt{k}+ \frac{2 \kappa}{\sqrt{k}}} = \sqrt{k} \cdot (2 \kappa+1).
    \end{align*}
    Combining the above results, we have thus shown that:
    \begin{align*}
        \Ex_{\bS}&\bracket*{\Ex_{\bX \sim \mcD}[G_{\bS}(\bX)]} \\&= \Ex_{\bX \sim \mcD}\bracket*{F(\bX) \cdot \sum_{i \in [k]} q_{i,\bS}(\bX^{(i)})}\\
        &\leq \Bigg|\Ex_{\bX \sim \mcD}\bracket*{F(\bX) \cdot \sum_{i \in [k]} q_{i,\bS}(\bX^{(i)})} \\
        &\qquad - ck \cdot \Ex_{\bX \sim \mcD}[F(\bX)] + ck \cdot \Ex_{\bX \sim \mcD}[F(\bX)]\Bigg| \\
        &\leq \Bigg|\Ex_{\bX \sim \mcD}\bracket*{F(\bX) \cdot \sum_{i \in [k]} q_{i,\bS}(\bX^{(i)})} \\
        &\qquad - ck \cdot \Ex_{\bX \sim \mcD}[F(\bX)]\Bigg| + ck \cdot \abs*{\Ex_{\bX \sim \mcD}[F(\bX)]} \tag{triangle inequality} \\
        &\leq \Ex_{\bX \sim \mcD}\bracket*{\sum_{i \in [k]}\abs*{q_{i,\bS}(\bX^{(i)}) - c}} + ck \cdot \Ex_{\bX \sim \mcD}[F(\bX)] \\
        &\leq \sqrt{k} \cdot (2\kappa+1) + k \cdot \Ex_{\bX \sim \mcD}[F(\bX)] \tag{$c \leq 1$} \\
        &\leq \sqrt{k} \cdot (2\kappa+2). \tag{We assume $\abs{\Ex_{\bX \sim \mcD}[F(\bX)]} \leq 1/\sqrt{k}$}
    \end{align*}
    \else
    Expanding the above notation,
    \begin{align*}
        \Ex_{\bS}\bracket*{\Ex_{\bX \sim \mcD}[G_{\bS}(\bX)]} &= \sum_{i \in [k], x \in \bitsm} \Prx_{\bX \sim \cD}[\bX^{(i)} = x] \cdot \Ex_{\bX \sim \cD} [F(\bX) \mid \bX^{(i)} = x] \cdot q_{i,\bS}(x)\\
        &= \Ex_{\bX \sim \mcD}\bracket*{F(\bX) \cdot \sum_{i \in [k]} q_{i,\bS}(\bX^{(i)})}.
    \end{align*}
    The key insight is that $q_{i,\bS}(\bx)$ concentrates so the above is roughly proportional to $\Ex_{\bX \sim \mcD}[F(\bX)]$. Let $c \coloneqq 1 - (1 - 1/m)^m$. Then, since $F(\bX) \in \bits$,
    \begin{equation*}
        \abs*{\Ex_{\bX \sim \mcD}\bracket*{F(\bX) \cdot \sum_{i \in [k]} q_{i,\bS}(\bX^{(i)})} - ck \cdot \Ex_{\bX \sim \mcD}[F(\bX)]} \leq \Ex_{\bX \sim \mcD}\bracket*{\sum_{i \in [k]}\abs*{q_{i,\bS}(\bX^{(i)}) - c}}.
    \end{equation*}
    Recall that $q_{i,\bS}(x) = 1 - (1 - \cD_i(x))^m$. This function is $m$-Lipschitz as long as $t \geq 0$, so $|q_{i,\bS}(x_i) - c| \leq m|\cD_i(x_i) - 1/m| = |m\cD_i(x_i) - 1|$ , and therefore,
    \begin{align*}
        \Ex_{\bX \sim \mcD}\bracket*{\sum_{i \in [k]}\abs*{q_{i,\bS}(\bX^{(i)}) - c}} &\leq \sum_{i \in [k]}  \Ex_{\bX \sim \mcD}\bracket*{\abs*{m\cD_i(\bX^{(i)}) - 1}} \\
        &= \sum_{i \in [k]}\int_{0}^\infty \Prx_{\bX \sim \mcD}[\abs{m\cD_i(\bX^{(i)}) - 1} \geq t]dt \\
        &= k \cdot \int_{0}^\infty\Prx_{\bX \sim \cD, \bi \in [k]}[\abs{m\cD_{\bi}(\bX^{(\bi)}) - 1} \geq t] dt \\
        &\leq k \cdot \paren*{1/\sqrt{k}+ \int_{1/\sqrt{k}}^\infty \frac{2\kappa}{t^2k}dt}\tag{\Cref{claim:p-concentrate}} \\
        &= k \cdot \paren*{1/\sqrt{k}+ \frac{2 \kappa}{\sqrt{k}}} = \sqrt{k} \cdot (2 \kappa+1).
    \end{align*}
    Combining the above results, we have thus shown that:
    \begin{align*}
        \Ex_{\bS}\bracket*{\Ex_{\bX \sim \mcD}[G_{\bS}(\bX)]} &= \Ex_{\bX \sim \mcD}\bracket*{F(\bX) \cdot \sum_{i \in [k]} q_{i,\bS}(\bX^{(i)})}\\
        &\leq \abs*{\Ex_{\bX \sim \mcD}\bracket*{F(\bX) \cdot \sum_{i \in [k]} q_{i,\bS}(\bX^{(i)})} - ck \cdot \Ex_{\bX \sim \mcD}[F(\bX)] + ck \cdot \Ex_{\bX \sim \mcD}[F(\bX)]} \\
        &\leq \abs*{\Ex_{\bX \sim \mcD}\bracket*{F(\bX) \cdot \sum_{i \in [k]} q_{i,\bS}(\bX^{(i)})} - ck \cdot \Ex_{\bX \sim \mcD}[F(\bX)]} + ck \cdot \abs*{\Ex_{\bX \sim \mcD}[F(\bX)]} \tag{triangle inequality} \\
        &\leq \Ex_{\bX \sim \mcD}\bracket*{\sum_{i \in [k]}\abs*{q_{i,\bS}(\bX^{(i)}) - c}} + ck \cdot \Ex_{\bX \sim \mcD}[F(\bX)] \\
        &\leq \sqrt{k} \cdot (2\kappa+1) + k \cdot \Ex_{\bX \sim \mcD}[F(\bX)] \tag{$c \leq 1$} \\
        &\leq \sqrt{k} \cdot (2\kappa+2). \tag{We assume $\abs{\Ex_{\bX \sim \mcD}[F(\bX)]} \leq 1/\sqrt{k}$}
    \end{align*}
    \fi

    Finally, we show that the above quantity concentrates over the randomness of the sample $\bS$ using the bounded differences inequality. Consider any training sets $S, S'$ that differ in one data point. Then, $g_{i,S}$ and $g_{i,S'}$ can differ on at most $2$ inputs (corresponding to the $x_i$ and $x_i'$ of the differing point). Each change in $g_i$ can only change $\Ex_{\bX}[G_{S}(\bX)]$ by at most $2 \cdot \kappa/m$, so
    \begin{equation*}
        \Ex_{\bX}[G_{S}(\bX)] - \Ex_{\bX}[G_{S'}(\bX)]  \leq \frac{4k\kappa}{m}.
    \end{equation*}
    \ifnum\focs=1
    Therefore, by \Cref{fact:bounded-diff},
    \begin{align*}
        \Prx_{\bS}&\bracket*{\abs*{\Ex_{\bX \sim \cD}[G_{\bS}(\bX)]} \geq \Ex_{\bS'}\bracket*{\abs*{\Ex_{\bX \sim \cD}[G_{\bS'}(\bX)]}} + \eps} \\
        &\leq \exp\paren*{-\frac{\eps^2m}{8k^2\kappa^2}},
    \end{align*}
    Setting $\eps = O(\sqrt{k}\kappa)$ and using the above bound $\Ex_{\bS'}\bracket*{\abs*{G_{\bS'}(\bX)}} \leq O(\sqrt{k} \kappa)$, we have that
    \begin{equation*}
        \Prx_{\bS}\bracket*{\abs*{\Ex_{\bX \sim \cD}[G_{\bS}(\bX)]} \geq O(\sqrt{k}\kappa)} \leq \exp\paren*{-\Omega \paren*{\frac{m}{k}}}. \qedhere 
    \end{equation*}
    \else
    Therefore, by \Cref{fact:bounded-diff},
    \begin{equation*}
        \Prx_{\bS}\bracket*{\abs*{\Ex_{\bX \sim \cD}[G_{\bS}(\bX)]} \geq \Ex_{\bS'}\bracket*{\abs*{\Ex_{\bX \sim \cD}[G_{\bS'}(\bX)]}} + \eps} \leq \exp\paren*{-\frac{\eps^2m}{8k^2\kappa^2}},
    \end{equation*}
    Setting $\eps = O(\sqrt{k}\kappa)$ and using the above bound $\Ex_{\bS'}\bracket*{\abs*{G_{\bS'}(\bX)}} \leq O(\sqrt{k} \kappa)$, we have that
    \begin{equation*}
        \Prx_{\bS}\bracket*{\abs*{\Ex_{\bX \sim \cD}[G_{\bS}(\bX)]} \geq O(\sqrt{k}\kappa)} \leq \exp\paren*{-\Omega \paren*{\frac{m}{k}}}. \qedhere 
    \end{equation*}
    \fi
\end{proof}
\Cref{claim:g-concentrates-in-expectation} tells us that, with high probability over the draw of the random sample $\bS$, the {\sl expected value} of $G_{\bS}$ concentrates. We want to use it to show \Cref{lem:g-concentrates} which states that with high probability over the draw of the random sample $\bS$, $G_{\bS}$ concentrates on $\cD$.

\begin{proof}
    The quantity we are interested in is $\Prx_{\bX \sim \cD} \left [\abs*{G_{\bS}(\bX)} \leq O(\sqrt{k\log{k}}\kappa) \right ]$. We will bound it by showing that, for any sample $S$, with high probability over $\bX \sim \cD$, $G_{S}(\bX)$ is close to its expectation. Let $\mcU$ be the uniform distribution on $\bitsmk$. We start by proving the simpler statement that $G_{S}(\bX)$ is close to its expectation on the uniform distribution ($\bX \sim \mcU$) then show how this result can be leveraged to prove our original statement. By definition, $G_{S}(X) = \sum_{i \in [k]} g_{i,S}(X^{(i)})$. The advantage of the uniform distribution is that $\bX^{(i)}$ is independent from $\bX^{(j)}$ for all $i \neq j$. Hence $G_{S}(\bX)$ for $\bX$ uniform is a sum of $k$ independent $\bits$ random variables and we can use a Hoeffding bound (\Cref{fact:hoeffding}) to get:
    \ifnum\focs=1
    \begin{equation}
    \label{eq:G-hoeffding}
        \Prx_{\bX } \left [\abs*{ G_{S}(\bX) - \Ex_{\boldsymbol{X'}}[G_{S}(\boldsymbol{X'})] } \geq \eps \right ] \leq 2 \exp \left ( \frac{-\eps^2}{k} \right ).
    \end{equation}
    where both $\bX$ and $\bX'$ are uniformly distributed.
    We also show that the difference in the expected value of $G_{S}$ under $\cD$ and $\cU$ can't be too large.
    \begin{align*}
        &\abs*{\Ex_{\bX \sim \cD}\left [ G_{S}(\bX)\right ] - \Ex_{\bX \sim \cU}\left [ G_{S}(\bX)\right ]} \\ &\leq \Ex_{\bX \sim \cD}\left [ \abs*{ G_{S}(\bX) - \Ex_{\bX \sim \cU}\left [ G_{S}(\bX)\right ] } \right ] \tag{triangle inequality}\\
        &\leq \sqrt{\Ex_{\bX \sim \cD}\left [ \paren*{ G(\bX) - \Ex_{\bX \sim \cU}\left [ G_{S}(\bX)\right ] }^2 \right ] } \tag{Jensen's inequality} \\
        &\leq \sqrt{\kappa \cdot \Ex_{\bX \sim \cU}\left [ \paren*{ G_{S}(\bX) - \Ex_{\bX \sim \cU}\left [ G_{S}(\bX)\right ] }^2 \right ] } \tag{$\cD$ is $\kappa$-smooth} \\
        &\leq \sqrt{\kappa k}, \stepcounter{equation}\tag{\theequation}\label{eq:expectations-close}
    \end{align*}
    \else
    \begin{equation}
    \label{eq:G-hoeffding}
        \Prx_{\bX \sim \mcU} \left [\abs*{ G_{S}(\bX) - \Ex_{\boldsymbol{X'} \sim \mcU}[G_{S}(\boldsymbol{X'})] } \geq \eps \right ] \leq 2 \exp \left ( \frac{-\eps^2}{k} \right ).
    \end{equation}
    
    We also show that the difference in the expected value of $G_{S}$ under $\cD$ and $\cU$ can't be too large.
    \begin{align*}
        \abs*{\Ex_{\bX \sim \cD}\left [ G_{S}(\bX)\right ] - \Ex_{\bX \sim \cU}\left [ G_{S}(\bX)\right ]} &\leq \Ex_{\bX \sim \cD}\left [ \abs*{ G_{S}(\bX) - \Ex_{\bX \sim \cU}\left [ G_{S}(\bX)\right ] } \right ] \tag{triangle inequality}\\
        &\leq \sqrt{\Ex_{\bX \sim \cD}\left [ \paren*{ G(\bX) - \Ex_{\bX \sim \cU}\left [ G_{S}(\bX)\right ] }^2 \right ] } \tag{Jensen's inequality} \\
        &\leq \sqrt{\kappa \cdot \Ex_{\bX \sim \cU}\left [ \paren*{ G_{S}(\bX) - \Ex_{\bX \sim \cU}\left [ G_{S}(\bX)\right ] }^2 \right ] } \tag{$\cD$ is $\kappa$-smooth} \\
        &\leq \sqrt{\kappa k}, \stepcounter{equation}\tag{\theequation}\label{eq:expectations-close}
    \end{align*}
    \fi

where the last inequality comes from the fact that $\Ex_{\bX \sim \cU}\left [ \paren*{ G_{S}(\bX) - \Ex_{\bX \sim \cU}\left [ G_{S}(\bX)\right ] }^2 \right ]$ is the variance of $G_{S}(\bX)$. Since $G_{S}(\bX)$ with $\bX$ uniform is a sum of $k$ independent $\bits$ random variables, its variance is at most $k$.

Going back to our original claim, we can now prove that $G_{S}(\bX)$ concentrates around its expectation. To simplify notation, we will write $\mu_{\cD, S} \coloneqq \Ex_{\bX \sim \cD} \left [ G_{S}(\bX) \right ]$ and analogously for $\mu_{\cU,S}$.
\ifnum\focs=1
\begin{align*}
    \Prx_{\bX \sim \cD}& \left [ \abs*{G_{S}(\bX) - \mu_{\cD,S}} \geq \eps \right ] \\&\leq \Prx_{\bX \sim \cD} \left [ \abs*{G_{S}(\bX) - \mu_{\cU,S}} + \abs*{\mu_{\cU,S} - \mu_{\cD,S}} \geq \eps \right ] \tag{triangle inequality} \\
    &\leq \Prx_{\bX \sim \cD} \left [ \abs*{G_{S}(\bX) - \mu_{\cU,S}} \geq \eps - \sqrt{\kappa k} \right ] \tag{\Cref{eq:expectations-close}} \\
    &\leq \kappa \cdot \Prx_{\bX \sim \cU} \left [ \abs*{G_S(\bX) - \mu_{\cU,S}} \geq \eps - \sqrt{\kappa k} \right ] \tag{$\cD$ is $\kappa$-smooth} \\
    &\leq 2\kappa \cdot \exp \paren*{\frac{-(\eps - \sqrt{\kappa k})^2}{k}}. \tag{\Cref{eq:G-hoeffding}} \\
\end{align*}
\else
\begin{align*}
    \Prx_{\bX \sim \cD} \left [ \abs*{G_{S}(\bX) - \mu_{\cD,S}} \geq \eps \right ] &\leq \Prx_{\bX \sim \cD} \left [ \abs*{G_{S}(\bX) - \mu_{\cU,S}} + \abs*{\mu_{\cU,S} - \mu_{\cD,S}} \geq \eps \right ] \tag{triangle inequality} \\
    &\leq \Prx_{\bX \sim \cD} \left [ \abs*{G_{S}(\bX) - \mu_{\cU,S}} \geq \eps - \sqrt{\kappa k} \right ] \tag{\Cref{eq:expectations-close}} \\
    &\leq \kappa \cdot \Prx_{\bX \sim \cU} \left [ \abs*{G_S(\bX) - \mu_{\cU,S}} \geq \eps - \sqrt{\kappa k} \right ] \tag{$\cD$ is $\kappa$-smooth} \\
    &\leq 2\kappa \cdot \exp \paren*{\frac{-(\eps - \sqrt{\kappa k})^2}{k}}. \tag{\Cref{eq:G-hoeffding}} \\
\end{align*}
\fi
Setting $\eps = \sqrt{k \log{(2k^2\kappa)}} + \sqrt{\kappa k} = O(\sqrt{k \log{k} \kappa})$ we get:
\begin{equation}
\label{eq:G-close-expectation}
    \Prx_{\bX \sim \cD} \left [ \abs*{G_{S}(\bX) - \mu_{\cD,S}} \geq O(\sqrt{k \log{k} \kappa}) \right ] \leq \frac{1}{k^2}.
\end{equation}

We can conclude by using the result from \Cref{claim:g-concentrates-in-expectation}. Since we assume that $\abs{\Ex_{\bX \sim \mcD}[F(\bX)]} \leq 1/\sqrt{k}$, then we know that with probability at least $1 - \exp \paren*{-\Omega(\frac{m}{k})}$ over the draw of the random sample $\bS$,
    \begin{equation*}
        \abs*{\Ex_{\bX \sim \mcD}[G_{\bS}(\bX)]} \leq O(\sqrt{k}\kappa).
    \end{equation*}
    
Plugging this into \Cref{eq:G-close-expectation}, we get that, with probability at least  $1 - \exp \paren*{-\Omega(\frac{m}{k})}$ over the draw of the random sample $\bS$,
\ifnum\focs=1
\begin{align*}
    \Prx_{\bX \sim \cD}& \left [ \abs*{G_{\bS}(\bX)} \geq O(\sqrt{k \log{k} \kappa}) \right ] \\ &\leq \Prx_{\bX \sim \cD} \left [ \abs*{G_{\bS}(\bX) - O(\sqrt{k\kappa})} \geq O(\sqrt{k \log{k} \kappa}) \right ] \\
    &\leq \frac{1}{k^2}.
\end{align*}
\else
\begin{align*}
    \Prx_{\bX \sim \cD} \left [ \abs*{G_{\bS}(\bX)} \geq O(\sqrt{k \log{k} \kappa}) \right ] \leq \Prx_{\bX \sim \cD} \left [ \abs*{G_{\bS}(\bX) - O(\sqrt{k\kappa})} \geq O(\sqrt{k \log{k} \kappa}) \right ] \leq \frac{1}{k^2}.
\end{align*}
\fi
We conclude by using the fact that $\uint = \uintval$.
\end{proof}

\subsection{The weak learner is weakly correlated with $F$}
    We start by showing a useful property of the sign function, namely, the expectation over a continuous interval $\btau \sim [-a, a]$ of the sign function $\sign(y \geq \btau)$ looks like a linear version of a threshold function. More formally,
    \begin{claim}[Expected value of sign functions over a symmetric interval]
    \label{claim:ev-sign-fctn} The following holds,
        \begin{equation*}
            \Ex_{\btau \sim [-a, a]}\left [\sign(y \geq \btau) \right ] = \begin{cases*}
                -1 & if $y < -a$ \\
                y/a & if $y \in [\pm a]$ \\
                1 & if $y > a$
            \end{cases*}
        \end{equation*}
    \end{claim}
\begin{proof}
    The cases where $y \notin [\pm a]$ are immediate. For the case where $y \in [\pm a]$ we have that:
    \ifnum\focs=1
    \begin{align*}
        \Ex_{\btau \sim [-a, a]}&\left [\sign(y \geq \btau) \, | \, y \in [\pm a] \right ] \\
        &= 1 \cdot \Prx_{\btau \sim [-a, a]}[y \geq \btau] + (-1) \cdot \Prx_{\btau \sim [-a, a]}[y < \btau] \\
        &= \frac{y + a}{2a} - \frac{a - y}{2a} \\
        &= \frac{y}{a}.
    \end{align*}
    \else
    \begin{align*}
        \Ex_{\btau \sim [-a, a]}\left [\sign(y \geq \btau) \, | \, y \in [\pm a] \right ] &= 1 \cdot \Prx_{\btau \sim [-a, a]}[y \geq \btau] + (-1) \cdot \Prx_{\btau \sim [-a, a]}[y < \btau] \\
        &= \frac{y + a}{2a} - \frac{a - y}{2a} \\
        &= \frac{y}{a}. \qedhere
    \end{align*}
    \fi
\end{proof}

\begin{lemma}[Existence of a weakly correlated hypothesis]
\label{lem:weakly-correlated-h-exists-if-f-low-bias}
        Let $\uint = \uintval$ and $\mcH$ be the set of hypotheses, as defined in \Cref{fig:weak-learner}. Assume that $\Ex_{\bX \sim \mcD}\bracket*{F(\bX) \cdot G_{S}(\bX)} \geq \Omega(1/\kappa^3)$ and $\Prx_{\bX \sim \cD}\bracket*{\abs*{G_{S}(\bX)} \leq u} \geq 1 - 1/k^2$. Then, there exists an $h^* \in \mcH$ such that:
    \begin{equation*}
        \Ex_{\bX \sim \cD}[F(\bX) \cdot h^*(\bX)] \geq \Omega \paren*{\frac{1}{\uint \kappa^3}}.
    \end{equation*}
\end{lemma}
\begin{proof}
    We will consider the correlation achieved by weak learning algorithm's threshold functions in expectation. Here we are taking the expectation over continuous threshold functions $h_\tau$ for $\tau \in \left [\lint, \uint \right ]$ instead of over integers. We will show below that since $G_{S}(X)$ only outputs integer values, this is equivalent to having $\tau$ only take integer values in the interval. To simplify notation, we will write $\Ex_{\btau} \coloneqq \Ex_{\btau \sim \left [\lint, \uint \right ]}$. 

    \ifnum\focs=1
    \begin{align*}
    \Ex_\btau &\left [\Ex_{\bX }[F(\bX) \cdot h_\btau(\bX)] \right ] \\&= \Ex_{\bX } \left [F(\bX) \cdot \Ex_\btau[h_\btau(\bX)] \right ] \\
    &= \Ex_{\bX } \left [F(\bX) \cdot \Ex_\btau[\sign(G_{S}(\bX) \geq \btau)] \right ]  \tag{definition of $h_\tau$}\\
    &= \Ex_{\bX } \left [-F(\bX) \, | \, G_{S}(\bX) < \lint \right] \Prx_{\bX }[G_{S}(\bX) < \lint]+\\
     & \quad\Ex_{\bX } \left [F(\bX) \, | \, G_{S}(\bX) > \uint \right] \Prx_{\bX }[G_{S}(\bX) > \uint]+\\
         & \quad \bigg(\Ex_{\bX } \left [\frac{F(\bX) \cdot G_{S}(\bX)}{\uint} \, \bigg | \, G_{S}(\bX) \in \left [\pm \uint \right ]  \right] \cdot\\
         &\quad\Prx_{\bX } \left [G_{S}(\bX) \in \left [\pm \uint \right ] \right ]\bigg)
    \end{align*}
    \else
    \begin{align*}
    \Ex_\btau \left [\Ex_{\bX \sim \cD}[F(\bX) \cdot h_\btau(\bX)] \right ] &= \Ex_{\bX \sim \cD} \left [F(\bX) \cdot \Ex_\btau[h_\btau(\bX)] \right ] \\
    &= \Ex_{\bX \sim \cD} \left [F(\bX) \cdot \Ex_\btau[\sign(G_{S}(\bX) \geq \btau)] \right ]  \tag{definition of $h_\tau$}.\\
    \begin{split}
    &= \Ex_{\bX \sim \cD} \left [-F(\bX) \, | \, G_{S}(\bX) < \lint \right] \Prx_{\bX \sim \cD}[G_{S}(\bX) < \lint]\\
     & \quad + \Ex_{\bX \sim \cD} \left [F(\bX) \, | \, G_{S}(\bX) > \uint \right] \Prx_{\bX \sim \cD}[G_{S}(\bX) > \uint]\\
         & \quad + \Ex_{\bX \sim \cD} \left [\frac{F(\bX) \cdot G_{S}(\bX)}{\uint} \, \bigg | \, G_{S}(\bX) \in \left [\pm \uint \right ]  \right] \Prx_{\bX \sim \cD} \left [G_{S}(\bX) \in \left [\pm \uint \right ] \right ],
    \end{split}
    \end{align*}
    \fi
where the last step uses \Cref{claim:ev-sign-fctn}.

Since we are interested in showing that the expected correlation of the $h_\tau$ with $F$ is high, we provide lower bounds for the 3 terms on the right. The first 2 are immediate. Indeed, we are assuming that $\Prx_{\bX \sim \cD}\bracket*{\abs*{G_{S}(\bX)} \leq u} \geq 1 - 1/k^2$ and we have that $F(X) \in \{-1, 1 \}$, hence:
\ifnum\focs=1
\begin{align*}
    \Ex_{\bX \sim \cD} &\left [-F(\bX) \, | \, G_{S}(\bX) < \lint \right] \Prx_{\bX \sim \cD}[G_{S}(\bX) < \lint] \\&\geq - \pconcentrate.
\end{align*}
\else
    \begin{equation*}
        \Ex_{\bX \sim \cD} \left [-F(\bX) \, | \, G_{S}(\bX) < \lint \right] \Prx_{\bX \sim \cD}[G_{S}(\bX) < \lint] \geq - \pconcentrate.
    \end{equation*}
\fi
The same holds for the second term.

It remains to lower bound 
\ifnum\focs=1
\begin{align*}
    \Ex_{\bX} &\left [\frac{F(\bX) \cdot G_{S}(\bX)}{\uint} \, \bigg | \, G_{S}(\bX) \in \left [\pm \uint \right ]  \right] \cdot\\ &\Prx_{\bX} \left [G_{S}(\bX) \in \left [\pm \uint \right ] \right ]
\end{align*}
\else
$$\Ex_{\bX} \left [\frac{F(\bX) \cdot G_{S}(\bX)}{\uint} \, \bigg | \, G_{S}(\bX) \in \left [\pm \uint \right ]  \right] \Prx_{\bX} \left [G_{S}(\bX) \in \left [\pm \uint \right ] \right ]$$ 
\fi
where $\bX \sim \cD$. By assumption, we have that $\Prx_{\bX \sim \cD} \left [G_{S}(\bX) \in \left [\pm \uint \right ] \right ] \geq 1 - 1/k^2$. We know from \Cref{lem:weak-correlation} that $F$ and $G$ are weakly correlated on $\cD$, so we have to show that most of this correlation remains when we condition on $G_{S}(\bX) \in \left [\pm \uint \right ]$. To show this, we will use the same trick as earlier by conditioning the expectation.
\ifnum\focs=1
\begin{align*}
        \Ex_{\bX \sim \cD}& \left [F(\bX) \cdot G_{S}(\bX) \right]\\
        &= \bigg(\Ex_{\bX \sim \cD} \left [F(\bX) \cdot G_{S}(\bX) \, | \, G_{S}(\bX) \in \left [\pm \uint \right ]  \right] \\ &\quad \cdot\Prx_{\bX \sim \cD} \left [G_{S}(\bX) \in \left [\pm \uint \right ] \right ]\bigg) \\
        & \quad + \bigg(\Ex_{\bX \sim \cD} \left [F(\bX) \cdot G_{S}(\bX) \, | \, G_{S}(\bX) \notin \left [\pm \uint \right ]  \right] \\
        &\quad \cdot \Prx_{\bX \sim \cD} \left [G_{S}(\bX) \notin \left [\pm \uint \right ] \right ]\bigg) \\
        & \leq \Ex_{\bX \sim \cD} \left [F(\bX) \cdot G_{S}(\bX) \, | \, G_{S}(\bX) \in \left [\pm \uint \right ]  \right] + \frac{2}{k},
\end{align*}
\else
\begin{align*}
        \Ex_{\bX \sim \cD} \left [F(\bX) \cdot G_{S}(\bX) \right] &= \Ex_{\bX \sim \cD} \left [F(\bX) \cdot G_{S}(\bX) \, | \, G_{S}(\bX) \in \left [\pm \uint \right ]  \right]  \Prx_{\bX \sim \cD} \left [G_{S}(\bX) \in \left [\pm \uint \right ] \right ] \\
    \begin{split}
        & \quad + \Ex_{\bX \sim \cD} \left [F(\bX) \cdot G_{S}(\bX) \, | \, G_{S}(\bX) \notin \left [\pm \uint \right ]  \right]  \Prx_{\bX \sim \cD} \left [G_{S}(\bX) \notin \left [\pm \uint \right ] \right ]
    \end{split} \\
        & \leq \Ex_{\bX \sim \cD} \left [F(\bX) \cdot G_{S}(\bX) \, | \, G_{S}(\bX) \in \left [\pm \uint \right ]  \right] + \frac{2}{k},
\end{align*}
\fi
where the last inequality is given by the fact that $G$ is upper bounded by $k$ and $F$ is upper bounded by $1$, and the fact that $\Prx_{\bX \sim \cD}\bracket*{\abs*{G_{S}(\bX)} \notin [\pm u]} \leq 1/k^2$. Using our assumption that $\Ex_{\bX \sim \mcD}\bracket*{F(\bX) \cdot G_{S}(\bX)} \geq \Omega(1/\kappa^3)$, we can rearrange and find:

\ifnum\focs=1
\begin{align*}
    \Ex_{\bX \sim \cD}& \left [F(\bX) \cdot G_{S}(\bX) \, | \, G_{S}(\bX) \in \left [\pm \uint \right ]  \right] \\&\geq \Ex_{\bX \sim \cD} \left [F(\bX) \cdot G_{S}(\bX) \right] - \frac{2}{k} \\
    &\geq \wcorr - \frac{2}{k}.
\end{align*}

We can now conclude:
\begin{align*}
    \Ex_\btau& \left [\Ex_{\bX \sim \cD}[F(\bX) \cdot h_\btau(\bX)] \right ] \\&= \Ex_{\bX \sim \cD} \left [F(\bX) \cdot \Ex_\btau[\sign(G_{S}(\bX) \geq \btau)] \right ] \\
    &= \Ex_{\bX \sim \cD} \left [-F(\bX) \, | \, G(\bX) < \lint \right] \Prx_{\bX \sim \cD}[G_{S}(\bX) < \lint]\\
     & \quad + \Ex_{\bX \sim \cD} \left [F(\bX) \, | \, G_{S}(\bX) > \uint \right] \Prx_{\bX \sim \cD}[G_{S}(\bX) > \uint]\\
         & \quad + \bigg(\Ex_{\bX \sim \cD} \left [\frac{F(\bX) \cdot G_{S}(\bX)}{\uint} \, \bigg | \, G_{S}(\bX) \in \left [\pm \uint \right ]  \right] \\&\qquad\cdot\Prx_{\bX \sim \cD} \left [G_{S}(\bX) \in \left [\pm \uint \right ] \right ]\bigg)\\
    &\geq - \frac{2}{k^2} + \Omega \paren*{\frac{1}{\uint \kappa^3}} - \frac{2}{\uint k} \\
    &\geq \Omega \paren*{\frac{1}{\uint \kappa^3}},
    \end{align*}
\else
\begin{align*}
    \Ex_{\bX \sim \cD} \left [F(\bX) \cdot G_{S}(\bX) \, | \, G_{S}(\bX) \in \left [\pm \uint \right ]  \right] &\geq \Ex_{\bX \sim \cD} \left [F(\bX) \cdot G_{S}(\bX) \right] - \frac{2}{k} \\
    &\geq \wcorr - \frac{2}{k}.
\end{align*}

We can now conclude:
\begin{align*}
    \Ex_\btau \left [\Ex_{\bX \sim \cD}[F(\bX) \cdot h_\btau(\bX)] \right ] &= \Ex_{\bX \sim \cD} \left [F(\bX) \cdot \Ex_\btau[\sign(G_{S}(\bX) \geq \btau)] \right ] \\
    \begin{split}
    &= \Ex_{\bX \sim \cD} \left [-F(\bX) \, | \, G(\bX) < \lint \right] \Prx_{\bX \sim \cD}[G_{S}(\bX) < \lint]\\
     & \quad + \Ex_{\bX \sim \cD} \left [F(\bX) \, | \, G_{S}(\bX) > \uint \right] \Prx_{\bX \sim \cD}[G_{S}(\bX) > \uint]\\
         & \quad + \Ex_{\bX \sim \cD} \left [\frac{F(\bX) \cdot G_{S}(\bX)}{\uint} \, \bigg | \, G_{S}(\bX) \in \left [\pm \uint \right ]  \right] \Prx_{\bX \sim \cD} \left [G_{S}(\bX) \in \left [\pm \uint \right ] \right ]
    \end{split}\\
    &\geq - \frac{2}{k^2} + \Omega \paren*{\frac{1}{\uint \kappa^3}} - \frac{2}{\uint k} \\
    &\geq \Omega \paren*{\frac{1}{\uint \kappa^3}},
    \end{align*}
\fi 
    
    where in the last line we use the definition of $\uint$ ($\uint = \uintval$).
    
In particular, this implies that there exists $h_\tau^*$ for $\tau \in \left [\lint, \uint \right ]$ such that 
\begin{equation*}
\Ex_{\bX \sim \cD}[F(\bX) \cdot h_{\tau}^*(\bX)] \geq \Omega \paren*{\frac{1}{\uint \kappa^3}}.
\end{equation*}
We remark that since $G_S(X)$ only outputs integer values,
\ifnum\focs=1
\begin{align*}
h_{\tau}^*(\bX) &= \sign(G_{S}(\bX) \geq \tau) \\ &= \sign(G_{S}(\bX) \geq \lceil \tau \rceil )\\ &= h_{\lceil \tau \rceil}(\bX).
\end{align*}
\else
$$h_{\tau}^*(\bX) = \sign(G_{S}(\bX) \geq \tau) = \sign(G_{S}(\bX) \geq \lceil \tau \rceil ) = h_{\lceil \tau \rceil}(\bX).$$
\fi 
This implies that there exists $h_\tau^*$ for $\tau \in \{-u, \dots, u\}$ such that
\begin{equation*}
\Ex_{\bX \sim \cD}[F(\bX) \cdot h_{\tau}^*(\bX)] \geq \Omega \paren*{\frac{1}{\uint \kappa^3}}.
\end{equation*}
We conclude by noting that $h_\tau$ for $\tau \in \{-u, \dots, u\}$ are hypotheses in $\mcH$.
\end{proof}

\subsection{Proof of \Cref{lem:boosting-upper}}
\label{section:unif-convergence-weak-learner}
The previous results imply that there exists a hypothesis $h \in \mcH$ that achieves weak correlation with $F$ with high probability. All that remains to show is that if there exists such a good hypothesis in $\mcH$, then the hypothesis that the weak learning algorithm chooses also achieves weak correlation with $F$. To do this, we show that the hypothesis class we are using satisfies uniform convergence.
\begin{definition}[Loss function]
    We define the \emph{loss} of a classifier $h$ w.r.t.~a target function $F$ over a distribution $\cD$ as: 
    \begin{equation*}
        L_{\cD, F}(h) \coloneqq \Prx_{\bX \sim \cD}[h(\bX) \neq F(\bX)].
    \end{equation*}
    We overload the notation to define the empirical loss of $h$ w.r.t.~$F$ over a set of inputs $S$:
    \begin{equation*}
        L_{S, F}(h) = \frac{1}{|S|} \sum_{X \in S} \Ind[h(X) \neq F(X)].
    \end{equation*}
    We will omit the subscript for the target function $F$ when it is clear from context.
\end{definition}
In what follows, we will use a standard result from learning theory that states that for finite hypothesis classes, the error of the hypotheses on the training set is close to their error on the underlying distribution with high probability.
\begin{lemma}[Uniform convergence for finite hypothesis classes, section 4.2 from \cite{shalev14:understanding}]
\label{lemma:finite-class-unif-convergence}
    Let $\mcH$ be a finite hypothesis class, $Z$ a domain and $L: \mcH \times Z \rightarrow [0,1]$ a loss function then, $\mcH$ has the uniform convergence property with sample complexity $\left \lceil \frac{\log (2 |\mcH|/\delta)}{2 \eps^2} \right \rceil$. Formally, this means that if $\bS$ is a sample of $m \geq \left \lceil \frac{\log (2 |\mcH|/\delta)}{2 \eps^2} \right \rceil$ examples drawn i.i.d.~according to $\mcD$, then with probability at least $1 - \delta$ over the random sample $\bS$, we have that:
    \begin{equation*}
        \text{For all } h \in \mcH, |L_{\bS}(h) - L_\cD(h)| \leq \eps.
    \end{equation*}
\end{lemma}

\begin{corollary}[The weak learner converges uniformly]
\label{corollary:weak-learner-unif-convergence}
    Let $\mcH = \{h_\tau(\boldsymbol{S_\mathrm{train}}) \, | \, \tau \in \{\lint,\ldots, \uint\}\}$ be the threshold functions constructed in \Cref{fig:weak-learner}. Let $F$ be the $\Maj(f_1, \dots, f_k)$ function from \Cref{thm:boosting-formal}. If $|\boldsymbol{S_\mathrm{val}}| \geq \Omega \left ( \frac{\log (k \kappa/\delta)}{\eps^2} \right )$ then, with probability $1 - \delta$ over the random sample $\boldsymbol{S_\mathrm{val}}$ we have that:
    \begin{equation*}
        \text{For all } h \in \mcH: |L_{\boldsymbol{S_\mathrm{val}}, F}(h) - L_{\cD, F}(h)| \leq \eps.
    \end{equation*}
\end{corollary}

\begin{proof}
    The proof is immediate by applying \Cref{lemma:finite-class-unif-convergence} using the fact that the hypothesis class $\mcH$ has size $2\uint$.
\end{proof}

\begin{corollary}
\label{corollary:erm-learner-is-close}
    Let $\mcH = \{h_\tau(\boldsymbol{S_\mathrm{train}}) \, | \, \tau \in \{\lint,\ldots, \uint\}\}$ be the threshold functions constructed in \Cref{fig:weak-learner}. Let $F$ be the $\Maj(f_1, \dots, f_k)$ function from \Cref{thm:boosting-formal}. Note $h_{\boldsymbol{S_\mathrm{val}}} \in \arg \min L_{\boldsymbol{S_\mathrm{val}}}(h)$ the hypothesis chosen by the weak learning algorithm. If $|\boldsymbol{S_\mathrm{val}}| \geq \Omega \left ( \frac{\log (k \kappa /\delta)}{\eps^2} \right )$ then, with probability at least $1 - 2\delta$ over the random sample $\boldsymbol{S_{val}}$, we have that:
    \begin{equation*}
        L_\cD(h_{\boldsymbol{S_\mathrm{val}}}) \leq \min_{h \in \mcH} L_\cD(h) + 2 \eps.
    \end{equation*}
\end{corollary}

\begin{proof}
    For every hypothesis $h \in \mcH$, we have:
    \begin{equation*}
         L_\cD(h_{\boldsymbol{S_\mathrm{val}}}) \leq  L_{\boldsymbol{S_\mathrm{val}}}(h_{\boldsymbol{S_\mathrm{val}}}) + \eps \leq L_{\boldsymbol{S_\mathrm{val}}}(h) + \eps \leq L_\cD(h) + 2 \eps,
    \end{equation*}
    where the first and third inequalities come from \Cref{corollary:weak-learner-unif-convergence} and the second inequality holds by construction of $h_{\boldsymbol{S_\mathrm{val}}}$.
\end{proof}

We can now combine our results to prove \Cref{lem:boosting-upper}. We will show the following lemma:

\begin{lemma}[Precise statement of \Cref{lem:boosting-upper}] Let $h$ be the hypothesis returned by \Cref{fig:weak-learner}. For any constant $c$, with probability at least $1 - m^{-c}$ over the draw of the random sample $\bS$,
\begin{equation*}
     \Ex_{\bX \sim \cD}[F(\bX) \cdot h(\bX)] \geq \Omega \paren*{\advantage}.
\end{equation*}
\end{lemma}

\begin{proof}
    There are two possible cases. If $\abs{\Ex_{\bX \sim \mcD}[F(\bX)]} > \frac{1}{\sqrt{k}}$ then as shown in the proof overview, this implies that there exists a hypothesis $h^* \in \mcH$ that achieves correlation $\frac{1}{\sqrt{k}}$. If $\abs{\Ex_{\bX \sim \mcD}[F(\bX)]} \leq \frac{1}{\sqrt{k}}$ then we know that \Cref{lem:weak-correlation,lem:g-concentrates} hold with probability at least $1 - \exp \paren*{-\Omega(\frac{m}{k})}$ over the random sample $\bS$. In the case they both hold, we can apply \Cref{lem:weakly-correlated-h-exists-if-f-low-bias} to conclude that there exists a hypothesis $h^* \in \mcH$ such that $\Ex_{\bX \sim \cD}[F(\bX) \cdot h^*(\bX)] \geq \Omega \paren*{\advantage}$.

    Combining both cases, we get that, with probability at least $1 - \exp \paren*{-\Omega(\frac{m}{k})}$ over the random sample $\bS$, there exists a hypothesis $h^* \in \mcH$ such that $\Ex_{\bX \sim \cD}[F(\bX) \cdot h^*(\bX)] \geq \Omega \paren*{\advantage}$. Using the assumption that $n \geq \Omega(\log k)$ and $m =  O(2^n)$, the proof is then a straightforward application of \Cref{corollary:erm-learner-is-close}, setting $\eps = O \paren*{\advantage}$ and $\delta = m^{-c}$.
\end{proof}

%% file: Boosting.tex
\section{The sample complexity overhead of distribution-independent boosting}
\label{sec:lower-bound-boosting-overhead}
%We are asking the question: in the setting of \cite{boosting_simple_learners}, can the same sample complexity bound be achieved by a smooth-boosting algorithm. More formally, we would like to prove a claim of the following form:

%This would mean that the advantages of smooth boosting (noise resistance for examples) come at the cost of sample inefficiency.

\begin{claim}[Lower bound on the sample complexity of strong learning relative to the sample complexity of weak learning]
\label{claim:simple-gap-smooth}
    Let $\mcX$ be any domain of size $m$. Let $\mcC$ be the class of all functions over $\mcX$. Then for any $\gamma > 0$, the following facts are true:
\begin{enumerate}
        \item For any distribution $\cD$ over $\mcX$, there exists a weak learner that can learn $\mcC$ to accuracy $1/2 + \gamma$ using $O(\gamma m)$ samples with high probability.
        \item Learning $\mcC$ to accuracy 0.99 with respect to the uniform distribution requires $\Omega(m)$ samples.
    \end{enumerate}\end{claim}
Note that this lower bound implies that any booster incurs a sample complexity overhead of $\Omega(1/\gamma)$. In particular, this applies to smooth boosters.
\begin{proof}
    
    The first fact is shown by a weak learner $\mcA$ that memorizes the labels for the $O(\gamma m)$ samples it sees then returns a random bit for the inputs it didn't memorize. Note that $\mcA$ is a randomized hypothesis, we will show how to derandomize it below.
    
    We note $\bS$ the random sample, and use $l$ to denote the number of samples in $\bS$. By construction, $l = O(\gamma m)$.

    Our weak learning algorithm always answers correctly for elements that are in the sample. Consequently, to prove that it achieves good accuracy, we need to show that, with high probability over the sampling procedure, $\Prx_{\bx \sim \mcD}[\bx \in \bS] \geq 2\gamma$.

    We define $G$, the set of ``good'' $x$'s as $G = \{x \in \mcX \mid \Prx_{\bx \sim \cD}[\bx=x] \geq \frac{1}{2m}\}$. Note that since the points $x \notin G$ all have weight less than $\frac{1}{2m}$ and there are at most $m$ of them then we have that $\Prx_{x \sim \cD}[x \in G] \geq \frac{1}{2}$.

    We start by showing that, in expectation over the random sample $\bS$, at least half of the elements in $\bS$ are from $G$.

    \begin{align*}
        \Ex_{\bS \sim \cD^l}\bracket*{\sum_{x \in \bS} \Ind\{x \in G\}} &= \sum_{i = 1}^l \Prx_{\bx \sim \cD} [\bx \in G] \\
        &\geq \frac{l}{2} \tag{Using $\Prx_{\bx \sim \cD}[\bx \in G] \geq \frac{1}{2}$}.
    \end{align*}

    By a Hoeffding bound (\Cref{fact:hoeffding}), we get that:
    \ifnum\focs=1
    \begin{align*}
        \Prx_{\bS}&\bracket*{\abs*{\sum_{x \in \bS} \Ind\{x \in G\} - \Ex_{\bS'}\bracket*{\sum_{x \in \bS'} \Ind\{x \in G\}}} \geq \frac{l}{4}} \\ &\leq 2 \exp \paren*{\frac{-l}{16}}
    \end{align*}
    where $\bS,\bS'\sim \cD^l$.
    \else
    \begin{align*}
        \Prx_{\bS \sim \cD^l}\bracket*{\abs*{\sum_{x \in \bS} \Ind\{x \in G\} - \Ex_{\bS' \sim \cD^l}\bracket*{\sum_{x \in \bS'} \Ind\{x \in G\}}} \geq \frac{l}{4}} \leq 2 \exp \paren*{\frac{-l}{16}}.
    \end{align*}
    \fi
    
    Thus, with high probability over the random sample $\bS$, at least a fourth of the elements in $\bS$ have probability at least $\frac{1}{2m}$.

    It remains to show that, conditioned on that event, $\Prx_{\bx \sim \cD}[\bx \in \bS] \geq 2\gamma$. This follows since
    \begin{align*}
        \Prx_{\bx\sim\mcD}\bracket*{\bx\in S \, \bigg | \, \sum_{x \in S} \Ind\{x \in G\} \geq \frac{l}{4}} &\geq \frac{l}{4} \cdot \frac{1}{2m} \\
        &\geq 2\gamma \tag{$l = O(\gamma m)$}.
    \end{align*}
    Thus, with probability at least $1 - \exp\paren*{-\Omega(l)}$, $\Prx_{\bx\sim\mcD}\bracket*{\bx\in S} \geq 2 \gamma.$
    
    Let $c \in \mcC$ be the target concept. We now show how to derandomize our hypothesis. Since $\mcA(x)$ returns a random bit on an input $x \notin S$, by symmetry, we have that:
    \ifnum\focs=1
    \begin{align*}
        \Pr&\bracket*{\Ex_{\bx \sim \cD}[\mcA(\bx)c(\bx) \mid \bx \notin \bS] \geq 0} \\ &=  \Pr\bracket*{\Ex_{\bx \sim \cD}[\mcA(\bx)c(\bx) \mid \bx \notin \bS] \leq 0},
    \end{align*}
    \else
    $$\Pr\bracket*{\Ex_{\bx \sim \cD}[\mcA(\bx)c(\bx) \mid \bx \notin \bS] \geq 0} =  \Pr\bracket*{\Ex_{\bx \sim \cD}[\mcA(\bx)c(\bx) \mid \bx \notin \bS] \leq 0},$$
    \fi
    where the randomness is taken over the coin flips from the random hypothesis on inputs not in $S$.
    This implies that with probability at least $1/2$, the hypothesis returned is such that $\Ex_{\bx \sim \cD}[\mcA(\bx)c(\bx) \mid \bx \notin \bS] \geq 0$.

    If we assume such a ``good'' hypothesis is chosen by $\mcA$, we can now conclude that $\mcA$ will have $2\gamma$ correlation with $c$:
    \ifnum\focs=1
    \begin{align*}
    \Ex_{\bx \sim \cD}&[\mcA(\bx)c(\bx)] \\&= \Ex_{\bx \sim \cD}[\mcA(\bx)c(\bx) \mid \bx \in \bS] \Prx_{\bx \sim \cD}[\bx \in \bS] \\&\quad + \Ex_{\bx \sim \cD}[\mcA(\bx)c(\bx) \mid \bx \notin \bS] \Prx_{\bx \sim \cD}[\bx \notin \bS] \\
    &\geq 1 \cdot \Prx_{\bx \sim \cD}[\bx \in \bS] + 0 \tag{$\mcA$ learns perfectly on samples in $S$}\\
    &\geq 2 \gamma.
    \end{align*}
    \else
    \begin{align*}
    \Ex_{\bx \sim \cD}[\mcA(\bx)c(\bx)] &= \Ex_{\bx \sim \cD}[\mcA(\bx)c(\bx) \mid \bx \in \bS] \Prx_{\bx \sim \cD}[\bx \in \bS] + \Ex_{\bx \sim \cD}[\mcA(\bx)c(\bx) \mid \bx \notin \bS] \Prx_{\bx \sim \cD}[\bx \notin \bS] \\
    &\geq 1 \cdot \Prx_{\bx \sim \cD}[\bx \in \bS] + 0 \tag{$\mcA$ learns perfectly on samples in $S$}\\
    &\geq 2 \gamma.
    \end{align*}
    \fi
    Thus, with high probability, $\mcA$ learns $\mcC$ to accuracy $1/2 + \gamma$.

    The second fact uses the fundamental theorem of PAC learning that states that learning a concept class $\mcC$ with VC dimension $d$ to accuracy 0.99 requires $\Omega\paren*{d}$ samples, see for example \cite{shalev14:understanding} theorem 6.8. Since $\mcC$ is defined as the class of all functions over the domain $\mcX$, it has VC dimension $m$. Thus, learning to accuracy 0.99 requires $\Omega(m)$ samples.
\end{proof}

% \begin{claim}
%     Let $\gamma>0$ and let $\mathcal{C}$ be a class of VC-dimension $d\in \N$. If $L$ is a distribution-independent $\gamma$-advantage weak learner with sample complexity $m$. Then $m\ge \gamma d$. 
% \end{claim}
%\begin{proof}
%    Let $H=\{x_1,\ldots,x_d\}$ be a shattering set of $d$ points and let $\mcD$ be the uniform distribution over $T$. For a set of $m$ samples $S\sse H$, we write $\bh^S$ to denote the hypothesis output by the weak learner after seeing the samples $S$. The hypothesis is randomized since the learning algorithm is randomized. The weak learner succeeds with some fixed confidence, say $0.99$, which implies that
%    $$
%    0.99\le \Prx_{\bS\text{ size }s}
%    $$
    
%    Suppose the algorithm uses $m$ samples and outputs some hypothesis $h$.

%    Since $H$ is a shattering set, there is some $c\in\mcC$ such that $c(x)\neq h(x)$ for all $x\in S$ which is not one of the $m$ samples. Let $\bh^{}$ denote the random variable which is the output of the weak learning algorithm (where the randomness of the learning algorithm has been incorporated into $\bh$). Since the weak learning outputs 
%\end{proof}

\begin{claim}[Upper bound on the sample of complexity of strong learning relative to the sample complexity of weak learning]
\label{claim:ub-sample-cxty-strong}
    Let $\mcC$ be a concept class and let $\gamma>0$. 
    If the sample complexity of weak learning $\mathcal{C}$ to accuracy $1/2 + \gamma$ in the distribution-independent setting is $m$, then the sample complexity of strong learning $\mathcal{C}$ to accuracy 0.99 in the distribution-independent setting is $O(m/\gamma)$.
\end{claim}
\begin{proof}
    We show that $m\ge \gamma d$ where $d$ is the VC dimension of $\mcC$. The claim then follows by the fact that the VC dimension characterizes the sample complexity of strong learning to constant accuracy. Let $H=\{x_1,\ldots,x_d\}$ be a shattering set of $d$ points and let $\mcD$ be the uniform distribution over $H$. For a set of $m$ samples $S\sse H$, we write $\bh^S$ to denote the hypothesis output by the weak learner after seeing the samples $S$ (the hypothesis is randomized to incorporate the randomness of the learning algorithm). First, we observe that there is a concept $c\in\mcC$ such that for all $x\in H\setminus S$, $\Ex[\bh^S(x)c(x)]\le 0$ and $c$ is consistent with the labels of the points in $S$. This follows from the fact $H$ is a shattering set, so for every labeling of the points in $H\setminus S$, there is a concept $c\in\mcC$ that witnesses the labeling, and therefore, it is possible to choose $c$ so that $\Ex[\bh^S(x)c(x)]\le 0$. In fact, this shows that the best choice of $\bh^S(x)$ is to output a random bit. It follows that for all samples $S$ of size $m$, there is a concept $c\in\mcC$ such that:
    \ifnum\focs=1
    \begin{align*}
        \Ex_{\bx\sim\mcD}&[\bh^S(\bx)c(\bx)]\\
        &= \Prx_{\bx\sim\mcD}[\bx\in S]\Ex_{\bx\sim\mcD}[\bh^S(\bx)c(\bx)\mid \bx\in S]\\
        &\quad +\Prx_{\bx\sim\mcD}[\bx\not\in S]\Ex_{\bx\sim\mcD}[\bh^S(\bx)c(\bx)\mid \bx\not\in S]\\
        &\le \Pr[\bx\in S]=\frac{m}{d}.
    \end{align*}
    where the last inequality uses the fact that $\Ex_{\bx\sim\mcD}[\bh^S(\bx)c(\bx)\mid \bx\not\in S]\le 0$ and $\Ex_{\bx\sim\mcD}[\bh^S(\bx)c(\bx)\mid \bx\in S]\le 1$.
    \else
    \begin{align*}
        \Ex_{\bx\sim\mcD}[\bh^S(\bx)c(\bx)]&= \Prx_{\bx\sim\mcD}[\bx\in S]\Ex_{\bx\sim\mcD}[\bh^S(\bx)c(\bx)\mid \bx\in S]+\Prx_{\bx\sim\mcD}[\bx\not\in S]\Ex_{\bx\sim\mcD}[\bh^S(\bx)c(\bx)\mid \bx\not\in S]\\
        &\le \Pr[\bx\in S]=\frac{m}{d}.\tag{$\Ex_{\bx\sim\mcD}[\bh^S(\bx)c(\bx)\mid \bx\not\in S]\le 0$ and $\Ex_{\bx\sim\mcD}[\bh^S(\bx)c(\bx)\mid \bx\in S]\le 1$}
    \end{align*}
    \fi
    Finally, since $\Ex_{\bx\sim\mcD}[\bh^S(\bx)c(\bx)] = 2 \Prx_{\bx\sim\mcD}[\bh^S(\bx) = c(\bx)] -1$, we can conclude that if the weak learning algorithm achieves accuracy $1/2 + \gamma$ then $\gamma\le m/d$ as desired. 
\end{proof}

%% file: main.bbl
\newcommand{\etalchar}[1]{$^{#1}$}
\begin{thebibliography}{HJKRR18}

\bibitem[AASY16]{AASY16}
Benny Applebaum, Sergei Artemenko, Ronen Shaltiel, and Guang Yang.
\newblock Incompressible functions, relative-error extractors, and the power of nondeterministic reductions.
\newblock {\em Computational complexity}, 25:349--418, 2016.

\bibitem[AS14]{AS14}
Sergei Artemenko and Ronen Shaltiel.
\newblock Lower bounds on the query complexity of non-uniform and adaptive reductions showing hardness amplification.
\newblock {\em Computational Complexity}, 23:43--83, 2014.

\bibitem[BCS20]{BCS20}
Mark Bun, Marco~Leandro Carmosino, and Jessica Sorrell.
\newblock Efficient, noise-tolerant, and private learning via boosting.
\newblock In {\em Proceedings of the 33rd Annual Conference on Learning Theory (COLT)}, pages 1031--1077, 2020.

\bibitem[BDB20]{BB20}
Shalev Ben-David and Eric Blais.
\newblock A tight composition theorem for the randomized query complexity of partial functions.
\newblock In {\em 2020 IEEE 61st Annual Symposium on Foundations of Computer Science (FOCS)}, pages 240--246. IEEE, 2020.

\bibitem[Ber46]{Ber46}
Sergei~Natanovich Bernstein.
\newblock {\em The Theory of Probabilities}.
\newblock Gostechizdat, Moscow, Leningrad, 1946.

\bibitem[BFJ{\etalchar{+}}94]{BFJKMR94}
Avrim Blum, Merrick Furst, Jeffrey Jackson, Michael Kearns, Yishay Mansour, and Steven Rudich.
\newblock Weakly learning dnf and characterizing statistical query learning using {Fourier} analysis.
\newblock In {\em Proceedings of the 26th Annual ACM Symposium on Theory of Computing (STOC)}, pages 253--262, 1994.

\bibitem[BHK09]{BHK09}
Boaz Barak, Moritz Hardt, and Satyen Kale.
\newblock The uniform hardcore lemma via approximate bregman projections.
\newblock In {\em Proceedings of the 20th Annual ACM-SIAM Symposium on Discrete Algorithms (SODA)}, pages 1193--1200, 2009.

\bibitem[CDV24]{CDV24}
Sílvia Casacuberta, Cynthia Dwork, and Salil Vadhan.
\newblock Complexity-theoretic implications of multicalibration.
\newblock In {\em Proceedings of the 55th Annual {ACM} Symposium on Theory of Computing (STOC)}, 2024.

\bibitem[Die00]{Die00}
Thomas~G Dietterich.
\newblock An experimental comparison of three methods for constructing ensembles of decision trees: Bagging, boosting, and randomization.
\newblock {\em Machine learning}, 40:139--157, 2000.

\bibitem[DIK{\etalchar{+}}21]{DIKLST21}
Ilias Diakonikolas, Russell Impagliazzo, Daniel~M Kane, Rex Lei, Jessica Sorrell, and Christos Tzamos.
\newblock Boosting in the presence of massart noise.
\newblock In {\em Proceedings of the 34th Annual Conference on Learning Theory (COLT)}, pages 1585--1644, 2021.

\bibitem[DRV10]{DGV10}
Cynthia Dwork, Guy~N Rothblum, and Salil Vadhan.
\newblock Boosting and differential privacy.
\newblock In {\em Proceedings of the 51st Annual Symposium on Foundations of Computer Science (FOCS)}, pages 51--60, 2010.

\bibitem[DW00]{DW00}
Carlos Domingo and Osamu Watanabe.
\newblock Madaboost: A modification of adaboost.
\newblock In {\em Proceedings of the 13th Annual Conference on Computational Learning Theory (COLT)}, pages 180--189, 2000.

\bibitem[Fel10]{Fel10}
Vitaly Feldman.
\newblock Distribution-specific agnostic boosting.
\newblock In Andrew~Chi{-}Chih Yao, editor, {\em Proceedings of the 1st Innovations in Computer Science}, pages 241--250, 2010.

\bibitem[Fre92]{Fre92}
Yoav Freund.
\newblock An improved boosting algorithm and its implications on learning complexity.
\newblock In {\em Proceedings of the 5th Annual Workshop on Computational Learning Theory}, pages 391--398, 1992.

\bibitem[Fre95]{Fre95}
Yoav Freund.
\newblock Boosting a weak learning algorithm by majority.
\newblock {\em Information and computation}, 121(2):256--285, 1995.

\bibitem[FS97]{FS97}
Yoav Freund and Robert~E Schapire.
\newblock A decision-theoretic generalization of on-line learning and an application to boosting.
\newblock {\em Journal of computer and system sciences}, 55(1):119--139, 1997.

\bibitem[Gav03]{Gav03}
Dmitry Gavinsky.
\newblock Optimally-smooth adaptive boosting and application to agnostic learning.
\newblock {\em Journal of Machine Learning Research}, 4(May):101--117, 2003.

\bibitem[GNW11]{GNW11}
Oded Goldreich, Noam Nisan, and Avi Wigderson.
\newblock On yao's xor-lemma.
\newblock {\em Studies in Complexity and Cryptography}, 6650:273--301, 2011.

\bibitem[GR08]{GR08}
Dan Gutfreund and Guy~N Rothblum.
\newblock The complexity of local list decoding.
\newblock In {\em International Workshop on Approximation Algorithms for Combinatorial Optimization}, pages 455--468. Springer, 2008.

\bibitem[GSV18]{GSV19}
Aryeh Grinberg, Ronen Shaltiel, and Emanuele Viola.
\newblock Indistinguishability by adaptive procedures with advice, and lower bounds on hardness amplification proofs.
\newblock In {\em Proceedings of the 59th Annual Symposium on Foundations of Computer Science (FOCS)}, pages 956--966, 2018.

\bibitem[HJKRR18]{HJKRR18}
Ursula H{\'e}bert-Johnson, Michael Kim, Omer Reingold, and Guy Rothblum.
\newblock Multicalibration: Calibration for the (computationally-identifiable) masses.
\newblock In {\em International Conference on Machine Learning}, pages 1939--1948. PMLR, 2018.

\bibitem[Hoe63]{hoeffding63:probability}
Wassily Hoeffding.
\newblock Probability inequalities for sums of bounded random variables.
\newblock {\em Journal of the American Statistical Association}, 58(301):13--30, 1963.

\bibitem[Hol05]{Hol05}
Thomas Holenstein.
\newblock Key agreement from weak bit agreement.
\newblock In {\em Proceedings of the 37th Annual ACM Symposium on Theory of Computing (STOC)}, pages 664--673, 2005.

\bibitem[IdW23]{IdW23}
Adam Izdebski and Ronald de~Wolf.
\newblock {Improved Quantum Boosting}.
\newblock In {\em Proceedings of the 31st Annual European Symposium on Algorithms (ESA 2023)}, volume 274, pages 64:1--64:16, 2023.

\bibitem[ILPS22]{ILPS22}
Russell Impagliazzo, Rex Lei, Toniann Pitassi, and Jessica Sorrell.
\newblock Reproducibility in learning.
\newblock In {\em Proceedings of the 54th Annual ACM SIGACT Symposium on Theory of Computing (STOC)}, pages 818--831, 2022.

\bibitem[Imp95]{Imp95}
Russell Impagliazzo.
\newblock Hard-core distributions for somewhat hard problems.
\newblock In {\em Proceedings of 36th Annual Foundations of Computer Science (FOCS)}, pages 538--545, 1995.

\bibitem[Jac97]{Jac97}
Jeffrey~C Jackson.
\newblock An efficient membership-query algorithm for learning dnf with respect to the uniform distribution.
\newblock {\em Journal of Computer and System Sciences}, 55(3):414--440, 1997.

\bibitem[KK09]{KK09}
Varun Kanade and Adam Kalai.
\newblock Potential-based agnostic boosting.
\newblock {\em Advances in Neural Information Processing Systems}, 22, 2009.

\bibitem[KS03]{KS03}
Adam~R Klivans and Rocco~A Servedio.
\newblock Boosting and hard-core set construction.
\newblock {\em Machine Learning}, 51:217--238, 2003.

\bibitem[KV89]{KV89}
Michael Kearns and Leslie Valiant.
\newblock Crytographic limitations on learning boolean formulae and finite automata.
\newblock In {\em Proceedings of the 21st Annual ACM Symposium on Theory of Computing (STOC)}, pages 433--444, 1989.

\bibitem[LR22]{LR22}
Kasper~Green Larsen and Martin Ritzert.
\newblock Optimal weak to strong learning.
\newblock {\em Advances in Neural Information Processing Systems (NeurIPS)}, 35:32830--32841, 2022.

\bibitem[LTW11]{LTW11}
Chi-Jen Lu, Shi-Chun Tsai, and Hsin-Lung Wu.
\newblock Complexity of hard-core set proofs.
\newblock {\em computational complexity}, 20:145--171, 2011.

\bibitem[Lup58]{Lup58}
O.~B. Lupanov.
\newblock A method of circuit synthesis.
\newblock {\em Izvesitya VUZ, Radiofiz}, 1:120--140, 1958.
\newblock (In Russian).

\bibitem[M{\etalchar{+}}89]{Mcd89}
Colin McDiarmid et~al.
\newblock On the method of bounded differences.
\newblock {\em Surveys in combinatorics}, 141(1):148--188, 1989.

\bibitem[RTTV08]{RTTV08}
Omer Reingold, Luca Trevisan, Madhur Tulsiani, and Salil Vadhan.
\newblock Dense subsets of pseudorandom sets.
\newblock In {\em Proceedings of the 49th Annual IEEE Symposium on Foundations of Computer Science (FOCS)}, pages 76--85, 2008.

\bibitem[Sch90]{Sha90}
Robert~E Schapire.
\newblock The strength of weak learnability.
\newblock {\em Machine learning}, 5:197--227, 1990.

\bibitem[Sch99]{Sch99}
Robert~E Schapire.
\newblock Theoretical views of boosting and applications.
\newblock In {\em Proceedings of the 10th International Conference on Algorithmic Learning Theory}, pages 13--25, 1999.

\bibitem[Ser03]{Ser03}
Rocco~A Servedio.
\newblock Smooth boosting and learning with malicious noise.
\newblock {\em The Journal of Machine Learning Research}, 4:633--648, 2003.

\bibitem[SF12]{FS12}
Robert~E. Schapire and Yoav Freund.
\newblock {\em Boosting: Foundations and Algorithms}.
\newblock The MIT Press, 2012.

\bibitem[Sha49]{Sha49}
Claude~E Shannon.
\newblock The synthesis of two-terminal switching circuits.
\newblock {\em The Bell System Technical Journal}, 28(1):59--98, 1949.

\bibitem[Sha04]{Sha04}
Ronen Shaltiel.
\newblock Towards proving strong direct product theorems.
\newblock {\em Computational Complexity}, 12(1/2):1--22, 2004.

\bibitem[Sha23]{Sha23}
Ronen Shaltiel.
\newblock Is it possible to improve yao's {XOR} lemma using reductions that exploit the efficiency of their oracle?
\newblock {\em Comput. Complex.}, 32(1):5, 2023.

\bibitem[SSBD14]{shalev14:understanding}
Shai Shalev-Shwartz and Shai Ben-David.
\newblock {\em Understanding Machine Learning - From Theory to Algorithms.}
\newblock Cambridge University Press, 2014.

\bibitem[SV10]{SV10}
Ronen Shaltiel and Emanuele Viola.
\newblock Hardness amplification proofs require majority.
\newblock {\em SIAM Journal on Computing}, 39(7):3122, 2010.

\bibitem[Tre07]{Tre07}
Luca Trevisan.
\newblock {The Impagliazzo Hard-Core-Set Theorem}.
\newblock \url{https://lucatrevisan.wordpress.com/2007/11/06/the-impagliazzo-hard-core-set-theorem/}, 2007.

\bibitem[Tre10]{Tre10}
Luca Trevisan.
\newblock {The Impagliazzo Hard-Core Lemma for the Mathematician}.
\newblock \url{https://lucatrevisan.wordpress.com/2010/03/12/the-impagliazzo-hard-core-lemma-for-the-mathematician/}, 2010.

\bibitem[TTV09]{TTV09}
Luca Trevisan, Madhur Tulsiani, and Salil Vadhan.
\newblock Regularity, boosting, and efficiently simulating every high-entropy distribution.
\newblock In {\em Proceedings of the 24th Annual IEEE Conference on Computational Complexity (CCC)}, pages 126--136, 2009.

\bibitem[Uhl74]{Uhl74}
Ditmar Uhlig.
\newblock On the synthesis of self-correcting schemes from functional elements with a small number of reliable elements.
\newblock {\em Matematicheskie Zametki}, 15(6):937--944, 1974.

\bibitem[Uhl92]{Uhl92}
Dietmar Uhlig.
\newblock Networks computing boolean functions for multiple input values.
\newblock In {\em Poceedings of the London Mathematical Society symposium on Boolean function complexity}, pages 165--173, 1992.

\bibitem[Ver18]{Ver18}
Roman Vershynin.
\newblock {\em High-Dimensional Probability: An Introduction with Applications in Data Science}.
\newblock Cambridge Series in Statistical and Probabilistic Mathematics. Cambridge University Press, 2018.

\bibitem[VZ12]{VZ12}
Salil Vadhan and Colin~Jia Zheng.
\newblock Characterizing pseudoentropy and simplifying pseudorandom generator constructions.
\newblock In {\em Proceedings of the forty-fourth annual ACM symposium on Theory of computing}, pages 817--836, 2012.

\bibitem[Yao82]{Yao82}
Andrew~C Yao.
\newblock Theory and application of trapdoor functions.
\newblock In {\em Proceedings of the 23rd Annual Symposium on Foundations of Computer Science (FOCS)}, pages 80--91. IEEE, 1982.

\end{thebibliography}
